%% file: 000-main.tex
\newif\ifsubmit 

\newif\ifeprint
\newif\ifconference
\newif\ifproceedings

\input{111-PARAMS}

\ifproceedings
\documentclass[runningheads]{llncs}
\fi
 
\ifconference
\documentclass{llncs}
\fi

\ifeprint
\documentclass[11pt]{article} 
\fi

\input{000-include/macros/config}
\input{000-include/macros/commoncmd}
\input{111-THIS-PROJECT/cmd}

\input{111-THIS-PROJECT/title-authors}

\begin{document}
\maketitle

\ifproceedings
\input{112-PROCEEDINGS}
\fi

\ifconference
\input{112-CONFERENCE}
\fi

\ifeprint
\input{112-EPRINT}

\fi

\bibliographystyle{alpha}
\end{document}

%% file: 111-PARAMS.tex
\submittrue
\eprinttrue
\conferencefalse
\proceedingsfalse

%% file: 000-include/macros/config.tex
\usepackage[utf8]{inputenc}

\ifproceedings
\usepackage[T1]{fontenc}
\fi

\ifeprint
\usepackage[pdftex,left=1in,top=1in,bottom=1in,right=1in]{geometry}
\usepackage{amsthm}
\fi

\usepackage{amsmath,amssymb,dsfont,mathrsfs}

\usepackage{braket}
\usepackage{physics}
\usepackage{mathtools}

\usepackage{xcolor}
\usepackage[framemethod=tikz]{mdframed}
\usepackage{graphicx}
\usepackage{caption}
\usepackage{bm}

\usepackage{verbatim} 
\usepackage{enumitem}

\usepackage[colorlinks=true,citecolor=blue,linkcolor=blue]{hyperref}
\usepackage[nameinlink, noabbrev, capitalize]{cleveref}

\hypersetup{
    colorlinks,
    linkcolor={red!50!black},
    citecolor={blue!50!black},
    urlcolor={blue!80!black}
}

\ifeprint
\usepackage{tocloft}

\fi

\definecolor{darkgreen}{rgb}{0,0.5,0}
\definecolor{darkblue}{rgb}{0,0,0.6}

\allowdisplaybreaks

\usepackage{tabularx}

%% file: 000-include/macros/commoncmd.tex
\newcommand{\addrefs}{\bibliography{111-THIS-PROJECT/refs}}

\DeclareMathAlphabet{\mathpzc}{OT1}{pzc}{m}{it}

\newcommand{\E}{\mathds{E}}

\newcommand{\zo}{\{0,1\}}
\newcommand{\poly}{\mathsf{poly}}

\newcommand{\ora}{\mathcal{O}}

\newcommand{\negl}{\mathsf{negl}}

\ifeprint
\newtheorem{theorem}{Theorem}

 \newtheorem{lemma}{Lemma}

 \newtheorem{definition}{Definition}
 
 \newtheorem{remark}{Remark}
 \newtheorem{fact}{Fact}
\fi

%% file: 111-THIS-PROJECT/cmd.tex
\newcommand{\QMA}{\mathsf{QMA}}
\newcommand{\QCMA}{\mathsf{QCMA}}
\newcommand{\QCIP}{\mathsf{QCIP}}
\newcommand{\TV}{\operatorname{TV}}
\newcommand{\YES}{\mathsf{YES}}
\newcommand{\NO}{\mathsf{NO}}
\newcommand{\ACCEPT}{\mathsf{ACCEPT}}
\newcommand{\REJECT}{\mathsf{REJECT}}
\newcommand{\bits}{\{0,1\}}

\newcommand{\protocoltitle}[1]{\noindent\underline{#1}\par\smallskip}
\newcommand{\protocolbody}{\par\smallskip\hrule height .4pt\relax\smallskip}

\newmdenv[
  nobreak=true,
  linewidth=.6pt,
  innerleftmargin=10pt,
  innerrightmargin=10pt,
  innertopmargin=8pt,
  innerbottommargin=8pt,
  skipabove=9pt,
  skipbelow=9pt
]{protocol}

%% file: 111-THIS-PROJECT/title-authors.tex
\title{Separating $\QMA$ from $\QCIP$\\
  with a Classical Oracle\\[0.5em]
  {\large\itshape or, the Power of Quantum Proofs over Classical Interaction\\
  for Quantum Verifiers}}

\ifproceedings
\fi

\date{}

\ifproceedings
\author{Alper Cakan\inst{1}\orcidID{0000-0003-3567-1704}}

\authorrunning{A. Cakan}

\institute{Carnegie Mellon University, Pittsburgh PA, USA \email{alpercakan98@gmail.com}}
\fi

\ifeprint
 \author{Alper Cakan\thanks{Carnegie Mellon University. \texttt{alpercakan98@gmail.com}.}}
\fi

%% file: 112-EPRINT.tex
\input{111-THIS-PROJECT/abstract}

\newpage
\setcounter{tocdepth}{2}
{\small\tableofcontents}

\newpage
\input{sections/introduction}
\subsection{Acknowledgments}
\input{111-THIS-PROJECT/acknowledgement}

\newpage
\input{sections/technical-overview}

\newpage
\input{sections/preliminaries}

\newpage
\part{Technical Tools}
\label{part:tools}
\input{sections/incompressibility}
\newpage
\input{sections/interactive-proof-tools}

\newpage
\part{Separating QMA from QCIP}
\label{part:separation}
\input{sections/oracle-problem}

\newpage
\input{sections/qma-upper-bound}

\newpage
\input{sections/qcip-lower-bound}

\newpage
\addrefs 

\newpage
\appendix
\part*{Appendix}
\addcontentsline{toc}{part}{Appendix}
\input{appendices/oracle-separations}

\newpage
\input{appendices/qcip-amplification}

\newpage
\input{appendices/interactive-proof-tools}

\newpage
\input{appendices/random-set-containment}

%% file: 111-THIS-PROJECT/abstract.tex
\begin{abstract}
  Whether some problems require quantum proofs has been a central
  question in quantum complexity (Aharonov and Naveh, 2002; Aaronson
  and Kuperberg, CCC 2007). Recently, breakthrough work
  of Bostanci, Haferkamp, Nirkhe, and Zhandry (STOC 2026), followed by
  a simpler separation due to Bostanci, Huang, and Vaikuntanathan
  (FOCS 2026), established a classical oracle separation between
  $\QMA$ and $\QCMA$. However, while they are not in $\QCMA$, the problems used in both separations
  still lie in $\mathsf{AM}$: they admit a two-message public-coin proof
  system with a \emph{classical} verifier.
  
  In this work, we ask whether some problems \emph{truly} require quantum
  proofs. More formally,
  we consider the complexity class $\QCIP$, introduced by Buhrman, Le Gall, and Weggemans
  (2024), where an efficient quantum verifier interacts with an unbounded prover over a \emph{classical channel} for an arbitrary polynomial number of rounds.

  We construct a classical oracle $\ora$ such that
  \[
    \QMA^\ora\not\subseteq\QCIP^\ora,
  \]
  thus showing that some languages indeed require quantum proofs,
  with no classical replacements. Since $\QCMA=\QCIP[1]$,
  this strengthens the earlier $\QMA$--$\QCMA$ separations, which
  now follow as a special case of our result.

  As a technical contribution, we extend to the complexity theory setting the techniques developed by Cakan, Goyal, and Shmueli (CRYPTO 2026) in the context of cryptography for analyzing classically interacting quantum machines. We believe this may be of independent interest.

  \ifproceedings
  \begingroup
  \raggedright
  \keywords{Quantum complexity \and quantum witnesses \and interactive proofs \and oracle separations \and QMA vs QCMA.}\par
  \endgroup
  \fi
\end{abstract}
\ifeprint
\vfill\par\textbf{Keywords:} Quantum complexity, quantum witnesses, interactive proofs, oracle separations, QMA vs QCMA.
\fi

%% file: sections/introduction.tex
\section{Introduction}
The question of whether quantum proofs are more powerful than classical ones, and if some languages require quantum proofs has been considered by many works. One of the first evidence that quantum proofs are indeed more powerful came with the work of \cite{AK07} who showed a unitary oracle $U$ such that $\QMA^U \not\subseteq \QCMA^U$. More recently, there have been many attempts on obtaining a classical oracle analogue of this result, culminating in the breakthrough work of Bostanci, Haferkamp, Nirkhe, and
Zhandry~\cite{BHNZ} who proved such a classical oracle separation for the first time, and later on the work of Bostanci, Huang, and Vaikuntanathan~\cite{BHV} who give a simpler classical oracle separation.

\paragraph{Classical replacements for quantum proofs?}
While \cite{BHNZ} and \cite{BHV} have shown the existence (relative to a classical oracle $\ora$) of a language $L$ such that $L \in \QMA^\ora$ but $L \not\in \QCMA^\ora$, for both works, the language they construct also can easily be shown to satisfy $L \in \mathsf{AM}^\ora$, meaning that it admits a single-round (i.e. two-message) public-coin proof system with a \emph{classical} verifier. For example, for the language $L$ considered in \cite{BHV}, the protocol is simple: the classical verifier sends a uniformly random challenge $r$, and the prover applies measurement (that depends on $r$) to their quantum proof and sends the measurement outcome. Then, the verifier makes their decision using a single classical oracle query. Thus, while the language $L$ is indeed not in $\QCMA$ (since it does not admit a classical witnes), in a sense, $L$ does not seem to \emph{truly} require quantum proofs.
Indeed, Bostanci, Huang, and Vaikuntanathan themselves highlight this distinction:
their problem ``does not need the full power of $\QMA$'' \cite{BHV}.

Thus, we in this work, we ask the following question.

\begin{quote}
  \centering\itshape
  Are there problems that truly require quantum proofs, with no classical replacements?
\end{quote}

Aside from understanding the power of quantum proofs, we note that this question has important to connections to other problems and areas, as we discuss later on.

\paragraph{Quantum proof systems with classical communcation}
To formalize this question, we consider the complexity class $\QCIP$, introduced by
Buhrman, Le Gall, and Weggemans~\cite{BLW}, which can be considered an extension of $\QCMA$ and $\mathsf{IP}$, in different axis. In a QCIP protocol, an efficient (i.e. polynomial-time) quantum verifier exchanges classical
messages of arbitrary polynomial length with a (unbounded) prover, over an arbitrary polynomial number of rounds, and the parties can retain private quantum memory between messages. Thus, $\QCIP$ is the classical-communication
analogue of $\mathsf{QIP}$ (or, quantum-verifier analogue of $\mathsf{IP}$). Note that $\QCMA$ is simply the one message case of $\QCIP$ (thus denoted $\QCIP[1]$), and we trivially have $\QCMA^\ora \subseteq \QCIP^\ora \subseteq \mathsf{QIP}^\ora$ and $\mathsf{IP}^\ora \subseteq \QCIP^\ora$ relative to any oracle $\ora$.

Using this formalization, we ask if there exist a language $L$ that admits quantum proofs, $L \in \QMA$, but no classical-message interactive proof system with a quantum verifier, $L \not\in \QCIP$.

\subsection{Our Results}

In this work, we answer the above question and show that some problems indeed truly require quantum proofs.

\begin{theorem}[Informal; see \cref{cor:oracle}]
  There exist a classical oracle $\ora:\bits^*\to\bits$ and a unary
  language $L$ such that
  \[
    L\in\QMA^\ora\setminus\QCIP^\ora.
  \]
\end{theorem}

Since $\QCMA^\ora=\QCIP^\ora[1]$ relative to any oracle $\ora$, a classical oracle
separation between $\QMA$ and $\QCMA$ follows as a special case of our result.

We also note that, in non-relativizing setting, we have $\QCIP \supseteq \QMA$ since $\QCIP \supseteq \mathsf{IP} = \mathsf{PSPACE}$ and $\QMA \subseteq \mathsf{PSPACE}$.

\paragraph{Technical contribution.} A recent work of Cakan, Goyal, Shmueli \cite{QFire} showed, in a cryptographic context, how to analyze classical-message interactions between two quantum machines; using the notion of \emph{incompressibility} and by showing how to simulate the interactive protocol between the machines by another protocol where one machine simply makes classical oracle queries. However, their technique is specialized to the case where both machines are efficient in the original protocol (making polynomially many quantum queries to the oracles), and relies on modifying both machines. In this work, also based on \emph{incompressibility}, we develop new techniques to analyze an inefficient prover and an efficient quantum-verifier interacting over a classical channel, by approximating the protocol (against any prover) by only simulating the verifier (with classical oracle queries). We believe this may have independent applications in quantum complexity theory.

\subsection{Organization}
In \cref{sec:overview}, we give a technical overview of the paper.
In \cref{sec:prelims}, we give the preliminaries, followed by technical
tools in \cref{sec:codes,sec:protocol-transformations}.
In \cref{sec:oracle}, we define the oracle distributions and state
the distributional separation. We prove the QMA upper bound in
\cref{sec:qma} and the QCIP lower bound in \cref{sec:hardness}.

%% file: 111-THIS-PROJECT/acknowledgement.tex
The author is extremely grateful to Kai-Min Chung, Wei-Hsiang Hung and Tzu-Yi Yang for collaborating on \cite{CCHY} and for extensive discussions on \cite{BHV} and \cite{QFire}.

\paragraph{AI Acknowledgment}
The results in the paper were humanly generated directly by the author. ChatGPT-6 (Astra) was used to generate the initial drafts of Preliminaries and also the parts that recall the lemmas/results from previous works and their routine generalizations (such as the ones in appendix), and also used to draft the initial LaTeX code of some proofs from the author's handwriting. Afterwards, ChatGPT-6 (Astra) was also used to help with choice of some of the parameters, and verifying the correctness of the manuscript. The author takes full responsibilty for the full paper.

\paragraph{Funding acknowledgment.}
The author was supported in part by a Carnegie Mellon University CyLab Presidential
Fellowship and in part by the following grants of Vipul Goyal: NSF award 1916939, the
DARPA SIEVE program, a gift from Ripple, a U.S. Department of Energy NETL award, a J.P.
Morgan Faculty Fellowship, a PNC Center for Financial Services Innovation Award, and a
Carnegie Mellon University CyLab seed funding award.

%% file: sections/technical-overview.tex
\section{Technical Overview}
\label{sec:overview}

In this section, we give a technical overview of our work.

\paragraph{Overview Organization} In \cref{overview:defns}, we recall the model, the complexity classes we consider and our main theorem statement. Then in \cref{overview:construction}, we give an overview of the intuition for our separation oracles and in \cref{overview:oracles} we give an overview of our oracles. Finally in \cref{overview:qma} and in \cref{overview:proof}, we give an overview of our $\QMA$ upper bound and $\QCIP$ lower bound proofs, respectively.

The formal definitions, theorem statements and proofs can be found in
the relevant sections.

\paragraph{Notation.}
We write $x\leftarrow\mu$ for sampling from a distribution $\mu$,
and $x\leftarrow S$ for uniform sampling from a finite set $S$.
We write $\langle(P^*)^\ora,V^\ora\rangle$ for an interaction with oracle
$\ora$, and $\ACCEPT$ for the verifier's accepting decision. We refer the reader to the preliminaries for more detailed notation guide.

\input{sections/overview/model}

\input{sections/overview/oracle-design}

\input{sections/overview/oracles}

\input{sections/overview/qma-upper-bound}

\input{sections/overview/qcip-lower-bound}

%% file: sections/overview/model.tex
\subsection{Model and Our Main Theorem}
\label{overview:defns}

In this section, we recall the relevant complexity classes and the oracle setting we consider, and our main theorem statement. In particular, similar to some of the previous work \cite{Zha25,BHNZ,BHV}; our main theorem, instead of directly obtaining an oracle separation, will consider an \emph{oracle-input problem} with a distribution over the instances, which can be converted to an oracle separation through routine techniques (which we provide in \cref{app:diagonal} for completeness).

\paragraph{Proof systems.}
A proof system $(P,V)$ consists of an efficient quantum verifier $V$,
running in time $\poly(|x|)$, and a possibly unbounded prover $P$. When we consider the interaction of $V$ and some prover $P^*$ (not necessarily the honest one $P$), they receive a common input $x\in\bits^*$, and exchange messages, after
which $V$ outputs $\ACCEPT$ or $\REJECT$. We denote the verifier's (probabilistic)
decision by $\langle P^*,V\rangle(x)$.
\begin{itemize}
  \item For {QCIP} systems~\cite{BLW}, all messages are classical.
    The parties exchange polynomially many bits over an arbitrary
    polynomial number of rounds, and both may retain private quantum
    memory.\footnote{Since the prover are allowed to be unbounded and communication is classical; for both completeness and soundness it suffices to
    consider unbounded classical randomized provers with no private memory between rounds, that simply receive the
    message and full transcript so far and oracle access if applicable (see \cref{app:classical-provers}).}
    The case of one classical message from $P^*$ to $V$ is
    $\QCMA=\QCIP[1]$.
  \item For {QMA} systems~\cite{AK07,MW05}, $P^*$ sends a single
    $\poly(|x|)$-qubit message (\emph{witness}) to $V$, who then decides without further
    interaction. Thus, we can also write the proof system $(P,V)$ as $(\{\ket{\psi_x}\}_{x \in \zo^*}, V)$ and any prover $P^*$ as $\{\ket{\phi_x}\}_{x \in \zo^*}$.
\end{itemize}

For a promise problem $L = (L_{\YES}, L_{\NO})$ or a language $L \subseteq \zo^*$ (we write $L_{\YES} = L$ and $ L_{\NO} = \zo^* \setminus L_{\YES}$), we say that $(P,V)$ is a proof system for $L$ if it satisfies completeness and soundness:
\[
  \begin{aligned}
    x\in L_{\YES} \implies &\quad
      \Pr[\langle P,V\rangle(x)=\ACCEPT]\ge2/3,\\
     x\in L_{\NO}\implies &\quad \forall P^*~
      \Pr[\langle P^*,V\rangle(x)=\ACCEPT]\le1/3.
  \end{aligned}
\]

\paragraph{Oracle-input problems.}
In an \emph{oracle-input problem} (\cite{BBCMW01,Zha25}), the algorithms access the input as a quantumly-queriable oracle $\ora$ (a finite Boolean
table $\ora:\mathcal Z_\lambda\to\bits$), instead of receiving an input $x$ directly. We write $\lambda$ to denote the sequence parameter (i.e. the \emph{input length}) and in this case, resource bounds (such as the number of queries to this oracle or lengths of messages in protocols) are measured
in $\lambda$, rather than in the table's actual length.
For example, one usually takes $\mathcal Z_\lambda=\bits^{p(\lambda)}$
for a polynomial $p$, in which case the table has $2^{p(\lambda)}$ entries.

An oracle-input language specifies, for each $\lambda$, a set
\[
  L_\lambda\subseteq\{\ora:\mathcal Z_\lambda\to\bits\}
\]
of YES tables; all remaining tables are NO inputs. A promise problem
instead specifies disjoint sets $L_{\YES,\lambda}$ and
$L_{\NO,\lambda}$, with no requirement on tables outside their union.

The complexity classes in this model are defined as expected: For example, in QMA, the verifier receives a $\poly(\lambda)$-qubit
witness, which may depend on $\ora$: completeness requires a suitable
witness for each YES table, and soundness must hold for every witness
on each NO table. Similarly, in QCIP both parties have query access to
$\ora$ and interact on common input $1^\lambda$: completeness holds
for $P$ on each YES table, and soundness holds against every $P^*$ on
each NO table. We give the formal
definitions in \cref{def:oracle-input}.

\paragraph{Our main theorem.} In our main theorem (\Cref{thm:main}), which will be the focus of our technical overview, we construct two distributions $\YES_\lambda,\NO_\lambda$ over
oracle-inputs (with overlap in their supports) and
establish the following acceptance guarantees.
\begin{enumerate}
  \item \emph{QMA upper bound.}
    There is a polynomial-time quantum verifier $V$ and $a(\lambda) = \poly(\lambda)$ qubit witnesses $\{\ket{\psi_\ora}_{\ora}\}$ such that, for all sufficiently large $\lambda$,
    \[
      \begin{aligned}
        &\Pr_{\ora\leftarrow\YES_\lambda}
        \left[
          \Pr[V^\ora(1^\lambda,\ket{\psi_{\ora}})=\ACCEPT]\ge1-2^{-\lambda}\right]
          \ge1-2^{-\lambda},\\
        &\Pr_{\ora\leftarrow\NO_\lambda}
        \left[\forall\ket\psi:\
          \Pr[V^\ora(1^\lambda,\ket\psi)=\ACCEPT]\le2^{-\Omega(\lambda)}\right]
          \ge1-2^{-\Omega(\lambda)}.
      \end{aligned}
    \]

  \item \emph{QCIP lower bound.}
    For every QCIP protocol $(P',V')$, all sufficiently large $\lambda$,
    and every $0\le\varepsilon<1/3$, if
    \[
      \Pr_{\ora\leftarrow\YES_\lambda}
      \left[\Pr[\langle(P')^\ora,(V')^\ora\rangle(1^\lambda)=\ACCEPT]\ge2/3\right]
      \ge1-\varepsilon,
    \]
    then
    \[
      \Pr_{\ora\leftarrow\NO_\lambda}
      \left[\exists P^*:\
        \Pr[\langle(P^*)^\ora,(V')^\ora\rangle(1^\lambda)=\ACCEPT]>1/3\right]
      \geq 1/3-\varepsilon.
    \]
    Thus, if $(P',V')$ accepts with probability at least $2/3$ on
    a $1-\varepsilon$ fraction of oracles under $\YES_\lambda$, then,
    on at least a $1/3-\varepsilon$ fraction under $\NO_\lambda$,
    some $P^*$ makes $V'$ accept with probability greater than $1/3$.
\end{enumerate}

Using an elementary restriction argument and standard oracle diagonalization techniques, we can first obtain an oracle-input promise language in $\QMA$ but not in $\QCIP$, and then finally a standard oracle separation:
\begin{quote}
There exist a classical oracle $\ora:\bits^*\to\bits$ and a unary
language
$
  L^\ora\in\QMA^\ora\setminus\QCIP^\ora.
$
\end{quote}
For completeness, we give both transformations in \cref{app:diagonal}.

%% file: sections/overview/oracle-design.tex
\subsection{Oracle Intuition}
\label{overview:construction}

In this section, we give an overview of the intuition for our oracle distributions $\YES_\lambda, \NO_\lambda$.

Our high level starting point is the well-known result of Fortnow and
Sipser~\cite{FS88}, giving an oracle separation between $\mathsf{coNP}$ and $\mathsf{IP}$. In their separation, they consider the simple language
\[
  L^{\ora}=\{1^\lambda:\ora(z)=1\text{ for every }z\in\bits^\lambda\}.
\]
Observe that $L^{\ora}\in\mathsf{coNP}^{\ora}$: a single input $z$ such that 
$\ora(z) = 0$ certifies non-membership for $L$. However, to decide or even just to verify membership in $L$, a verifier must check that the every entry of
the exponentially large table $\ora$ is marked (which is not efficient), or a prover must somehow convince the verifier that this is the case. However, a cheating prover can always delibaretly choose values $z$ that are marked (i.e. $\ora(z) = 1$). Of course, formally proving that $L \not\in \mathsf{IP}$ takes some more effort than this. A little bit more formally; fix any prover strategy that succeeds on the all-one
table. If the verifier makes $T$ classical queries, then there is an entry
that it checks with probability at most $T/2^\lambda$. However, changing that entry to zero gives a NO table, yet following the same prover strategy \emph{decreases} (i.e. additively) acceptance probability by at most $T/2^\lambda$.

Unfortunately, it is easy to see that this language is also hard with a quantum witness (i.e. $L \not\in \QMA$), nevertheless it provides us with a useful and simple intuition for obtaining a language outside $\mathsf{IP}$: Our oracle should allow a cheating prover to have a large number of marked options to choose from at all times.

\paragraph{Incompressibility and equivocation.}
Now, to actually obtain a language that is in $\QMA$ while utilizing the above idea to preserve the $\mathsf{IP}$ lower bound (and eventually obtain a $\mathsf{QCIP}$ lower bound), we use the following idea that have also been utilized in many related works~\cite{QFire,CCHY,NN24,LLPY24,BDK24,Zha25,GGJL25,BHNZ,BHV}:
\begin{itemize}
  \item \emph{Classical Incompressibility and Quantum Equivocation.} We will use a search problem that satisfies an incompressibility property (\cite{QFire}), that is, no classical string can contain \emph{too many} solutions, whereas a quantum state can be \emph{equivocated} to solve any instance of the search problem. 
\end{itemize}
Then, as our oracles, we will have either an oracle that contains instance of this problem, or an oracle that essentially randomly includes each point, with a distribution that requires the verifier to ask for solutions on many instances. Then, through incompressibility, we will conclude that a classical message will not let the verifier obtain many solutions, thus making it either not complete, or not sound. Note that, as we discuss later on, classical interaction is significantly different from the incompressibility guarantee satisfied for a single classical string, thus, the $\mathsf{QCIP}$ lower bound proof will require significantly more work beyond the above intuition: Indeed, we handle this issue by extending the techniques of Cakan, Goyal and Shmueli \cite{QFire} for analyzing quantum machines that interact classically. Interestingly, this step of our proof, and similarly the analysis of \cite{QFire}, also relies on incompressibility.

%% file: sections/overview/oracles.tex
\subsection{The Oracles}
\label{overview:oracles}
In this section, we give an overview of our oracle distributions $\YES_\lambda, \NO_\lambda$.

For the search problem, we use the \emph{code-preimage problem} (or \emph{code-hashing problem}) of Yamakawa and
Zhandry~\cite{YZ24}, also utilized in the recent sequence of attemps \cite{LLPY24,BDK24} to separate QMA--QCMA with a classical oracle, and finally in the work of \cite{BHV} to who proved the separation using this search problem. In fact, we use the \emph{biased} version of the problem, introduced by \cite{GGJL25} and also used by \cite{BHV}.\footnote{The bias in~\cite[Sections 4.2--4.3]{GGJL25} allows stronger list recovery while preserving quantum decoding. Our structured marking below uses the same type of biased coordinate constraints.}

In more detail, let $C\subseteq\Sigma^q$ be a (public) linear code, let $H_i:\Sigma\to\bits$ (which we refer to as a hash function), and write $H=(H_1,\ldots,H_q)$ and $H(v)=(H_1(v_1),\ldots,H_q(v_q))$. In the Yamakawa--Zhandry code-preimage search problem, one is given an \emph{image} $y$, and the goal is to find a codeword $v\in C$ such that $H(v)=y$. As shown by~\cite{YZ24}, suitable error-correction properties of the dual code $C^\perp$ allow an efficient quantum algorithm to find such a solution for any fixed target $y\in\zo^q$ with overwhelming probability.\footnote{In the original Yamakawa--Zhandry setting, the algorithm has quantum oracle access to $H$, and this probability is over a uniformly random $H$ and the algorithm's measurements. We use the biased-hash decoding algorithm version proven in \cite{BHV}.}

Now we present our oracle distributions $\YES_\lambda$, $\NO_\lambda$ over oracles $\ora: \mathcal Z_\lambda \to \zo $, with parameters $p,\delta\in(0,1)$ which are specified later and $Z_\lambda=\bits^q\times C$.

\begin{protocol}
  \noindent\textbf{Oracle distributions $\YES_\lambda$ and $\NO_\lambda$.}
  \begin{itemize}
    \item \textbf{Code-preimage structure ($\YES_\lambda$):} Sample all entries of each $H_i$ independently from $\mathrm{Bernoulli}(p)$ and set
      \[
        \ora(y,v)=
        \begin{cases}
          1,&H(v)=y,\ y_{\lambda+1}=\cdots=y_q=0,\\
          0,&\text{otherwise}.
        \end{cases}
      \]
    \item \textbf{Independent marking ($\NO_\lambda$):}
      For every $(y,v)\in\mathcal Z_\lambda$, sample $\ora(y,v)$ independently
      from $\mathrm{Bernoulli}(\delta)$.
  \end{itemize}
\end{protocol}
For $(y,v)\in\mathcal Z_\lambda$, we call the pair a
  \emph{code-preimage pair} if $H(v)=y$, and \emph{marked} if $\ora(y,v)=1$. Note that the verifier has query access only to $\ora$, with no separate access to $H$.

\paragraph{Choosing the Parameters and the Witness}
To show that our language is in QMA, we use the same quantum witness as in~\cite{Liu23,BHV}, based on Liu's observation~\cite{Liu23} that a single quantum state (that depends on the hashes $H$) can be used to obtain code preimages without further access to $H$. We denote this QPT decoding algorithm by $\mathsf{CodePreimageDecode}$.
To show our QCIP lower bound, in the proof we choose our parameters so that, as explained before, a cheating prover will have sufficiently many choices marked by $\ora$ under $\NO_\lambda$ and will be able to fool any verifier. In particular, to satisfy both the QMA upper bound and the QCIP lower bound, we will choose the parameters $p,\delta$ as $p=\lambda^{-4}$ and
$\delta=2^{-\lceil\lambda/12\rceil}$ and we choose the univariate multiplicity code $C$ from \cite{BHV} over the alphabet $\Sigma=\mathbb{F}_q^\lambda$, where $q$ is prime and $\lambda^5<q\le 2\lambda^5$.
These choices remain fixed and implicit throughout the rest of the overview.

%% file: sections/overview/qma-upper-bound.tex
\subsection{The QMA Proof System}
\label{overview:qma}
In this section, we give an overview of the QMA upper bound part of our main theorem. 

\paragraph{The verifier and witnesses.}
For $\ora$ sampled under $\YES_\lambda$, we use the witness
\[
  \ket{\Psi_H}\propto
  \bigotimes_{i=1}^{q}\left[
    \left(\sum_{z:H_i(z)=0}\ket z\right)
    \otimes\left(\sum_{z:H_i(z)=1}\ket z\right)\right].
\]
Given $\ket\phi$, $V^\ora$ samples $u\leftarrow\bits^\lambda$
uniformly, sets $y=u\|0^{q-\lambda}$, and runs
\[
  v\leftarrow\mathsf{CodePreimageDecode}(\ket\phi,y).
\]
It rejects if $v=\bot$, otherwise, it queries $\ora(y,v)$ and accepts
if and only if the answer is $1$.

For our code $C$ and parameters $p,q$, it follows from~\cite{BHV} that with probability at least $1-2^{-\lambda}$ over the sampling under $\YES_\lambda$, for every $u\in\bits^\lambda$,
\[
  \Pr_{v\leftarrow\mathsf{CodePreimageDecode}(\ket{\Psi_H},u\|0^{q-\lambda})}
    [v\in C\text{ and }H(v)=u\|0^{q-\lambda}]
  \ge 1-2^{-\lambda}.
\]

Thus, the verifier accepts with probability at least $1 - 2^{-\lambda}$ on at least a $1 - 2^{-\lambda}$ fraction of oracles under $\YES_\lambda$.

\paragraph{Acceptance under the NO distribution.}
To bound acceptance for arbitrary witnesses, observe that $V$ samples $y$ after receiving the
$a$-qubit witness $\ket\phi$. Thus, unlike an interactive classical-message proof system, the prover essentially \emph{commits} to the quantum state it sends and cannot adaptively choose marked pairs to cheat.

In more detail, sample $\ora\leftarrow\NO_\lambda$.
The functions $\ora(u\|0^{q-\lambda},\cdot)$ are independent across
$u\in\bits^\lambda$. Hence, even though the witness depends on the
whole oracle $\ora$, its mutual information with the function
$\ora(u\|0^{q-\lambda},\cdot)$ for a uniformly chosen $u$ is at most
$a/2^\lambda$ on average.
However, if we consider a modified execution of the verifier $V$ where we
keep the original witness and decoded pair $(y,v)$ unchanged but replace
$\ora(y,\cdot)$ by an independent copy when evaluating $\ora(y,v)$,
it accepts with probability at most $\delta$.
By Pinsker's inequality, the actual execution's acceptance probability is
at most that of this modified execution plus $\sqrt{a/2^\lambda}$ on average.
Applying this to a maximizing witness
for each oracle gives
\[
  \E_{\ora\leftarrow\NO_\lambda}\max_{\ket\phi}
    \Pr[V^{\ora}(\ket\phi)=\ACCEPT]
    \le\delta+\sqrt{a/2^\lambda}.
\]
With $\gamma_{\NO}=(\delta+\sqrt{a/2^\lambda})^{1/2}$, Markov's
inequality shows that, on all but a $\gamma_{\NO}$ fraction of
oracles under $\NO_\lambda$, every witness is accepted with probability at most
$\gamma_{\NO}$ (\cref{qma:noprob}).

%% file: sections/overview/qcip-lower-bound.tex
\subsection{QCIP Lower Bound}
\label{overview:proof}
In this section, we give an overview of our QCIP lower bound. 

Suppose a QCIP protocol $(P_0,V_0)$ accepts
with probability at least $2/3$ on at least a $1-\varepsilon$ fraction
of oracles under $\YES_\lambda$. We show that, on at least a
$1/3-\varepsilon$ fraction of oracles under $\NO_\lambda$, some
cheating prover $P^*$ makes $V_0$ accept with probability greater than $1/3$
(for $0\le\varepsilon<1/3$ and all sufficiently large $\lambda$).

Our proof proceeds in three steps. For the rest of the overview, fix
$(P_0,V_0)$ satisfying the stated acceptance bound under $\YES_\lambda$,
and let $r(\lambda)$ be the number of rounds, $N(\lambda)$ be the number
of queries made by $V_0$ in each round, and $c(\lambda)$ be the total
transcript length.

\begin{enumerate}
  \item \emph{Classical simulation:} Using code-preimage incompressibility under $\YES_\lambda$ and independent-marking incompressibility under $\NO_\lambda$, we show that the quantum verifier $V_0$ can be simulated by a \emph{classical} verifier $V_1$ that makes polynomially many classical queries\footnote{Our simulator will take exponential time outside oracle calls, which is not an issue for our overall proof.} while preserving (up to an additive error of $1/100$) the acceptance probability against every prover $P^*$ on all but a negligible fraction of oracle inputs under each of the YES and NO distributions.
  \item \emph{Public coins and perfect completeness:} We next apply the
    transformations of Goldwasser--Sipser~\cite{GS} and
    F\"urer et al.~\cite{FGMSZ} to $(P_0,V_1)$, obtaining a public-coin protocol
    $(P_2,V_2)$ with deterministic $P_2$ such that $V_2$ makes all oracle queries after the interaction,\footnote{Although
    our case differs slightly from the original formulations, both
    transformations still apply; see \cref{app:combined-transformations}
    for details.} and satisfies the following guarantees:
    \begin{itemize}
      \item $(P_2,V_2)$ accepts with probability $1$ on at least a
        $1-\varepsilon-\negl(\lambda)$ fraction of oracles under
        $\YES_\lambda$.
      \item Except for a negligible fraction of oracles under
        $\NO_\lambda$, if some prover makes $V_2$ accept with probability
        greater than $1/3$, then some prover makes $V_0$ accept with
        probability greater than $1/3$ on the same oracle.
    \end{itemize}
  \item \emph{A cheating prover:} Let $N_2(\lambda)$ be the number
    of queries made by $V_2$. Again using incompressibility, we show that if $(P_2,V_2)$ accepts with
    probability $1$ on a set of probability at least $\delta^2$ under
    $\YES_\lambda$, then we can construct a cheating prover $\widetilde P$ (from $P_2$) such that
    \[
      \E_{\ora\leftarrow\NO_\lambda}
        \Pr[\langle\widetilde P^\ora,V_2^\ora\rangle=\ACCEPT]
        \ge1-(N_2+1)\delta.
    \]
\end{enumerate}

Combining these steps, we conclude that on at least a $1/3-\varepsilon$ fraction of
oracles under $\NO_\lambda$, some prover makes $V_0$ accept with
probability greater than $1/3$, as claimed.

\input{sections/overview/classical-simulation}

\input{sections/overview/cheating-prover}

%% file: sections/overview/classical-simulation.tex
\subsubsection{Classical Simulation}
\label{overview:simulation}
In this section, we given overview of our simulation lemma, which constructs a verifier $V_1$ that simulates $V_0$
using polynomially many classical queries. For either oracle
distribution, except for a negligible fraction of oracles $\ora$,
the simulation preserves the acceptance probability against every
prover $P^*$ up to an additive error of $1/100$:
\[
  \left|
    \Pr[\langle (P^*)^\ora,V_0^\ora\rangle=\ACCEPT]
    -\Pr[\langle (P^*)^\ora,V_1^\ora\rangle=\ACCEPT]
  \right|<1/100.
\]

Our simulator is inspired by a similar result of Cakan, Goyal, and
Shmueli~\cite{QFire}. Their method also relied on
incompressibility properties, however, (i) their result crucially relied on
both parties making polynomially many queries, whereas in our case $P^*$
will be arbitrary/inefficient and (ii) their result relied on also tailoring the \emph{sender} (corresponding to our prover) for the simulator, whereas our simulation holds for all provers $P^*$ simultaneously.

\paragraph{The simulating verifier, high level.} 
For clarity, in this overview we describe our simulator as a quantum algorithm with classical oracle queries; however, it can be converted to a (time-inefficient, query-efficient) fully classical algorithm in a straightforward way.

Our simulator essentially learns each query of $V_0$ one by one and creates an oracle-estimation database for these queries. To learn these queries, the simulator repeatedly creates query states of $V_0$ many times and measures the input register. Note that after the first $j$ queries are learned, the learning process for query $j+1$ uses the estimated databases to simulate the first $j$ queries, thus creating many copies of the $j+1$-th query state does not cost any oracle calls. Thanks to the incompressibility property, we know that classical messages will not contain too many solutions to our search problem, which allows us to bound the number of queries we need while creating the estimation databases.

We also rely on the fact that the messages $V_0$ receives are classical, which allows us to both invoke our incompressibility property (which only applies to classical strings), and also run $V_0$ many times (which would be impossible with quantum messages, even if we had some oracle estimation database). Finally, while $V_0$ might create some possibly unclonable internal state during interaction, we essentially show how to classicalize its description using our oracle-estimation databases.

\paragraph{The simulating verifier, more details.}  The simulator $V_1$ stores the transcript $\pi$ and databases $D_{i,j}$ of
pairs satisfying $\ora(y,v)=1$. Let $0 < \theta(\lambda) < 1$ be some inverse-polynomial parameter and let $t(\lambda), K_0(\lambda), t(\lambda)$ be some parameters satisfying $t(\lambda)=c(\lambda)+\lambda\leq O(\poly(\lambda))$,$\frac{2\lambda^8}{\theta(\lambda)}\leq K_0(\lambda)
\leq O(\poly(\lambda))$ (see \cref{sec:simulation} for the parameter choices).

\begin{protocol}
  \protocoltitle{$V_1^\ora(1^\lambda)$}
  Initialize $\pi=\varnothing$ and $D_{0,N}=\varnothing$.
  \protocolbody

  For each round $i$:
  \begin{enumerate}
    \item Receive the prover's message $m_i$ and set
      $D_{i,0}\gets D_{i-1,N}$.
    \item From $j=1$ to $j = N$; initialize $D_{i,j}\gets D_{i,j-1}$ and repeat the following trial until $K_0$ consecutive trials add nothing to $D_{i,j}$ and then move to next $j$:
      \begin{enumerate}[label*=\arabic*.,leftmargin=*]
        \item Prepare
          $\ket{\psi_i}\leftarrow\mathsf{PrepState}(\pi,(D_{s,k})_{s<i,\,k\in[N]})$.
          \item Simulate $V_0$ on $\ket{\psi_i}$ and $m_i$, stopping
          immediately before its $j$-th oracle query.
          During the simulation, the $k$-th ($k < j$) oracle query is answered using $D_{i,k}$ with the mapping $(y,v) \mapsto 1$ if $(y,v)\in D_{i,k}$ and $0$ otherwise. In particular, no queries to $\ora$ is made.
          Let $\ket{\psi_{i,j}}$ be the resulting state on query-input
          register $\mathsf X$, answer register $\mathsf A$, and
          workspace $\mathsf W$.
        \item Measure the query-input register $\mathsf X$ to obtain
          $z=(y,v)$. If $z\notin D_{i,j}$, query $b\gets\ora(z)$ and,
          if $b=1$, set $D_{i,j}\gets D_{i,j}\cup\{z\}$.
          \textbf{\textcolor{blue}{If $|D_{i,j}|=t$, output $\REJECT$ and terminate.}}
      \end{enumerate}
    \item Prepare $\ket{\psi_i}\leftarrow\mathsf{PrepState}(\pi,(D_{s,k})_{s<i,\,k\in[N]})$
      again and run $V_0$ from $\ket{\psi_i}$ on message $m_i$, using the saved databases $D_{i,j}$ (as above). 
      Measure its outgoing-message register
      $\mathsf M$ to obtain $m'_i$ and set
      $\pi\gets\pi\|(m_i,m'_i)$. Send $m'_i$, or output it as
      the decision in the final round.
  \end{enumerate}
\end{protocol}

The procedure $\mathsf{PrepState}(\pi,(D_{s,k})_{s<i,\,k\in[N]})$ prepares
$\ket{\psi_i}$, the simulated private state at the start of round $i$ conditioned on the transcript $\pi$: It essentially does a rejection-sampling of executions of $V_0$ until the same transcript $\pi$ is obtained, and crucially, it uses the databases $D_{s,k}$ to answer the queries of $V_0$ (as described above), so it does not make any real oracle calls.

\paragraph{Bounding the simulation error.}
We first analyze one query, $j$, in a fixed round $i$. We can then replace the
$N$ queries in that round by their simulation one at a time, and finally sum the errors
over the $r$ rounds.

Write the query-input state from step 2.2 (state for query $j$ in round $i$) as
\[
  \ket{\psi_{i,j}}=\sum_{z,o,w}\alpha_{z,o,w}
    \ket z_{\mathsf X}\ket o_{\mathsf A}\ket w_{\mathsf W},
  \qquad
  p_{i,j}(D_{i,j})=\sum_{\substack{z\notin D_{i,j}\\\ora(z)=1}}
    \sum_{o,w}|\alpha_{z,o,w}|^2,
\]
Note that the quantity $p_{i,j}(D_{i,j})$ is the total weight
on $1$-inputs (i.e. inputs such that $\ora(y,v)=1$) that are outside the database $D_{i,j}$.
Consider two cases:
\begin{enumerate}
  \item \emph{Case 1, $p_{i,j}(D_{i,j})>\theta$}: The total weight on unrecorded
    $1$-inputs is still high. Each trial has an independent probability greater than $\theta$ of adding to the database, thus, $K_0$ consecutive trials add nothing
    with probability at most $(1-\theta)^{K_0}$.
  \item \emph{Case 2, $p_{i,j}(D_{i,j})\le\theta$}: The total weight on unrecorded
    $1$-inputs is at most $\theta$. Therefore, answering zero outside $D_{i,j}$ versus using the actual oracle $\ora$ changes the post-query state by at most
    $2\sqrt\theta$ in distance.
\end{enumerate}
We first analyze the simulator without the cutoff rule written in blue font. First, note that reaching Case 2 means that the database-based simulation of a query is sufficiently close. Then, note that the inner loop does enter Case 2 eventually, except with probability at most
$(2^{\lambda^8+1}+1)(1-\theta)^{K_0} \leq 2^{-\lambda^8}$ \footnote{Since there are at most $2^{\lambda^8}$ possible pairs $(y,v)$.}, since the loop stops only after $K_0$
consecutive trials add nothing new.\footnote{The simulator cannot know
$p_{i,j}(D_{i,j})$ and which case we are in, so we use these consecutive trials without an
addition as its stopping rule.}

A hybrid argument then bounds the change in acceptance probability over
the all $r$ rounds (each of which with $N$ queries) by a total of
$
  rN\left(2\sqrt\theta+2^{-\lambda^8}\right),
$
which is at most $1/200$ for sufficiently large $\lambda$, by our parameter choices. Thus, the simulator without the cutoff is indeed a good approximator to $V_0$.

\paragraph{Adding the cutoff.}
Unfortunately, the uncapped simulator may make exponentially many queries.
We now analyze the actual simulator who stops when a database reaches size $t$, and thus makes at most 
$
  M(\lambda)=(rN+t)K_0\leq O(\poly(\lambda))
$
classical queries. We will show that, except for a
$2^{-\Omega(\lambda^2)}$ fraction of oracles under either distribution,
the probability of ever reaching the cutoff during the entire
interaction is at most $1/200$, for every prover $P^*$.

To prove the cutoff bound, we use two incompressibility properties
(\cref{thm:incompressibility,thm:independent-incompressibility}).
An \emph{encoder}, given the entire sampled oracle $\ora$, sends an
$s$-bit classical message $m$ to an oracle-free \emph{decoder}.
They may share randomness sampled independently of $\ora$, and the
encoder may choose $m$ after seeing it. The decoder tries to recover
$t$ distinct pairs $(y,v)$ satisfying $\ora(y,v)=1$.
\begin{itemize}
  \item Under $\YES_\lambda$, these pairs must solve the code-preimage
    problem. Code-preimage incompressibility bounds success by
    $2^s\delta^{4t}$ for $t\le2^\lambda$.
  \item Under $\NO_\lambda$, the oracle bits $\ora(y,v)$ for $v\in C$
    are independent.
    Independent-marking incompressibility bounds success by
    $2^s\delta^t$.
\end{itemize}

Now suppose an execution reaches the cutoff. Then consider the following encoder-decoder for the incompressibility property. We can encode its
$t$ database entries using only the transcript and the indices of
the trials that added them:
\[
  m=(\pi,I),\qquad
  s=c+t\lceil\log_2 M\rceil.
\]

To decode $m$, replay the simulator
using the messages in $\pi$ and the same estimation randomness\footnote{While we described our simulator as making quantum measurements, as mentioned before, this can be simulated classically.}.
For a tested pair outside the current database, answer one exactly
when the trial index belongs to $I$, and zero otherwise.
At the $t$-th addition, output the database. This recovers its
$t$ pairs satisfying $\ora(y,v)=1$ without any oracle queries.

For a fixed oracle $\ora$, let $g_\ora$ be the probability, over the
shared estimation randomness, that some $s$-bit message decodes to
$t$ pairs satisfying $\ora(y,v)=1$. Every execution reaching the cutoff
provides such a message, so $g_\ora$ bounds the cutoff probability
for every prover $P^*$. Applying incompressibility to this encoding
and then Markov's inequality gives
\[
  \E_{\ora\leftarrow\mathcal P}[g_\ora]\le2^s\delta^t,
  \qquad
  \Pr_{\ora\leftarrow\mathcal P}[g_\ora>1/200]
    \le200\cdot2^s\delta^t=2^{-\Omega(\lambda^2)},
\]
where $\mathcal P\in\{\YES_\lambda,\NO_\lambda\}$.
Consequently, outside this exceptional fraction of oracles, the cutoff
adds at most $1/200$ to the simulation error for every prover.
We provide the full proof in \cref{sec:simulation}.

%% file: sections/overview/cheating-prover.tex
\subsubsection{A Cheating Prover}
\label{overview:cheater}

In this section, we give an overview of the cheating prover
$\widetilde P$ we need for the third step of the proof.
We start with a public-coin protocol $(P_2, V_2)$ such that $P_2$ is deterministic and $V_2$ makes all of its $N_2$ queries after the interaction. Assuming 
\[
  \Pr_{\ora\leftarrow\YES_\lambda}
    [\Pr[\langle P_2^\ora,V_2^\ora\rangle=\ACCEPT]=1]\ge\delta^2,
\]
we construct $\widetilde P$ such that
\[
  \E_{\ora\leftarrow\NO_\lambda}
    \Pr[\langle\widetilde P^\ora,V_2^\ora\rangle=\ACCEPT]
    \ge1-(N_2+1)\delta.
\]
The main idea is to use incompressibility to limit how many code preimages
a classical history forces, while leaving sufficiently many alternatives of pairs marked in the actual oracle that the cheating prover can choose. For the rest of the overview, fix the actual oracle $\ora_{\NO}$ sampled from $\NO_\lambda$ that is used in the interaction of $\widetilde{P}$ and $V_2$.

\paragraph{What the prover keeps.}
 We consider candidate oracles under which $(P_2,V_2)$ accepts perfectly,
and let $\mu$ be the distribution which is $\YES_\lambda$ conditioned $(P_2,V_2)$ accepting perfectly.
The cheating prover keeps
\[
  \tau\quad\text{(the transcript so far)},\qquad
  E\subseteq\operatorname{supp}(\mu)\quad\text{(remaining candidate oracles)}.
\]
Initially, $\tau$ is empty and $E=\operatorname{supp}(\mu)$.
We will have the invariant that each candidate $\widehat\ora\in E$ is consistent with our replies, that is, running
$P_2^{\widehat\ora}$ on the random challenges in $\tau$ produces those replies (recall that the prover is deterministic).
The distribution $\mu$ is fixed independently of $\ora_{\NO}$ but we will modify $E$ as the interaction proceeds.

\paragraph{Choosing the next reply.}
On receiving challenge $r_j$, we try possible replies $m$ one by one until we find one that satisfies the conditions we describe below, and then we send it.
For each $m$, we first form
\[
  F_m=\{\widehat\ora\in E:P_2^{\widehat\ora}(\tau,r_j)=m\}.
\]
Thus, $F_m$ is the set of candidate remaining oracles under which $P_2$ would reply with $m$. In another way, these are the possibilities we can keep if we
choose this reply. 

For a set $T\subseteq\mathcal Z_\lambda$ of pairs $(y,v)$, define
\[
  F_m[T]=\{\widehat\ora\in F_m:
    \widehat\ora(y,v)=1\text{ for every }(y,v)\in T\}.
\]
That is, $F_m[T]$ consists of candidates that would send $m$ and mark
all pairs in $T$. Intuitively, keeping $F_m[T]$ is essentially a commitment for our cheating prover:
from then on, every candidate oracle we keep must mark these pairs.
Thus, we can make this commitment only if our actual oracle $\ora_{\NO}$ also
marks every pair in $T$, since the verifier might check those pairs and we will fail. 

To keep enough probability under $\mu$ when updating $E$, we use a polynomial
size bound $K(\lambda)$ and thresholds $\gamma_j(k)>0$ for round $j$
and $k$ pairs (see \cref{app:parameters,app:eligible}). We consider only
sets satisfying
\[
  |T|<K,\qquad \mu(F_m[T])\ge\gamma_j(|T|).
\]
Here the probability is over a candidate oracle sampled from the
fixed distribution $\mu$. We choose a largest such $T$, recomputing
this choice whenever $F_m$ changes. And if end up not being able to choose a suitable $T$, we move onto the next possible message.

\begin{protocol}
  \protocoltitle{$\widetilde P^{\ora_{\NO}}(r_j)$}
  \textbf{State:} $\tau,E$
  \protocolbody
  \begin{enumerate}
    \item For each reply $m$, in a fixed order:
      \begin{enumerate}[label*=\arabic*.]
        \item Initialize $F_m=\{\widehat\ora\in E:
          P_2^{\widehat\ora}(\tau,r_j)=m\}$.
        \item While there exists $T\subseteq\mathcal Z_\lambda$ with
          $|T|<K$ and $\mu(F_m[T])\ge\gamma_j(|T|)$:
          \begin{enumerate}[label*=\arabic*.]
            \item Choose such a $T$ of maximum cardinality.
            \item Query $\ora_{\NO}$ on all pairs in $T$ and retain only
              \[
                b\gets\bigwedge_{(y,v)\in T}\ora_{\NO}(y,v).
              \]
            \item If $b=1$, set
              \[
                E\gets F_m[T],\qquad \tau\gets\tau\|(r_j,m),
              \]
              send $m$, and return.
            \item If $b=0$, set $F_m\gets F_m\setminus F_m[T]$
              and repeat the inner loop.
          \end{enumerate}
        \item Discard $F_m$ and try the next $m$ with $(\tau,E)$ unchanged.
      \end{enumerate}
    \item If every reply fails, output $\bot$ and terminate.
  \end{enumerate}
\end{protocol}

Now we explain the intuition behind the rules in the cheating prover.

\paragraph{Why we remove candidates after a failed test.}
Observe that if $b=0$, every candidate in $F_m[T]$ marks a pair in $T$ that
$\ora_{\NO}$ leaves unmarked. We remove these candidates and try
again with the same reply $m$.
To see why, consider for example $V_2$ that sends $y$ and a public hash $h$, and expects a
bucket label $m$, and later asks for $v$ with $h(v)=m$ and
$\ora(y,v)=1$. One unmarked preimage does not necessarily rule out other
preimages in that bucket. Also note that failed tests change only $F_m$ and
each new reply starts afresh from $E$.

\paragraph{Why we choose a largest set $T$.}
Observe that if $V_2$ asks for two preimages and checks both, testing only one pair
could leave another required answer unchecked. Choosing a largest
$T$ prevents this: after a successful test, every pair marked by all
candidates in $E=F_m[T]$ belongs to $T$.

More precisely, using the incompressibility property and by our threshold, we get $|T|\le K-2$.
Maximality then gives the following bound for every nonempty set
$U$ of additional pairs outside $T$:
\[
  \Pr_{\widehat\ora\leftarrow\mu}
    [\forall(y,v)\in U,\ \widehat\ora(y,v)=1
      \mid\widehat\ora\in E]
    \le\delta^{2|U|}
  \qquad(|T\cup U|<K).
\]
Here we sample a candidate according to $\mu$, conditioned on its
belonging to the updated $E$. Taking $U=\{(y,v)\}$ shows that no
pair outside $T$ is marked by every candidate. The same bound holds
initially with $T=\varnothing$. 

\paragraph{Why the prover finds a reply.}
We now show that $\widetilde P$ finds a reply in every round, except
with probability at most $\delta$ over $\ora_{\NO}\leftarrow\NO_\lambda$
and the public coins.
To see how, fix the public random tape and a possible state $(\tau,E)$
before round $j$. Let $S$ be the pairs marked by every candidate in $E$.

We run the search with every test declared to fail, obtaining a fixed
sequence $T_1,T_2,\ldots$ without querying $\ora_{\NO}$. Observe that:
\begin{itemize}
  \item No candidate is removed twice, and our thresholds ensure that
    the removed candidates have total $\mu$-probability greater than
    $\mu(E)/2$.
  \item By the preceding bound for additional pairs, for every nonempty
    $U$ outside $S$, candidates removed by tests with $U\subseteq T_i$
    account for at most a $2\delta^{2|U|}$ fraction of this total probability.
\end{itemize}
In particular, the tests cannot all require the same additional pair
$(y,v)$. This is how the preceding incompressibility argument ensures
that we have different choices of pairs to test.

Recall that under $\NO_\lambda$, each pair is independently marked
with probability $\delta$. The two facts above let us apply
\cref{app:containment}: some tested $T_i\setminus S$ is entirely
marked, except with probability $2^{-2^{\Omega(\lambda)}}$.
If $(\tau,E)$ is reached, the pairs in $S$ have already been tested
and are marked, so this lets us send a reply.
We use this estimate without conditioning on reaching $(\tau,E)$.
A union bound over the $2^{\poly(\lambda)}$ possible states and
averaging over the public coins then give
\[
  \Pr_{\ora_{\NO}\leftarrow\NO_\lambda,\ \text{public coins}}
    [\widetilde P\text{ outputs }\bot\text{ in some round}]\le\delta
\]
(\cref{app:search-success}). The search includes all replies $m$,
including the bucket labels in our earlier example.

\paragraph{Why the verifier accepts.}
It remains to bound $V_2$'s rejection probability when every search
succeeds. Let $S$ be the pairs marked by every candidate in the final $E$.
We first run $V_2$'s final computation by answering one on $S$ and
zero elsewhere. By the bound for additional pairs and $N_2\delta^2<1$,
some candidate in $E$ agrees with all these answers. Perfect acceptance
for that candidate therefore makes this computation output $\ACCEPT$.

We now compare these answers with those of $\ora_{\NO}$, conditioned
on the public coins and the retained test bits $b$.
All pairs in $S$ are marked in $\ora_{\NO}$.
Recall that a failed test records only that not all pairs in $T$ are
marked. Conditioning on these failures cannot increase the marking
probability of a pair outside $S$ above $\delta$.
A union bound over the $N_2$ queries of the computation just described
therefore bounds the probability that any answer changes by $N_2\delta$
(\cref{lem:final-database-check}). Combining this with the probability
that some search fails, we obtain
\[
  \E_{\ora\leftarrow\NO_\lambda}
    \Pr[\langle\widetilde P^\ora,V_2^\ora\rangle=\REJECT]
    \le (N_2+1)\delta,
\]
as required.

%% file: sections/preliminaries.tex
\section{Preliminaries}
\label{sec:prelims}

In this section, we introduce our notation, oracle model, and proof
systems. We first define proof systems relative to a fixed oracle,
then consider problems in which the oracle table itself is the input.

\subsection{Notation and Conventions}
\label{sec:notation}

We write $[m]=\{1,\ldots,m\}$, $2^S$ for the power set of $S$, and
$\binom St$ for the collection of its $t$-element subsets.
We use $x\|y$ for concatenation and $\bot$ for failure. When the input
is an oracle table, $\lambda\ge1$ is the length of the unary input;
dependence on $\lambda$ is suppressed when clear from context.
We write $\poly(\cdot)$ for a polynomial bound. A nonnegative function
is negligible if it is eventually smaller than the inverse of every
positive polynomial. Logarithms are natural unless their base is
indicated.

\paragraph{Distributions and sampling.}
We write $x\leftarrow S$ for uniform sampling from a finite set $S$
and $x\leftarrow\mu$ for sampling from a distribution $\mu$, with
$\mu(z)=\Pr_{x\leftarrow\mu}[x=z]$ and
$\mu(E)=\Pr_{x\leftarrow\mu}[x\in E]$. For an event $E$ with
$\mu(E)>0$, $\mu\mid E$ denotes the conditional distribution, which
assigns probability $\mu(z)/\mu(E)$ to each $z\in E$ and zero otherwise.
We denote the support of $\mu$ by $\operatorname{supp}(\mu)=\{z:\mu(z)>0\}$.
When the meaning is clear, we also use a random variable to denote its
distribution. The distribution
$\mathrm{Bernoulli}(p)$ assigns probability $p$ to one and $1-p$ to zero.
To compare distributions $\mu,\nu$ on the same finite set, we use
their total-variation distance,
\[
  \TV(\mu,\nu)=\frac12\sum_z|\mu(z)-\nu(z)|.
\]
Total-variation distance cannot increase under a randomized algorithm;
for any $f$ taking values in $[0,1]$,
\[
  \bigl|\E_{\mu}f-\E_{\nu}f\bigr|\le\TV(\mu,\nu).
\]

\paragraph{Registers and states.}
We follow standard register and circuit conventions~\cite{MW05,JJUW11}.
A register is a collection of qubits, denoted by a sans-serif letter
such as $\mathsf X$. Juxtaposition denotes joint registers and tensor
products. A state is a positive semidefinite operator $\rho$ of trace
one. We identify a pure state $\ket\psi$ with $\ket\psi\!\bra\psi$.
We use register labels as subscripts or superscripts when needed and
write $\|\cdot\|_2$ for the Euclidean norm.
We compare quantum states $\rho,\sigma$ using trace distance,
$\tfrac12\|\rho-\sigma\|_1$, where
$\|A\|_1=\operatorname{Tr}\sqrt{A^\dagger A}$.
Trace distance cannot increase under a quantum algorithm and bounds
the difference in acceptance probabilities. For pure states it is at
most $\|\ket\psi-\ket\phi\|_2$.

\paragraph{Algorithms and interaction.}
We use the standard uniform quantum circuit model, allowing classical
control and computational-basis measurements~\cite{MW05,JJUW11}.
An algorithm may receive both classical and quantum inputs.
Verifier descriptions are fixed independently of the oracle.
We write $x\leftarrow A(a)$ for a classical output and
$\mathsf R\leftarrow A(a)$ when $A$ prepares its output in register
$\mathsf R$.

Interactive algorithms exchange messages and may retain private
memory between rounds. We write $V[i]$ for the verifier's computation
in round $i$, which receives its private state and the incoming
message and returns a new private state and an outgoing message.
The verifier starts in $\ket0$, uncorrelated with the prover.
Their classical messages form the transcript. We use fixed message lengths and a
fixed number of rounds, padding when necessary.

We purify each party's local computation while measuring outgoing
classical messages~\cite[Section~4.1]{AKN98}. Thus, for a fixed input,
oracle, and transcript of nonzero probability, its conditional private
state is pure (\cref{lem:classical-prover}). We suppress these
dependencies in ket notation when clear from context.

We write $\langle P^*,V\rangle(x)$ for the verifier's decision on common
input $x$. Acceptance and rejection are denoted by $\ACCEPT$ and
$\REJECT$ and encoded by one and zero. Unless otherwise specified,
acceptance probabilities are over the parties' randomness and
measurements, with the input and oracle fixed.

We reserve $P$ and its unstarred variants for designated provers, and use
$P^*$ and starred variants for arbitrary, possibly cheating provers.
Completeness is required for the designated prover, and soundness against
every prover.
We write $\widetilde P$ for the cheating prover constructed in our lower bound.
A protocol transformation may change both parties, taking $(P,V)$ to
$(P',V')$; we use $(P_i,V_i)$ for successive protocols.

\paragraph{Information measures.}
We measure mutual information $I(X:Y)$ and conditional mutual information
$I(X:Y\mid Z)$ in bits, representing classical random variables by
diagonal quantum states~\cite[Chapter 5]{Watrous}.

\subsection{Oracle Access}

We use the standard quantum oracle model~\cite[Section~III.4]{Watrous09}.
A classical oracle $\ora:\bits^*\to\bits$ is a Boolean function,
which a quantum algorithm accesses through the gate
\[
  U_\ora\ket z_{\mathsf X}\ket o_{\mathsf A}\ket w_{\mathsf W}
    =\ket z_{\mathsf X}\ket{o\oplus\ora(z)}_{\mathsf A}\ket w_{\mathsf W},
\]
where $\mathsf X$ is the query-input register, $\mathsf A$ is the
answer register, and $\mathsf W$ is the remaining workspace;
$z$ is the query string and $o\in\bits$ labels the answer register.
The gate $U_\ora$ acts linearly on superpositions, and each application counts
as one oracle query. A classical query specifies an address $z$ and
receives the bit $\ora(z)$. Thus polynomial running time implies a
polynomial query bound, whereas a query bound alone does not restrict
computation between queries.

We write $A^\ora$ for an algorithm with access to $\ora$ and specify
whether it may query classically or in superposition.
For a set $F$ of addresses, its membership oracle is
\[
  \ora_F(z)=
  \begin{cases}
    1,&z\in F,\\
    0,&z\notin F.
  \end{cases}
\]

\subsection{Proof Systems Relative to a Fixed Oracle}

We first fix an oracle $\ora:\bits^*\to\bits$ and define proof
systems for languages $L\subseteq\bits^*$, then give their promise
versions. We write $n=|x|$ for the input length. Throughout these
definitions, verifiers are uniform, run in polynomial time, and may
query $\ora$ in superposition. Provers are
unrestricted strategies with oracle access. In an interactive proof,
both parties receive the common input $x$ and retain private memory
between messages.

\begin{definition}[$\QMA^\ora$ and $\QCMA^\ora$~\cite{AK07,MW05}]
  \label{def:qma-qcma}
  A language $L$ is in $\QMA^\ora$ if there are a polynomial $a$ and
  a polynomial-time quantum verifier $V$ such that
  \[
    \begin{aligned}
      \forall x \in L:&\quad
      \exists\ket\psi\quad\Pr[V^\ora(x,\ket\psi)=\ACCEPT]\ge 2/3,\\
      \forall x \notin L:&\quad
      \forall\ket\psi\quad\Pr[V^\ora(x,\ket\psi)=\ACCEPT]\le 1/3,
    \end{aligned}
  \]
  where $\ket\psi$ ranges over all $a(n)$-qubit pure states.
  The witness may depend on $x$ and $\ora$.
  The class $\QCMA^\ora$ has the same definition with a classical witness
  $w \in \bits^{a(n)}$ in place of $\ket\psi$; its verifier remains quantum.
\end{definition}

Observe that pure witnesses suffice: for a fixed verifier, input,
and oracle, a pure state attains the maximum acceptance probability
by linearity and finite dimensionality.

\begin{definition}[$\QCIP^\ora$~\cite{BLW}]
  \label{def:qcip}
  A $\QCIP$ verifier is a polynomial-time quantum verifier that exchanges
  classical messages with one computationally unbounded classical prover.
  The verifier may retain private quantum memory between messages.
  The number of rounds and total communication are polynomial in $n$.
  A language $L$ is in $\QCIP^\ora$ if there is a protocol $(P,V)$
  with such a verifier and prover satisfying
  \[
    \begin{aligned}
      \forall x \in L:&\quad
      \Pr[\langle P^\ora,V^\ora\rangle(x)=\ACCEPT]\ge 2/3,\\
      \forall P^*\ \forall x \notin L:&\quad
      \Pr[\langle (P^*)^\ora,V^\ora\rangle(x)=\ACCEPT]\le 1/3.
    \end{aligned}
  \]
  We write $\QCIP[k]^\ora$ for the restriction to $k$ exchanged classical
  messages, excluding the verifier's final decision.
  In particular, $\QCIP[1]^\ora=\QCMA^\ora$.
\end{definition}

The following observation shows that allowing the prover private
quantum memory does not change the class. In particular, for both
completeness and soundness, it suffices to consider classical randomized
provers given the full transcript at each turn.

\begin{lemma}[Classical prover strategies]
  \label{lem:classical-prover}
  Let $P^*$ be a prover with arbitrary quantum memory, exchanging a
  bounded number of bounded-length classical messages with an initially
  uncorrelated quantum verifier $V$.
  Fix a common input $x$ and a classical oracle $\ora$.
  \begin{enumerate}
    \item The prover $P^*$ has a classical randomized replacement that
      receives $x$ and the full transcript $\pi$, has access to $\ora$,
      and retains no private state between replies. The same replacement
      preserves the transcript distribution and conditional verifier
      states against every verifier.
    \item Let $\mathsf V,\mathsf P$ be the original parties' private
      registers. Conditioned on a transcript $\pi$ of nonzero probability,
      their joint state satisfies
      \[
        \rho_\pi^{\mathsf V\mathsf P}
          =\rho_\pi^{\mathsf V}\otimes\rho_\pi^{\mathsf P}.
      \]
      Each marginal depends only on that party's strategy, $x$, $\ora$,
      and $\pi$.
    \item For fixed $V,x,\ora$, the maximum acceptance probability
      is attained by a deterministic reply table.
  \end{enumerate}
\end{lemma}

We give the proof in \cref{app:classical-provers}. Thus, we may describe
provers by their replies to classical transcripts. As in the standard
definition of interactive proofs,
$P$ is one strategy receiving the common input $x$
\cite[Definition A.1 and footnote 5]{Goldreich18IP}.
Requiring $V$ to be a classical probabilistic polynomial-time verifier
with classical oracle queries gives $\mathsf{IP}^\ora$.
The unrelativized classes are obtained by taking $\ora\equiv0$.

\paragraph{Promise problems.}
We use the standard promise-problem formulation~\cite{Goldreich06Promise,Watrous09}:
$L=(L_{\YES},L_{\NO})$ is a pair of disjoint subsets of $\bits^*$.
The preceding definitions extend to promise problems by requiring
completeness on $L_{\YES}$ and soundness on $L_{\NO}$.
No acceptance condition is imposed outside $L_{\YES}\cup L_{\NO}$.
The language case is recovered by taking
\[
  L_{\NO}=\bits^*\setminus L_{\YES}.
\]

\subsection{Oracle Inputs}

In this subsection, we define problems whose input is a finite table
accessed through oracle queries. Running time, number of queries, and
message lengths are measured as functions of a parameter $\lambda$,
rather than the table's length, following the query-model conventions
of~\cite[Section~1]{BBCMW01} and~\cite[Definition~4]{Zha25}.

\begin{definition}[Oracle inputs]
  \label{def:oracle-input}
  An algorithm receives $1^\lambda$ and query access to a finite table
  $\ora:\mathcal Z_\lambda\to\bits$.
  Query length, running time, and proof or message lengths are measured
  as functions of $\lambda$; query length is polynomial in $\lambda$.

  An oracle-input decision problem specifies a set
  \[
    L_\lambda\subseteq\{\ora:\mathcal Z_\lambda\to\bits\}
    \qquad\text{for each }\lambda.
  \]
  The task is to decide whether the supplied oracle $\ora$ belongs to $L_\lambda$.
  We apply the proof-system definitions above to these inputs, with
  common input $1^\lambda$ and access to $\ora$. Thus, all oracle
  quantifiers below range over tables $\mathcal Z_\lambda\to\bits$.

  \textbf{QMA.}
  A proof system has a polynomial-time quantum verifier $V$ and a
  polynomial witness bound $a$. The verifier may query $\ora$ in
  superposition and must satisfy, for every $\lambda$,
  \[
    \begin{aligned}
      \forall\ora\in L_\lambda:&\quad
      \exists\ket\psi\quad
      \Pr[V^\ora(1^\lambda,\ket\psi)=\ACCEPT]\ge 2/3,\\
      \forall\ora\notin L_\lambda:&\quad
      \forall\ket\psi\quad
      \Pr[V^\ora(1^\lambda,\ket\psi)=\ACCEPT]\le 1/3,
    \end{aligned}
  \]
  where $\ket\psi$ has $a(\lambda)$ qubits.
  Restricting the witness to a classical string gives $\QCMA$.

  \textbf{QCIP.}
  A proof system $(P,V)$ has a polynomial-time quantum verifier $V$
  with superposition access to $\ora$ and an unrestricted classical
  prover $P$ with oracle access such that, for every $\lambda$,
  \[
    \begin{aligned}
      \forall\ora\in L_\lambda:&\quad
      \Pr[\langle P^\ora,V^\ora\rangle(1^\lambda)=\ACCEPT]\ge 2/3,\\
      \forall P^*\ \forall\ora\notin L_\lambda:&\quad
      \Pr[\langle (P^*)^\ora,V^\ora\rangle(1^\lambda)=\ACCEPT]\le 1/3.
    \end{aligned}
  \]
  The parties exchange classical messages, and the verifier may retain
  private quantum memory. Its running time bounds the number of rounds
  and total communication. The prover may read the entire finite table $\ora$.

  \textbf{IP.}
  Requiring a classical probabilistic polynomial-time verifier with
  classical queries gives an $\mathsf{IP}$ proof system for oracle inputs.
\end{definition}

\paragraph{Promise problems.}
The promise version specifies disjoint sets of oracle tables
$L_{\YES,\lambda}$ and $L_{\NO,\lambda}$.
The preceding definition extends to this setting by requiring completeness
on $L_{\YES,\lambda}$ and soundness on $L_{\NO,\lambda}$, with no condition
outside their union.
The decision version is recovered by taking
\[
  L_{\YES,\lambda}=L_\lambda,
  \qquad
  L_{\NO,\lambda}
  =\{\ora:\mathcal Z_\lambda\to\bits\}\setminus L_\lambda.
\]

\paragraph{Distributional guarantees.}
We will also study distributions over oracle inputs. For a
distribution $\mathcal P$, we distinguish two probabilities:
\begin{itemize}
  \item the outer probability $\Pr_{\ora\leftarrow\mathcal P}[\cdot]$
    is the probability that an oracle sampled from $\mathcal P$
    satisfies the stated acceptance condition;
  \item the inner acceptance probability is over an execution with
    $\ora$ fixed.
\end{itemize}
The supports of the distributions need not be disjoint. In
\cref{app:diagonal}, we select disjoint promise sets to obtain
\cref{thm:oracle-input}, then construct the fixed oracle of
\cref{cor:oracle}.

For oracle inputs, we usually omit the common input
$1^\lambda$ from $\langle P^{\ora},V^\ora\rangle(1^\lambda)$.

%% file: sections/incompressibility.tex
\section{Codes and Incompressibility}
\label{sec:codes}
In this section, we provide the coding theory tools we use, and then we provide two incompressibility theorems, one for the \emph{code-preimage problem} \cite{YZ24} and one for \emph{randomly marked sets}.

\subsection{Codes and Code-Preimage Problem}
Let $C\subseteq\Sigma^q$ be a code and let
$H=(H_1,\ldots,H_q)$, where $H_i:\Sigma\to\bits$. We write
$H(v)=(H_i(v_i))_{i\in[q]}$.
For a target $y\in\bits^q$, the code-preimage problem asks for
$v\in C$ satisfying $H(v)=y$. We now present the following result about this problem, that follows from \cite{BHV}.

\begin{lemma}
  \label{lem:code}
  There exist an explicit code family $C\subseteq\Sigma^q$ with an injective $\mathbb F_q$-linear encoding computable
  in polynomial time and a uniform
  quantum polynomial-time algorithm $\mathsf{CodePreimageDecode}$ with
  the following properties for all sufficiently large $\lambda$.
  Here $\Sigma=\mathbb F_q^\lambda$, where $q$ is prime and
  $\lambda^5<q\le2\lambda^5$.
  \begin{enumerate}
    \item \textbf{Expansion.}
      For every $W \subseteq C$ with $|W|\le 2^\lambda$,
      \begin{equation}
        \label{eq:expansion}
        \abs{\{(i,v_i):i \in [q],\ v \in W\}} \ge \frac q2|W|.
      \end{equation}
    \item \textbf{Quantum decoding.}
      Sample $H_1,\ldots,H_q:\Sigma\to\bits$ by choosing all entries
      independently from $\mathrm{Bernoulli}(\lambda^{-4})$.
      Then, there is a pure state $\ket{\Psi_H}$ such that with probability at least $1-2^{-\lambda}$ over $H$, for every
      $u\in\bits^\lambda$, the output
      $v\leftarrow\mathsf{CodePreimageDecode}(\ket{\Psi_H},u\|0^{q-\lambda})$
      satisfies
      \begin{equation}
        \label{eq:bhvquantum}
        \Pr[v \in C\ \text{and}\ H(v) = u\|0^{q-\lambda}]
        \ge 1-2^{-\lambda}.
      \end{equation}
  \end{enumerate}
\end{lemma}
\begin{proof}
  \emph{Expansion.}
  We use the proof of Corollary 4.17 in~\cite{BHV}, together with
  Theorem 4.16 and Corollary 4.18 there: $C$ is a
  $(2^{\lambda+1},A)$-expander with $A>q/2$ in the graph that joins
  each codeword $v$ to its coordinate-symbol pairs $(i,v_i)$.
  The neighbors of $W$ are exactly the coordinate-symbol pairs in
  \eqref{eq:expansion}, so their number is at least $(q/2)|W|$.

  \emph{Quantum decoding.}
  We use the following state from \cite{Liu23} and \cite{BHV}:
  \[
    \ket{\Psi_H}\propto
      \bigotimes_{i=1}^q\left[
        \left(\sum_{z:H_i(z)=0}\ket z\right)\otimes
        \left(\sum_{z:H_i(z)=1}\ket z\right)\right],
    \qquad
    \ket C=\frac1{\sqrt{|C|}}\sum_{v\in C}\ket v.
  \]
  Since the code $C$ is public and has an efficient injective linear
  encoding, the decoder can prepare $\ket C$ in polynomial time,
  with trace-distance error at most $2^{-4\lambda}$.

  We then run the Yamakawa--Zhandry decoder \cite{YZ24} on
  $(\ket{\Psi_H},\ket C,u\|0^{q-\lambda})$, ignoring the padding.
  Write $e(H,u)$ for the probability that this algorithm does not
  output a codeword $v$ satisfying $H(v)=u\|0^{q-\lambda}$.
  For each fixed $u$, Theorem 5.22 and the proof of Corollary 5.23
  of~\cite{BHV} give
  \[
    \Pr_H[e(H,u)>16\cdot2^{-2\lambda}+2\cdot2^{-4\lambda}]
      \le2^{-2\lambda}.
  \]
  Here $e(H,u)$ includes the error in preparing $\ket C$.
  Since $16\cdot2^{-2\lambda}+2\cdot2^{-4\lambda}<2^{-\lambda}$
  for sufficiently large $\lambda$, a union bound over the $2^\lambda$
  targets gives \eqref{eq:bhvquantum}, as in~\cite[Corollary 5.24]{BHV}.
\end{proof}

\subsection{Incompressibility}
Similar to the incompressibility theorems in \cite{QFire}, and implicit in \cite{BHV}, we provide an incompressibility theorem for the code-preimage problem and for randomly marked sets.
\begin{theorem}[Code-preimage incompressibility]
  \label{thm:incompressibility}
  Let $C\subseteq\Sigma^q$ be as in \cref{lem:code}, and set
  $\delta=2^{-\lceil\lambda/12\rceil}$.
  Fix integers $s\ge0$ and $1\le t\le2^\lambda$.
  Let $\mathsf{Enc}$ and $\mathsf{Dec}$ be arbitrary classical algorithms,
  fixed before sampling $H$, with $\mathsf{Enc}$ outputting an
  $s$-bit message. They need not be efficient. Their only inputs
  are those specified below, and all randomness comes from the
  shared random tape $R$.

  \begin{protocol}
    \protocoltitle{$\mathsf{Expt}_{\mathrm{CP}}(\mathsf{Enc},\mathsf{Dec})$}
    \begin{enumerate}
      \item Sample $H_1,\ldots,H_q:\Sigma\to\bits$ with independent
        $\mathrm{Bernoulli}(\lambda^{-4})$ entries, and sample an
        independent shared random tape $R$.
      \item Run $m\gets\mathsf{Enc}(H,R)$.
      \item Run $T\gets\mathsf{Dec}(m,R)$, without oracle access.
      \item Output $1$ if $T$ consists of $t$ distinct pairs
        $(y,v)\in\bits^q\times C$ satisfying $H(v)=y$ for every
        $(y,v)\in T$. Otherwise, output $0$.
    \end{enumerate}
  \end{protocol}

  For all sufficiently large $\lambda$,
  \begin{equation}
    \label{eq:incompressibility-message}
    \Pr[\mathsf{Expt}_{\mathrm{CP}}(\mathsf{Enc},\mathsf{Dec})=1]
      \le2^s\delta^{4t}.
  \end{equation}
  Moreover, for every fixed $\mathsf{Dec}$ and $\varepsilon>0$,
  except with probability at most $2^s\delta^{4t}/\varepsilon$ over
  $H$, every encoder succeeds with probability at most $\varepsilon$
  over $R$.
\end{theorem}
\begin{proof}
  Write $\mathcal X=\bits^q\times C$ and
  $A_H=\{(y,v)\in\mathcal X:H(v)=y\}$.
  First fix $T\in\binom{\mathcal X}t$. If two pairs in $T$ give
  different targets for the same codeword, then
  $\Pr_H[T\subseteq A_H]=0$. Otherwise, $T$ specifies $t$ distinct
  codewords, so \eqref{eq:expansion} gives
  \[
    J(T):=\{(i,v_i):i\in[q],\ (y,v)\in T\},
    \qquad |J(T)|\ge\frac q2t.
  \]
  Each $(y,v)\in T$ requires $H_i(v_i)=y_i$ for every $i\in[q]$.
  If these requirements conflict at any entry, their joint probability
  is zero. Otherwise, the $|J(T)|$ entries are independent, and each
  prescribed value occurs with probability at most $1-\lambda^{-4}$.
  Thus, since $q>\lambda^5$, for all sufficiently large $\lambda$,
  \begin{equation}
    \label{eq:incompressibility-fixed-set}
    \Pr_H[T\subseteq A_H]
      \le(1-\lambda^{-4})^{qt/2}
      \le2^{-\lambda t/(2\ln2)}
      \le\delta^{4t}.
  \end{equation}

  For fixed $R$ and $m$, the decoder's output is fixed.
  Success requires that some $s$-bit message decode to $t$ pairs in $A_H$.
  Applying \eqref{eq:incompressibility-fixed-set} and taking a union
  bound over messages,
  \[
    \begin{aligned}
      \Pr[\mathsf{Expt}_{\mathrm{CP}}(\mathsf{Enc},\mathsf{Dec})=1]
        &\le\sum_{m\in\bits^s}\E_R
          \Pr_H[\mathsf{Dec}(m,R)\in\tbinom{A_H}t]\\
        &\le2^s\delta^{4t}.
    \end{aligned}
  \]
  To obtain one exceptional set for all encoders, take
  $\mathsf{Enc}^*$ to output the first message $m$ satisfying
  $\mathsf{Dec}(m,R)\in\binom{A_H}t$, or $0^s$ if none exists.
  Its success probability
  for each fixed $H$ is at least that of every other encoder.
  Markov's inequality bounds the probability over $H$ that this
  conditional success probability exceeds $\varepsilon$ by
  $2^s\delta^{4t}/\varepsilon$.
\end{proof}

\begin{theorem}[Independent-marking incompressibility]
  \label{thm:independent-incompressibility}
  Let $\mathcal X$ be a finite set and $0<\delta<1$.
  Fix integers $s\ge0$ and $1\le t\le|\mathcal X|$.
  Let $\mathsf{Enc}$ and $\mathsf{Dec}$ be arbitrary classical algorithms,
  fixed before sampling $\ora$, with $\mathsf{Enc}$ outputting an
  $s$-bit message. They need not be efficient. Their only inputs
  are those specified below, and all randomness comes from the
  shared random tape $R$.

  \begin{protocol}
    \protocoltitle{$\mathsf{Expt}_{\mathrm{Mark}}(\mathsf{Enc},\mathsf{Dec})$}
    \begin{enumerate}
      \item Sample $\ora:\mathcal X\to\bits$ with independent
        $\mathrm{Bernoulli}(\delta)$ entries, and sample an
        independent shared random tape $R$.
      \item Run $m\gets\mathsf{Enc}(\ora,R)$.
      \item Run $T\gets\mathsf{Dec}(m,R)$, without oracle access.
      \item Output $1$ if $T$ consists of $t$ distinct elements of
        $\mathcal X$ satisfying $\ora(z)=1$ for every $z\in T$.
        Otherwise, output $0$.
    \end{enumerate}
  \end{protocol}

  Then
  \begin{equation}
    \label{eq:independent-incompressibility}
    \Pr[\mathsf{Expt}_{\mathrm{Mark}}(\mathsf{Enc},\mathsf{Dec})=1]
      \le2^s\delta^t.
  \end{equation}
  Moreover, for every fixed $\mathsf{Dec}$ and $\varepsilon>0$,
  except with probability at most $2^s\delta^t/\varepsilon$ over
  $\ora$, every encoder succeeds with probability at most
  $\varepsilon$ over $R$.
\end{theorem}
\begin{proof}
  Write $A_\ora=\{z\in\mathcal X:\ora(z)=1\}$.
  For fixed $R,m$, a decoder output $T\in\binom{\mathcal X}t$ satisfies
  \[
    \Pr_\ora[T\subseteq A_\ora]=\prod_{z\in T}\Pr[\ora(z)=1]=\delta^t.
  \]
  A union bound over the $2^s$ messages gives
  \eqref{eq:independent-incompressibility}. As in the preceding proof,
  take the encoder that outputs the first message decoding to $t$
  distinct marked elements, or $0^s$ if none exists. Applying Markov's inequality
  to its success probability conditional on $\ora$ proves the bound
  $2^s\delta^t/\varepsilon$, simultaneously for all encoders.
\end{proof}

%% file: sections/interactive-proof-tools.tex
\section{Public Coins and Perfect Completeness}
\label{sec:protocol-transformations}

In this section, we observe that the public-coin transformation of
Goldwasser and Sipser~\cite{GS} and the perfect-completeness
transformation of F\"urer et al.~\cite{FGMSZ} apply to verifiers with
polynomially many queries and unrestricted local computation.
We state their combined guarantee after specifying the query model.

\begin{definition}[Verifiers in the query model]
  \label{def:query-model}
  On common input $1^\lambda$, a verifier in the query model exchanges
  classical messages with an unbounded classical prover. Both parties
  have access to an oracle input $\ora$. The prover has unrestricted
  oracle access; the number of rounds, total communication, and
  number and length of the verifier's queries are polynomial in $\lambda$.
  The verifier's local computation and workspace are unrestricted.
  We also allow both parties classical advice of arbitrary length,
  fixed independently of the oracle input.
  \begin{enumerate}
    \item A \emph{classical} verifier makes classical queries and uses
      polynomially many uniform random bits.
    \item A \emph{quantum} verifier may make quantum queries and
      uses the circuit model of \cref{sec:prelims}, without a
      polynomial bound on circuit size.
  \end{enumerate}
  We call a classical verifier \emph{public-coin} if its messages
  before the final decision are fresh uniform random strings and its
  final computation is deterministic given the transcript and oracle.
  We fix round counts and message lengths by padding.
\end{definition}

The two transformations give the following lemma. We verify their
query bounds in \cref{app:combined-transformations}.

\begin{lemma}[Public coins and perfect completeness]
  \label{lem:public-coin-perfect}
  Let $(P,V)$ be a protocol in the classical query model of
  \cref{def:query-model}. There is a public-coin protocol $(P',V')$
  in the same model, where $P'$ is deterministic and $V'$ depends
  only on $V$ and the polynomial bounds in \cref{def:query-model}.
  The verifier $V'$ makes all oracle queries after the last prover reply.
  For every oracle input $\ora$:
  \begin{enumerate}
    \item \textbf{Designated prover.}
      If $P^\ora$ makes $V^\ora$ accept with probability at least
      $2/3-1/100$, then $(P')^\ora$ makes $(V')^\ora$ accept with
      probability $1$.
    \item \textbf{Arbitrary provers.}
      If every prover $P^*$ makes $V^\ora$ accept with probability at most
      $1/3+1/100$, then every prover $P^*$ makes $(V')^\ora$ accept with
      probability at most $1/3$.
  \end{enumerate}
\end{lemma}

%% file: sections/oracle-problem.tex
\section{The Distributional Separation}
\label{sec:oracle}

In this section, we define the oracle distributions $\YES_\lambda$
and $\NO_\lambda$ and state the distributional separation. We prove
the QMA upper bound in \cref{sec:qma} and the QCIP lower bound in
\cref{sec:hardness}, then obtain the promise and fixed-oracle
separations in \cref{app:diagonal}.

\begin{definition}[YES and NO oracle distributions]
  \label{def:oracle-distributions}
  Let $C\subseteq\Sigma^q$ be the code from \cref{lem:code}. Both
  distributions are on Boolean oracles
  \begin{equation}
    \label{eq:oracle}
    \ora:\mathcal Z_\lambda\to\bits,\qquad
    \mathcal Z_\lambda=\bits^q\times C.
  \end{equation}
  Set
  \begin{equation}
    \label{eq:delta}
    p=\lambda^{-4},\qquad
    \delta=2^{-\lceil\lambda/12\rceil}.
  \end{equation}
  Both distributions sample every entry of
  $H_1,\ldots,H_q:\Sigma\to\bits$ independently from
  $\mathrm{Bernoulli}(p)$. Write
  $H(v)=(H_1(v_1),\ldots,H_q(v_q))$ and define or sample $\ora$ as follows:
  \begin{enumerate}
    \item \textbf{Code-preimage structure ($\YES_\lambda$).}
      Define
      \[
        \ora(y,v)=
        \begin{cases}
          1,&H(v)=y,\ y_{\lambda+1}=\cdots=y_q=0,\\
          0,&\text{otherwise}.
        \end{cases}
      \]
    \item \textbf{Independent marking ($\NO_\lambda$).}
      For every $(y,v)\in\mathcal Z_\lambda$, sample $\ora(y,v)$ independently
      from $\mathrm{Bernoulli}(\delta)$, independently of $H$.
      The sampled hash functions $H$ do not affect $\ora$ in this case.
  \end{enumerate}
  A pair $(y,v)\in\mathcal Z_\lambda$ is a \emph{code-preimage pair}
  if $H(v)=y$, and is \emph{marked} if $\ora(y,v)=1$.
  The verifier has query access to $\ora$, but not to $H$.
\end{definition}

In the following theorem, the outer probability is over the oracle
$\ora$, and the inner probability is over an execution with $\ora$
fixed (\cref{sec:prelims}).

\begin{theorem}[Distributional separation]
  \label{thm:main}
  Let $\YES_\lambda,\NO_\lambda$ be the distributions of
  \cref{def:oracle-distributions}.

  \textbf{QMA upper bound.}
  There is a uniform polynomial-time quantum verifier $V$
  (constructed in \cref{sec:qma-verifier}) that makes at most one
  classical oracle query.
  All witnesses below are $a(\lambda)$-qubit pure states, where
  $a(\lambda)=2q\lceil\log_2|\Sigma|\rceil=O(\lambda^6\log\lambda)$.
  For all sufficiently large $\lambda$,
  \begin{align*}
    \Pr_{\ora \leftarrow \YES_\lambda}
    \left[\exists\ket\psi:\
      \Pr[V^\ora(1^\lambda,\ket\psi) = \ACCEPT] \ge 1-2^{-\lambda}\right]
    &\ge 1-\gamma_{\YES},\\
    \Pr_{\ora \leftarrow \NO_\lambda}
    \left[\forall\ket\psi:\
      \Pr[V^\ora(1^\lambda,\ket\psi) = \ACCEPT] \le \gamma_{\NO}\right]
    &\ge 1-\gamma_{\NO},
  \end{align*}
  where $\gamma_{\YES}=2^{-\lambda}$ and
  \[
    \gamma_{\NO}=\left(\delta+
      \sqrt{\frac{a(\lambda)}{2^\lambda}}\right)^{1/2}
      =2^{-\Omega(\lambda)}.
  \]

  \textbf{QCIP lower bound.}
  For every QCIP protocol $(P',V')$, all sufficiently large $\lambda$, and
  every $0\le\varepsilon<1/3$, the following holds.
  If
  \begin{equation}
    \label{eq:main-completeness}
    \Pr_{\ora\leftarrow\YES_\lambda}
      \left[\Pr[\langle(P')^\ora,(V')^\ora\rangle=\ACCEPT]\ge2/3\right]
      \ge1-\varepsilon,
  \end{equation}
  then
  \begin{equation}
    \label{eq:main-soundness}
    \Pr_{\ora\leftarrow\NO_\lambda}
      \left[\forall P^*:\
        \Pr[\langle(P^*)^\ora,(V')^\ora\rangle=\ACCEPT]\le1/3\right]
      \le2/3+\varepsilon.
  \end{equation}
  The threshold for $\lambda$ depends only on $V'$'s number of rounds
  and queries and total communication.
  All provers receive the common input $1^\lambda$ and have unrestricted
  oracle access.
\end{theorem}

%% file: sections/qma-upper-bound.tex
\section{The QMA Upper Bound}
\label{sec:qma}

In this section, we give the QMA verifier $V$ and prove the acceptance
bounds in \cref{thm:main}. Under $\YES_\lambda$, we use the decoding
guarantee of \cref{lem:code}. Under $\NO_\lambda$, we bound the
information that an arbitrary witness contains about the oracle
values at a uniformly chosen target with an all-zero suffix.

\subsection{The Verifier}
\label{sec:qma-verifier}

The verifier $V$ samples $u\leftarrow\bits^\lambda$, sets
$y=u\|0^{q-\lambda}$, and uses $\mathsf{CodePreimageDecode}$ to
obtain a candidate preimage $v$ of $y$. If $v\in C$, it evaluates
$\ora(y,v)$; otherwise, it outputs $\REJECT$.
Recall that the code $C$ is public, so the decoder prepares the code
state itself, as in \cref{lem:code}.

\begin{protocol}
  \protocoltitle{$V^\ora(1^\lambda,\mathsf W)$}
  \textbf{Input:} $1^\lambda$ and an $a(\lambda)$-qubit witness register
  $\mathsf W$.\par
  \textbf{Oracle:} $\ora$.\par
  \textbf{Output:} $\ACCEPT$ or $\REJECT$.
  \protocolbody
  \begin{enumerate}
    \item Sample $u\leftarrow\bits^\lambda$ uniformly and set
      $y=u\|0^{q-\lambda}$.
    \item Run $v\leftarrow\mathsf{CodePreimageDecode}(\mathsf W,y)$.
      If $v=\bot$ or $v\notin C$, output $\REJECT$ and terminate.
    \item Query $\ora(y,v)$ classically. Output $\ACCEPT$ if the answer is one
      and $\REJECT$ otherwise, then terminate.
  \end{enumerate}
\end{protocol}

\subsection{Acceptance under the YES Distribution}

Under $\YES_\lambda$, all entries of $H$ are independent
$\mathrm{Bernoulli}(p)$ bits, where $p=\lambda^{-4}$. By the
definition of $\YES_\lambda$, for every $u\in\bits^\lambda$ and $v\in C$,
\begin{equation}
  \label{qma:marked-preimage}
  \ora(u\|0^{q-\lambda},v)=1
  \quad\Longleftrightarrow\quad
  H(v)=u\|0^{q-\lambda}.
\end{equation}
Thus, whenever the decoder finds a code preimage of
$u\|0^{q-\lambda}$, the verifier accepts. We now apply the decoding
guarantee of \cref{lem:code}.

\begin{lemma}[Acceptance under the YES distribution]
  \label{qma:yesprob}
  The verifier $V$ uses an $a(\lambda)$-qubit witness, runs in polynomial
  time, and makes at most one classical oracle query. For all sufficiently
  large $\lambda$, the YES distribution of \cref{def:oracle-distributions}
  satisfies
  \begin{equation}
    \label{eq:gammay}
    \Pr_{\ora \leftarrow \YES_\lambda}\left[
    \max_{\ket\psi}\Pr[V^\ora(1^\lambda,\ket\psi) = \ACCEPT] < 1-2^{-\lambda}\right]
    \le \gamma_{\YES} := 2^{-\lambda}.
  \end{equation}
\end{lemma}
\begin{proof}
  With probability at least $1-2^{-\lambda}$ over the sampling of $H$,
  the hashes $H$ satisfy the guarantee of \cref{lem:code}.
  For each such sample, give $V$ the witness $\ket{\Psi_H}$.
  Averaging \eqref{eq:bhvquantum} over $u$ and using
  \eqref{qma:marked-preimage}, we obtain
  \[
    \Pr[V^{\ora}(1^\lambda,\ket{\Psi_H})=\ACCEPT]
      \ge1-2^{-\lambda}.
  \]
  This proves \eqref{eq:gammay}.
  By \cref{lem:code}, the witness has $a(\lambda)$ qubits and the decoder
  runs in polynomial time. Membership in the public linear code $C$
  is decidable in polynomial time by linear algebra. The verifier makes
  at most one classical query, to evaluate $\ora(y,v)$.
\end{proof}

\subsection{Acceptance under the NO Distribution}

We next bound the acceptance probability for arbitrary witnesses.
Observe that $V$ samples $u$ after receiving the witness, and under
$\NO_\lambda$, the functions
$\ora(u\|0^{q-\lambda},\cdot)$ are independent across
$u\in\bits^\lambda$. Thus, although the witness may depend on the
entire oracle, it contains limited information about these values
for a uniformly chosen $u$. We quantify this using an argument
closely related to Nayak's random access code bound~\cite{Nayak99}.

For each $u$, we then replace $\ora(u\|0^{q-\lambda},\cdot)$ by
an independent function with the same distribution, keeping the
witness and decoded word $v$ unchanged. After this replacement,
the probability that the final query returns one is at most $\delta$.

\begin{lemma}[Acceptance under the NO distribution]
  \label{qma:noprob}
  For all sufficiently large $\lambda$, sample
  $\ora\leftarrow\NO_\lambda$, where
  $\delta = 2^{-\lceil\lambda/12\rceil}$. Set
  \[
    \gamma_{\NO}=\left(\delta+
      \sqrt{\frac{a(\lambda)}{2^\lambda}}\right)^{1/2}
      =2^{-\Omega(\lambda)}.
  \]
  Then
  \begin{equation}
    \label{eq:gamman}
    \Pr_{\ora\leftarrow\NO_\lambda}\left[
    \max_{\ket\psi}\Pr[V^{\ora}(1^\lambda,\ket\psi) = \ACCEPT] > \gamma_{\NO}\right]
    \le \gamma_{\NO}.
  \end{equation}
\end{lemma}
\begin{proof}
  For each oracle $\ora$, choose
  a maximizing pure witness $\ket{\psi_\ora}$ in an $a(\lambda)$-qubit
  register $\mathsf W$. We consider the classical--quantum state
  \[
    \rho_{\ora\mathsf W}
      =\E_{\ora\leftarrow\NO_\lambda}\left[
        \ket\ora\!\bra\ora\otimes
        \ket{\psi_\ora}\!\bra{\psi_\ora}\right].
  \]

  \emph{Information about a random target.}
  We first bound the information in $\mathsf W$ about the oracle
  values at a uniformly chosen target with an all-zero suffix. Write
  \[
    \ora_u=(\ora(u\|0^{q-\lambda},v))_{v\in C},
    \qquad
    \ora_{\mathrm{pad}}=(\ora_u)_{u\in\bits^\lambda},
  \]
  and let $\ora_{<u}=(\ora_{u'})_{u'<u}$ in lexicographic order.
  Note that the random variables $\ora_u$ are independent. Hence,
  data processing and the chain rule give
  \begin{equation}
    \label{qma:information}
    \begin{aligned}
      \sum_u I(\ora_u:\mathsf W)
      &\le \sum_u I(\ora_u:\mathsf W\ora_{<u})\\
      &=\sum_u I(\ora_u:\mathsf W\mid\ora_{<u})
        =I(\ora_{\mathrm{pad}}:\mathsf W)\le a(\lambda),
    \end{aligned}
  \end{equation}
  where the last inequality holds because $\ora_{\mathrm{pad}}$
  is classical and $\mathsf W$ has $a(\lambda)$ qubits. This remains
  true even though the witness depends on oracle values at other targets.

  \emph{Bounding the acceptance probability.}
  We now fix $u\in\bits^\lambda$ and compare the actual execution with the
  following modified execution:
  \begin{enumerate}
    \item Sample $\ora\leftarrow\NO_\lambda$, prepare $\ket{\psi_\ora}$, and run
      $v_u\leftarrow\mathsf{CodePreimageDecode}(\ket{\psi_\ora},u\|0^{q-\lambda})$.
    \item Independently sample $\ora'_u$ with the same distribution as
      $\ora_u$. Output $\ACCEPT$ if $v_u\in C$ and $\ora'_u(v_u)=1$,
      and $\REJECT$ otherwise.
  \end{enumerate}
  This changes only the answer to the final query. Since the decoder
  acts on $\mathsf W$ without making oracle queries, data processing
  and Pinsker's inequality give
  \[
    \TV\bigl((\ora_u,v_u),(\ora'_u,v_u)\bigr)
    \le\sqrt{I(\ora_u:v_u)}\le\sqrt{I(\ora_u:\mathsf W)}.
  \]
  Since $\ora'_u$ is independent of $v_u$ and each entry equals one
  with probability $\delta$,
  \[
    \Pr[v_u\in C,\ \ora'_u(v_u)=1]\le\delta.
  \]
  We now average over $u$. The total-variation bound above,
  Cauchy--Schwarz, and \eqref{qma:information} give
  \begin{equation}
    \label{qma:randomaccessbound}
    \begin{aligned}
      \E_{\ora\leftarrow\NO_\lambda}\max_{\ket\psi}\Pr[V^{\ora}(1^\lambda,\ket\psi) = \ACCEPT]
      & = \frac1{2^\lambda}\sum_u\Pr[v_u\in C,\ \ora_u(v_u)=1]\\
      & \le \delta+\frac1{2^\lambda}\sum_u\sqrt{I(\ora_u:\mathsf W)}\\
      & \le \delta+\sqrt{\frac{a(\lambda)}{2^\lambda}}.
    \end{aligned}
  \end{equation}
  The right-hand side is $\gamma_{\NO}^2$. Thus Markov's inequality
  bounds the probability that the maximum acceptance probability
  exceeds $\gamma_{\NO}$ by $\gamma_{\NO}$, proving \eqref{eq:gamman}.
\end{proof}

%% file: sections/qcip-lower-bound.tex
\section{The QCIP Lower Bound}
\label{sec:hardness}
In this section, we prove the QCIP lower bound in \cref{thm:main}.
We first transform a QCIP protocol $(P_0,V_0)$ in two steps:
\[
  (P_0,V_0)\longrightarrow(P_1,V_1)\longrightarrow(P_2,V_2).
\]
The verifier $V_1$ makes only classical queries, with $P_1=P_0$.
The second step gives public coins and perfect completeness, with all
of $V_2$'s queries after the interaction. We then construct a cheating
prover for $V_2$ and use the guarantees of these transformations to
deduce \eqref{eq:main-soundness} for $V_0$.

Since simulation may require superpolynomial local computation,
we use the query model of \cref{def:query-model}. We state the three
steps below and combine them to prove \cref{thm:main}.
We prove the simulation and cheating-prover lemmas in
\cref{sec:simulation,sec:application}, respectively. For public coins
and perfect completeness, we use \cref{lem:public-coin-perfect}.

\paragraph{Step 1: classical-query simulation.}
The following lemma gives a classical-query simulator $V_1$ for $V_0$.
Except on an event of probability $\eta_\lambda$ under either oracle
distribution, its acceptance probability differs from that of $V_0$
by less than $1/100$, simultaneously for every prover.

\begin{lemma}[Classical-query simulation for the oracle distributions]
  \label{est:application}
  Let $\YES_\lambda,\NO_\lambda$ be the distributions of
  \cref{def:oracle-distributions}, and let $V$ be a quantum verifier
  in the query model of \cref{def:query-model}.
  There exist a verifier $V'$ in the same model,
  making only classical queries with the same rounds and message lengths,
  and a sequence $\eta_\lambda=2^{-\Omega(\lambda^2)}$ such that, for all
  sufficiently large $\lambda$ and $\mathcal P\in\{\YES_\lambda,\NO_\lambda\}$,
  \begin{equation}
    \label{est:distributional}
    \Pr_{\ora\leftarrow\mathcal P}
    \left[\forall P^*:\
      \bigl|\Pr[\langle(P^*)^{\ora},V^\ora\rangle=\ACCEPT]
        -\Pr[\langle(P^*)^{\ora},(V')^\ora\rangle=\ACCEPT]\bigr|<1/100\right]
    \ge1-\eta_\lambda.
  \end{equation}
\end{lemma}

\paragraph{Step 2: public coins and perfect completeness.}
The simulator admits a randomized classical implementation
(\cref{est:simulation}). We apply \cref{lem:public-coin-perfect}
to $(P_1,V_1)$ to obtain a public-coin protocol $(P_2,V_2)$,
with deterministic $P_2$ and all verifier queries after the interaction.

\paragraph{Step 3: a cheating prover.}
The cheating prover constructed in \cref{sec:application} gives the
following acceptance guarantee under $\NO_\lambda$.

\begin{lemma}[Public-coin lower bound]
  \label{app:public-coin-lower-bound}
  Let $\YES_\lambda,\NO_\lambda$ be the distributions of
  \cref{def:oracle-distributions}, and let $(P,V)$ be a public-coin
  protocol in the classical query model of \cref{def:query-model}.
  Suppose $V$ makes $N(\lambda)$ oracle
  queries, all after the last prover reply, and set
  $\delta=2^{-\lceil\lambda/12\rceil}$.
  For all sufficiently large $\lambda$, if
  \begin{equation}
    \label{app:public-coin-perfect}
    \Pr_{\ora\leftarrow\YES_\lambda}
      [\Pr[\langle P^\ora,V^\ora\rangle=\ACCEPT]=1]
      \ge\delta^2,
  \end{equation}
  then there is a deterministic prover $\widetilde P$ such that
  \begin{equation}
    \label{app:public-coin-failure}
    \E_{\ora\leftarrow\NO_\lambda}
      \Pr[\langle\widetilde P^\ora,V^\ora\rangle=\ACCEPT]
      \ge1-(N+1)\delta.
  \end{equation}
  All provers have unrestricted oracle access. The threshold for $\lambda$
  depends only on $V$'s number of rounds and queries and total communication.
\end{lemma}

In particular, Markov's inequality applied to the rejection probability
under $\NO_\lambda$ gives
\begin{equation}
  \label{app:public-coin-outer}
  \Pr_{\ora\leftarrow\NO_\lambda}
    [\Pr[\langle\widetilde P^\ora,V^\ora\rangle=\ACCEPT]>1/3]
    \ge1-\tfrac32(N+1)\delta.
\end{equation}

\begin{proof}[Proof of \cref{thm:main}]
  The QMA upper bound follows from \cref{qma:yesprob,qma:noprob}.
  To prove the QCIP lower bound, fix $(P_0,V_0)$ satisfying
  \eqref{eq:main-completeness}, take $\lambda$ sufficiently large,
  and apply the first two steps above.

  \emph{Acceptance under $\YES_\lambda$.}
  Fix an oracle $\ora$ where
  $\Pr[\langle P_0^\ora,V_0^\ora\rangle=\ACCEPT]\ge2/3$
  and the simulation bound holds. By
  \cref{est:application,lem:public-coin-perfect},
  \[
    \begin{aligned}
      \Pr[\langle P_1^\ora,V_1^\ora\rangle=\ACCEPT]&\ge2/3-1/100,\\
      \Pr[\langle P_2^\ora,V_2^\ora\rangle=\ACCEPT]&=1.
    \end{aligned}
  \]
  A union bound under $\YES_\lambda$ gives
  \[
    \Pr_{\ora\leftarrow\YES_\lambda}
      [\Pr[\langle P_2^\ora,V_2^\ora\rangle=\ACCEPT]=1]
      \ge1-\varepsilon-\eta_\lambda.
  \]
  Since $\varepsilon<1/3$, this probability is at least $\delta^2$
  for all sufficiently large $\lambda$.
  Let $N_2$ bound the number of queries made by $V_2$.
  Applying \cref{app:public-coin-lower-bound} and
  \eqref{app:public-coin-outer}, we obtain a prover $\widetilde P$ with
  \[
    \Pr_{\ora\leftarrow\NO_\lambda}
      [\Pr[\langle\widetilde P^\ora,V_2^\ora\rangle=\ACCEPT]>1/3]
    \ge1-\tfrac32(N_2+1)\delta.
  \]
  \emph{Acceptance under $\NO_\lambda$.}
  We now use the two transformations to obtain the required bound for
  $V_0$. Recall that the simulation guarantee holds simultaneously
  for all provers, except with probability $\eta_\lambda$ under
  $\NO_\lambda$. On each oracle $\ora$ where it holds,
  \cref{est:application,lem:public-coin-perfect} imply
  \[
    \begin{aligned}
      &\forall P^*,\ \Pr[\langle(P^*)^\ora,V_0^\ora\rangle=\ACCEPT]\le1/3\\
      &\ \Longrightarrow\
      \forall P^*,\ \Pr[\langle(P^*)^\ora,V_1^\ora\rangle=\ACCEPT]
        \le1/3+1/100\\
      &\ \Longrightarrow\
      \forall P^*,\ \Pr[\langle(P^*)^\ora,V_2^\ora\rangle=\ACCEPT]\le1/3.
    \end{aligned}
  \]
  Taking the contrapositive, some prover makes $V_0$ accept with
  probability greater than $1/3$ on a set of oracles of probability at
  least $1-\tfrac32(N_2+1)\delta-\eta_\lambda$ under $\NO_\lambda$.
  Consequently,
  \[
    \begin{aligned}
      &\Pr_{\ora\leftarrow\NO_\lambda}
        [\forall P^*:\ \Pr[\langle(P^*)^\ora,V_0^\ora\rangle=\ACCEPT]\le1/3]\\
      &\qquad\le\tfrac32(N_2+1)\delta+\eta_\lambda\le\frac23+\varepsilon,
    \end{aligned}
  \]
  proving \eqref{eq:main-soundness}.
\end{proof}

\begin{remark}[Unrestricted computation and classical advice]
  \label{rem:query-model-hardness}
  The QCIP lower bound in \cref{thm:main} also holds when the quantum
  verifier has unrestricted local computation and both parties receive
  classical advice of arbitrary length, fixed independently of the
  sampled oracle input (\cref{def:query-model}). For fixed bounds on
  rounds, communication, and queries, the required lower bound on
  $\lambda$ is independent of the advice string.
  We use this in \cref{app:auxiliary}: the already fixed oracle values
  serve as advice when choosing the next block.
\end{remark}

\input{sections/qcip/simulation}

\input{sections/qcip/public-coin-lower-bound}

%% file: sections/qcip/simulation.tex
\subsection{From QCIP to Classical Queries}
\label{sec:simulation}
In this section, we prove that a quantum verifier exchanging classical
messages can be simulated with classical oracle queries. We first give
a general theorem whose error bound is expressed through an
incompressibility experiment. We then apply the two bounds from
\cref{sec:codes} to prove \cref{est:application}.

\subsubsection{Simulation with Classical Queries and Incompressibility}
\label{sec:general-simulation}
We consider an oracle $\ora_F$ for membership in an unknown set $F$.
Following the interactive estimation method of~\cite{QFire}, the
simulator prepares the state before each verifier query and measures
its input register to obtain an address $z$. If $z$ is not already in
the database, it queries $\ora_F(z)$ classically and adds $z$ if the answer
is one. We call this an \emph{estimation trial}.

To bound the number of trials, we will show that the database can be
recovered from a short classical string: the transcript and the indices
of trials that added an address. This allows us to apply incompressibility
even though the prover has unrestricted oracle access.

\begin{theorem}[Simulation with classical queries and incompressibility]
  \label{est:simulation}
  Let $\mathcal X$ be a finite oracle domain. For $F\subseteq\mathcal X$,
  let $\ora_F$ be its membership oracle.
  Let $V$ be a quantum verifier exchanging classical messages, with
  the circuit conventions of \cref{def:query-model}. Fix a common input,
  and let $r\ge1$ be its number of rounds, $N\ge1$ its number of queries
  per round, and $c$ its transcript length, including the final decision.
  We take these counts to be fixed by padding. Its queries have length
  at most $\ell$.\footnote{For the classical implementation, we use a
  standard universal gate set with algebraic entries, such as Hadamard,
  $T$, and CNOT~\cite{BMPRV00}.}

  For every integer $t\ge1$ and $0<\epsilon<1$, there are a verifier $V'$
  making only classical oracle queries and a randomized decoder $\mathsf{Dec}$
  that makes no oracle queries, both depending only on $V,t,\epsilon$.
  The decoder takes a message $m=(\pi,I)\in\bits^s$, where
  $\pi\in\bits^c$, $I\in\binom{[M]}t$, and
  \[
    s=c+t\lceil\log_2M\rceil,
    \qquad M=\poly(r,N,t,\ell,1/\epsilon).
  \]
  We write $\mathsf{Dec}(\pi,I)$ for this decoder, representing $I$ by
  its $t$ indices in increasing order.
  Write $\omega$ for the decoder's random tape, sampled
  independently of $F$. Define
  \begin{equation}
    g(F):=\Pr_\omega\!\left[
      \exists\pi\in\bits^c,\ I\in\binom{[M]}t:\ \mathsf{Dec}(\pi,I;\omega)\in\binom Ft
    \right].\label{est:gdefinition}
  \end{equation}
  Here the encoder knows both $F$ and the shared randomness $\omega$,
  so it may choose its message $(\pi,I)$ after seeing both. The same
  $\omega$ is used to decode every candidate message.
  For every $F\subseteq\mathcal X$ and every prover $P^*$,
  \begin{equation}
    \begin{aligned}
      &\bigl|\Pr[\langle (P^*)^{\ora_F},V^{\ora_F}\rangle=\ACCEPT]
        -\Pr[\langle (P^*)^{\ora_F},(V')^{\ora_F}\rangle=\ACCEPT]\bigr|\\
      &\hspace{5cm}\le\epsilon+g(F).
    \end{aligned}\label{est:uniform}
  \end{equation}
  The verifier $V'$ has the same rounds and message lengths as $V$,
  and makes at most $M$ classical queries. Moreover, $V'$ admits a
  randomized classical implementation using
  \[
    O\!\left((M+r)\cdot\left(\ell+c+\log\frac{M+r}{\epsilon}\right)\right)
  \]
  random bits. Local computation and workspace are unrestricted.
\end{theorem}

We prove the theorem in the following order.
\begin{enumerate}
  \item We construct a simulator $\widehat V_\infty$ without a database
    cutoff and bound the change in acceptance probability by $\epsilon/3$.
  \item We stop adding addresses once the database reaches size $t$,
    obtaining $\widehat V$ with at most $M$ queries. We show that
    $\mathsf{Dec}$ recovers the database whenever this occurs, so the
    probability of reaching the cutoff is at most $g(F)$.
  \item Since the simulators agree until the cutoff, the acceptance
    probabilities of $V$ and $\widehat V$ differ by at most
    $\epsilon/3+g(F)$. Finally, we obtain $V'$ by classical sampling
    with bounded randomness.
\end{enumerate}

\subsubsection{The Simulator without a Database Cutoff}
We first explain how to prepare the verifier's private state from its
classical interaction. By \cref{lem:classical-prover}, we may use the
same classical randomized replacement for $P^*$ against $V$, the
simulator, and all the hybrids below. This replacement receives the
full transcript and retains no private state between replies.

\paragraph{The verifier in each round.}
Let $\mathsf R$ be the verifier's private register, initially in $\ket0$.
In round $i$, it receives $m_i$ and runs $V[i]^\ora(\mathsf R,m_i)$.
This algorithm updates $\mathsf R$, makes $N$ queries, and writes its
outgoing message in $\mathsf M$. The verifier measures $\mathsf M$ to
obtain $m'_i$ and retains $\mathsf R$ for the next round; $m'_r$ is its
final decision. If $V$ speaks first, $m_1$ is empty. We use the
purification convention of \cref{sec:prelims}, so all private working
registers are included in $\mathsf R$.
If $V$ rejects early, it sends fixed messages in the remaining rounds
and outputs $\REJECT$ at the end.

The simulator will prepare the state before each query by repeating
these computations. To do so, it must use the same database at each
earlier query, so it saves the databases together with the transcript.

\begin{definition}[Simulation history]
  \label{est:state}
  The simulator keeps a database $D_{i,j}$ for query $j$ of round $i$,
  initialized as a copy of $D_{i,j-1}$ before estimating that query.
  It updates only this database during estimation, then keeps it fixed
  for the query and all later state preparations. We set
  $D_{0,N}=\varnothing$ and $D_{i,0}=D_{i-1,N}$.

  After round $i$, the transcript and the \emph{simulation history} are,
  respectively,
  \begin{equation}
    \label{est:conditionalstate}
    \begin{aligned}
      \pi_i&=(m_1,m'_1,\ldots,m_i,m'_i),\\
      T_i&=\bigl((m_s,m'_s,(D_{s,k})_{k=1}^N)\bigr)_{s=1}^i.
    \end{aligned}
  \end{equation}
  Fix $T_i$ and run rounds $1,\ldots,i$ from the initial state,
  supplying each recorded message $m_s$ and answering query $j$ in
  round $s$ with $\ora_{D_{s,j}}$. We denote the round-$s$ circuit
  with these answers by $V[s]^{(D_{s,k})_{k=1}^N}$. Write
  $\widehat m_1,\ldots,\widehat m_i$ for the measured messages of this
  execution, and put
  \[
    p(T_i):=\Pr[\widehat m_1=m'_1,\ldots,\widehat m_i=m'_i],
  \]
  where only the verifier's measurements are random.

  We call $T_i$ \emph{possible} if $p(T_i)>0$, and set
  $\pi_0=T_0=\varnothing$ and $p(T_0)=1$.
\end{definition}

We now define $\mathsf{Prep}$, which repeats this execution until its
outgoing messages agree with the transcript.

\begin{protocol}
  \protocoltitle{$\mathsf{Prep}(T_i)$}
  \textbf{Input:} A possible simulation history
  $T_i=\bigl((m_s,m'_s,(D_{s,k})_{k=1}^N)\bigr)_{s=1}^i$.
  \par
  \textbf{Output:} The private register $\mathsf R$ conditioned on the
  verifier messages specified in $T_i$.
  \protocolbody
  \begin{enumerate}
    \item
      Initialize a fresh register $\mathsf R$ to $\ket{0}$.
    \item
      For $s=1,\ldots,i$:
      \begin{enumerate}[label*=\arabic*.,leftmargin=*]
        \item
          Run $V[s]^{(D_{s,k})_{k=1}^N}$ on
          $(\mathsf R,m_s)$, writing its message in $\mathsf M$
          and retaining its private state in $\mathsf R$.
        \item
          Measure $\mathsf M$ in the computational basis to obtain
          $\widehat m_s$.
        \item
          If $\widehat m_s\ne m'_s$, discard
          $\mathsf M\mathsf R$ and restart at step 1.
          Otherwise discard $\mathsf M$ and continue.
      \end{enumerate}
    \item
      Return $\mathsf R$.
  \end{enumerate}
\end{protocol}

Observe that a successful attempt prepares the verifier's state
conditioned on the messages in $T_i$. Hence, if $p(T_i)>0$,
$\mathsf{Prep}(T_i)$ terminates with probability one and makes no oracle
queries. We also define its output to be $\ket0$ when $p(T_i)=0$;
this convention will be used for histories that cannot occur in the
simulated interaction.\footnote{The verifier's circuit and databases
determine $p(T_i)$, so the classical implementation below can test
whether it is zero.}

\paragraph{The state before a query.}
Fix the saved history $T_{i-1}$, incoming message $m_i$, and
databases $D_{i,1},\ldots,D_{i,j-1}$. We run
$\mathsf{Prep}(T_{i-1})$ to prepare a fresh register $\mathsf R$, then
run $V[i](\mathsf R,m_i)$, stopping immediately before query $j$ and using
$\ora_{D_{i,k}}$ at each query $k<j$.
We denote the resulting state by $\ket{\psi_{i,j}}$ and write
$\mathsf R=\mathsf X\mathsf A\mathsf W$, where
$\mathsf X$ is the query-input register, $\mathsf A$ is the answer register,
and $\mathsf W$ contains the remaining private qubits.
Writing $z$ for a query address, $o\in\bits$ for the answer bit, and
$w$ for a workspace basis label, we have
\begin{equation}
  \label{est:querystate}
  \ket{\psi_{i,j}}
  =\sum_{z,o,w}\alpha_{i,j,z,o,w}\ket{z}_{\mathsf X}
    \ket{o}_{\mathsf A}\ket{w}_{\mathsf W},
  \qquad \sum_{z,o,w}|\alpha_{i,j,z,o,w}|^2=1.
\end{equation}
Measuring $\mathsf X$ gives $z$ with probability
$\sum_{o,w}|\alpha_{i,j,z,o,w}|^2$. Since preparation uses only earlier
databases, additions to $D_{i,j}$ do not change this distribution.

\paragraph{The simulator.}
We can now give the simulator. It stores $T_i$ and calls $\mathsf{Prep}$
whenever it needs the verifier's private state.
Let $N_{\mathrm{tot}}=rN$ and set
\begin{equation}
  \label{est:parameters}
  \theta=\left(\frac{\epsilon}{12N_{\mathrm{tot}}}\right)^2,
  \qquad
  K_0=\left\lceil\theta^{-1}
    \left(2(\ell+1)+\ln\frac{6N_{\mathrm{tot}}}\epsilon\right)\right\rceil.
\end{equation}
Here $|\mathcal X|\le2^{\ell+1}$, and
$2\ell/\theta\le K_0=\poly(r,N,\ell,1/\epsilon)$.
For each query, estimation ends after $K_0$ consecutive trials without
an addition. We place no restriction on the database size yet.

\begin{protocol}
  \protocoltitle{$\widehat V_\infty^{\ora_F}$}
  \textbf{Oracle:} $\ora_F$.
  \protocolbody
  Initialize $D_{0,N}=\varnothing$ and history $T_0=\varnothing$.
  In round $i$, receive $m_i$, set $D_{i,0}\gets D_{i-1,N}$, and proceed
  as follows:
  \begin{enumerate}
    \item
      For $j=1,\ldots,N$:
      \begin{enumerate}[label*=\arabic*.,leftmargin=*]
        \item
          Copy $D_{i,j}\gets D_{i,j-1}$ and set $k=0$.
        \item
          While $k<K_0$:
          \begin{enumerate}[label*=\arabic*.,leftmargin=*]
            \item
              Run $\mathsf{Prep}(T_{i-1})$ to prepare a fresh register
              $\mathsf R$.
              Run $V[i]$ on $(\mathsf R,m_i)$, stopping before query $j$
              and answering each earlier query $u<j$ with $\ora_{D_{i,u}}$.
              Measure $\mathsf X$ in the resulting state
              $\ket{\psi_{i,j}}$ to obtain $z$, then discard $\mathsf R$.
            \item
              If $z\in D_{i,j}$, set $k\gets k+1$.
              Otherwise query $b\gets\ora_F(z)$ and set
              \[
                (D_{i,j},k)\gets
                \begin{cases}
                  (D_{i,j}\cup\{z\},0),&b=1,\\
                  (D_{i,j},k+1),&b=0.
                \end{cases}
              \]
          \end{enumerate}
        \item
          Keep $D_{i,j}$ fixed from this point onward.
      \end{enumerate}
    \item
      Run $\mathsf{Prep}(T_{i-1})$ to prepare a fresh register
      $\mathsf R$, then run
      $V[i]^{(D_{i,k})_{k=1}^N}(\mathsf R,m_i)$, writing the outgoing
      message in $\mathsf M$.
      Measure $\mathsf M$ to obtain $m'_i$, then discard $\mathsf R\mathsf M$.
    \item
      Append $(m_i,m'_i,(D_{i,k})_{k=1}^N)$ to $T_{i-1}$ to obtain
      $T_i$, and send $m'_i$;
      in the last round, output the decision $m'_r$.
  \end{enumerate}
\end{protocol}

Observe that $D_{i,j}\subseteq F$ throughout the execution and each
address is added at most once. There are therefore at most
$|\mathcal X|\le2^{\ell+1}$ additions in total. Divide estimation into
intervals starting at $k=0$ and ending at an addition or at $K_0$
consecutive trials without one. Each contains at most $K_0$ trials,
with one final interval per query, so
\[
  \#\text{classical queries}\le\#\text{trials}
    \le(N_{\mathrm{tot}}+2^{\ell+1})\cdot K_0.
\]
Thus $\widehat V_\infty$ finishes with probability one, but may make
exponentially many queries in $\ell$.

\subsubsection{Accuracy without a Database Cutoff}
We will show that $\widehat V_\infty$ changes the acceptance probability
by at most $\epsilon/3$. We first bound the query weight on addresses
in $F$ that have not been added to the database, then use this bound
to compare consecutive hybrids.

\begin{lemma}[Estimating one query]
  \label{est:onequery}
  Fix $F\subseteq\mathcal X$ and all messages, measured query addresses,
  and oracle answers preceding the estimation of query $j$ in round $i$
  of $\widehat V_\infty$. Initialize
  $D_{i,j}$ as a copy of $D_{i,j-1}\subseteq F$.
  Let $\ket{\psi_{i,j}}$ be as in \eqref{est:querystate}, let $\theta$
  be as in \eqref{est:parameters}, and define
  \begin{equation}
    \label{est:missingweight}
    p_{i,j}(D_{i,j})
    :=\sum_{z\in F\setminus D_{i,j}}\sum_{o,w}|\alpha_{i,j,z,o,w}|^2.
  \end{equation}
  Let $\mathcal B_{i,j}$ be the event that estimation finishes with
  $p_{i,j}(D_{i,j})>\theta$. Then
  \[
    \Pr[\mathcal B_{i,j}]\le\frac{\epsilon}{6N_{\mathrm{tot}}}.
  \]
  The probability is over the fresh samples used to estimate this query.
  Let $U_F,U_{D_{i,j}}$ denote the query gates of $\ora_F,\ora_{D_{i,j}}$.
  Outside this event, the final database $D_{i,j}$ satisfies
  \[
    \bigl\|(U_F-U_{D_{i,j}})\ket{\psi_{i,j}}\bigr\|_2^2\le4\theta.
  \]
\end{lemma}
\begin{proof}
  Recall that $\ket{\psi_{i,j}}$ does not change during estimation
  of query $j$.
  Condition on the trial outcomes preceding one estimation interval.
  Until the next addition, each trial independently adds an address
  with probability $p_{i,j}(D_{i,j})$. We consider two cases.

  \emph{Case 1: $p_{i,j}(D_{i,j})>\theta$.}
  The interval ends without an addition only if all $K_0$ trials
  fail to add an address. Thus
  \begin{equation}
    \label{est:misses}
    \Pr[\text{interval ends without an addition}]
    =(1-p_{i,j}(D_{i,j}))^{K_0}
    \le 2^{-K_0\theta/\ln2}
    \le\frac{\epsilon}{6N_{\mathrm{tot}}(2^{\ell+1}+1)}.
  \end{equation}
  \emph{Case 2: $p_{i,j}(D_{i,j})\le\theta$.}
  Further additions can only decrease $p_{i,j}(D_{i,j})$, so the
  final database also satisfies this inequality. The BBBV query
  bound~\cite[Section~3]{bbbv97} then gives
  \[
    \bigl\|(U_F-U_{D_{i,j}})\ket{\psi_{i,j}}\bigr\|_2^2
    \le4\sum_{z\in F\setminus D_{i,j}}\sum_{o,w}|\alpha_{i,j,z,o,w}|^2
    =4p_{i,j}(D_{i,j})\le4\theta.
  \]

  \emph{Combining the cases.}
  The event $\mathcal B_{i,j}$ requires an interval in Case 1 to end
  without an addition. There are at most $2^{\ell+1}+1$ intervals,
  so a union bound gives
  \[
    \Pr[\mathcal B_{i,j}]
      \le(2^{\ell+1}+1)\cdot
        \frac{\epsilon}{6N_{\mathrm{tot}}(2^{\ell+1}+1)}
      =\frac{\epsilon}{6N_{\mathrm{tot}}}.
  \]
  Outside $\mathcal B_{i,j}$, Case 2 gives
  $\|(U_F-U_{D_{i,j}})\ket{\psi_{i,j}}\|_2^2\le4\theta$.
\end{proof}

\paragraph{Hybrids.}
Fix $F\subseteq\mathcal X$ and a prover $P^*$. We replace the
verifier's queries one at a time, while $P^*$ retains access to $\ora_F$.
For $0\le s\le N_{\mathrm{tot}}$, define $\mathsf{Hyb}_s$ to use
the saved database oracles for the first $s$ queries and $\ora_F$
for the remaining queries. Estimation uses fresh samples and does
not act on the verifier's private register $\mathsf R$.

\begin{protocol}
  \protocoltitle{$\mathsf{Hyb}_s^{\ora_F}$}
  \textbf{Oracle:} $\ora_F$.
  \protocolbody
  Initialize $\mathsf R$ to $\ket0$, $D_{0,N}=\varnothing$, and $T_0=\varnothing$.
  In round $i$, receive $m_i$ and set $D_{i,0}\gets D_{i-1,N}$.
  Run $V[i]$ on $(\mathsf R,m_i)$, applying the following rules
  at each query $j$ and at the end of the round.
  \begin{enumerate}
    \item Immediately before query $j$ of round $i$, copy
      $D_{i,j}\gets D_{i,j-1}$ and run the estimation loop of
      $\widehat V_\infty$ using fresh samples. Then keep $D_{i,j}$ fixed.
    \item Apply the query to $\mathsf R$ using
      \[
        \begin{cases}
          U_{D_{i,j}},&(i-1)N+j\le s,\\
          U_F,&(i-1)N+j>s.
        \end{cases}
      \]
    \item At the end of the round, measure the verifier's outgoing
      register $\mathsf M$ to obtain $m'_i$. Append
      $(m_i,m'_i,(D_{i,k})_{k=1}^N)$ to $T_{i-1}$ to obtain $T_i$,
      and send $m'_i$;
      in the last round, output $m'_r$.
  \end{enumerate}
\end{protocol}

Note that estimation uses the saved databases in every hybrid, with
the convention for $\mathsf{Prep}$ on impossible histories. At the query
where consecutive hybrids differ, all earlier verifier queries also
used the saved databases. We next show that the verifier is then in
state $\ket{\psi_{i,j}}$, conditioned on the messages and estimation
outcomes so far. This will allow us to apply \cref{est:onequery}.

\begin{lemma}[Conditioning on the estimation outcomes]
  \label{est:refined-history}
  Consider $\mathsf{Hyb}_s$ just before estimating query $j$ of round $i$,
  where $(i-1)N+j\le s+1$. Let $A$ be the next trial index and record
  the messages and preceding estimation outcomes as
  \[
    \tau=(\pi_{i-1},m_i,(z_a,b_a)_{a<A}),
  \]
  where $z_a$ is the address sampled in trial $a$ and
  $b_a=\ora_F(z_a)$.
  Conditioned on any value of $\tau$ with positive probability,
  the verifier's register $\mathsf R$ is in state $\ket{\psi_{i,j}}$.
  This remains true after conditioning on further estimation samples
  before applying the query.
\end{lemma}
\begin{proof}
  The messages and trial outcomes in $\tau$ determine all saved
  databases. All preceding verifier queries used these databases,
  so we prove the statement by induction over the preceding operations.
  \begin{enumerate}
    \item \emph{Verifier operations.}
      The verifier applies the same gates as the preparation procedure.
    \item \emph{Estimation trials.}
      Fix the messages and outcomes so far, denoted by $\tau'$.
      Conditioned on $\tau'$, the fresh estimation registers are
      independent of $\mathsf R$, and
      $b_a=\ora_F(z_a)$. Hence, for every positive-probability extension,
      \[
        \rho^{\mathsf R}_{\tau',z_a,b_a}=\rho^{\mathsf R}_{\tau'},
      \]
      where $\rho^{\mathsf R}_{\tau'}$ is the verifier's conditional state.
    \item \emph{Classical messages.}
      Conditioning on $m'_u$ projects $\mathsf M$ onto $\ket{m'_u}$,
      as in $\mathsf{Prep}$.
      By \cref{lem:classical-prover}, the distribution of a prover
      reply depends only on $P^*,F$ and the transcript, so conditioning
      on the reply leaves $\mathsf R$ unchanged.
  \end{enumerate}
  Thus $p(T_{i-1})>0$ and the conditional state is $\ket{\psi_{i,j}}$,
  even after further estimation samples.
\end{proof}

\begin{lemma}[One hybrid step]
  \label{est:hybrid-step}
  For every $s\in[N_{\mathrm{tot}}]$,
  \[
    \begin{aligned}
      &\bigl|\Pr[\langle (P^*)^{\ora_F},\mathsf{Hyb}_{s-1}^{\ora_F}\rangle=\ACCEPT]
        -\Pr[\langle (P^*)^{\ora_F},\mathsf{Hyb}_s^{\ora_F}\rangle=\ACCEPT]\bigr|\\
      &\hspace{5cm}\le\frac{\epsilon}{3N_{\mathrm{tot}}}.
    \end{aligned}
  \]
\end{lemma}
\begin{proof}
  Fix $s=(i-1)N+j$ and condition on the same $\tau$ in both hybrids before
  estimating this query. By \cref{est:refined-history}, every
  positive-probability estimation outcome leaves both verifiers in
  state $\ket{\psi_{i,j}}$. Their states after the query are
  \[
    \ket{\phi_F}=U_F\ket{\psi_{i,j}},
    \qquad
    \ket{\phi_{D_{i,j}}}=U_{D_{i,j}}\ket{\psi_{i,j}}.
  \]
  We consider two cases:
  \begin{enumerate}
    \item $p_{i,j}(D_{i,j})\le\theta$. By \cref{est:onequery},
      \[
        \tfrac12\bigl\|\ket{\phi_F}\!\bra{\phi_F}
          -\ket{\phi_{D_{i,j}}}\!\bra{\phi_{D_{i,j}}}\bigr\|_1
        \le\|\ket{\phi_F}-\ket{\phi_{D_{i,j}}}\|_2\le2\sqrt\theta.
      \]
      Both hybrids then run the same continuation, including the
      transcript-dependent replies of $P^*$. Hence their acceptance
      probabilities differ by at most $2\sqrt\theta$.
    \item $p_{i,j}(D_{i,j})>\theta$. This is $\mathcal B_{i,j}$;
      we bound the difference by one.
  \end{enumerate}
  Averaging over the estimation outcomes and using \cref{est:onequery}
  bounds the difference by
  \[
    2\sqrt\theta+\Pr[\mathcal B_{i,j}]
    \le\frac{\epsilon}{6N_{\mathrm{tot}}}
       +\frac{\epsilon}{6N_{\mathrm{tot}}}
    =\frac{\epsilon}{3N_{\mathrm{tot}}}.
  \]
  Since this holds for every $\tau$ with positive probability,
  averaging over $\tau$ proves the lemma.
\end{proof}

\begin{lemma}[Accuracy without a database cutoff]
  \label{est:uncapped}
  For every $F\subseteq\mathcal X$ and every prover $P^*$,
  \[
    \begin{aligned}
      &\bigl|\Pr[\langle (P^*)^{\ora_F},V^{\ora_F}\rangle=\ACCEPT]
        -\Pr[\langle (P^*)^{\ora_F},\widehat V_\infty^{\ora_F}\rangle=\ACCEPT]\bigr|\\
      &\hspace{5cm}\le\epsilon/3.
    \end{aligned}
  \]
\end{lemma}
\begin{proof}
  We first identify the endpoints of the hybrid sequence:
  \begin{enumerate}
    \item $\mathsf{Hyb}_0$ has the same acceptance probability as $V$.
      Estimation finishes with probability one and uses fresh registers,
      leaving the verifier's state and messages unchanged.
    \item $\mathsf{Hyb}_{N_{\mathrm{tot}}}$ has the same acceptance
      probability as $\widehat V_\infty$.
      By \cref{est:refined-history}, retaining $\mathsf R$ between rounds
      or preparing it from the saved history gives the same distribution
      of messages and estimation outcomes.
  \end{enumerate}
  By \cref{est:hybrid-step}, the $N$ replacements in each round
  contribute at most $\epsilon/(3r)$. Summing over rounds gives
  $r\cdot\epsilon/(3r)=\epsilon/3$.
\end{proof}

\subsubsection{Bounding the Cutoff Probability}
We next bound the number of queries. Let $\widehat V$ run
$\widehat V_\infty$ until some database $D_{i,j}$ reaches size $t$.
If this happens, it stops estimation, sends fixed messages of the
prescribed lengths in the current and remaining rounds, and finally
outputs $\REJECT$ and terminates.

Since at most $t$ addresses are added before the simulator stops,
the interval count used for $\widehat V_\infty$ gives
\begin{equation}
  \label{est:querybound}
  \#\text{classical queries}\le\#\text{trials}
  \le M:=(N_{\mathrm{tot}}+t)\cdot K_0
  =\poly(r,N,t,\ell,1/\epsilon).
\end{equation}
It remains to bound the probability that the cutoff is reached.
For this, we encode the execution by its transcript and the indices
of trials that added an address. We will show that a decoder can
recover all $t$ addresses from this string without querying $\ora_F$.

\paragraph{Encoding an execution that reaches the cutoff.}
Fix $P^*$ and give the encoder $F$ and independent random tapes
$\omega=(\omega_1,\ldots,\omega_M)$ shared with the decoder
(one per trial).\footnote{
We realize each measurement as a sample from its conditional
outcome distribution; the exact classical implementation is given below.}
The encoder runs $\langle(P^*)^{\ora_F},\widehat V^{\ora_F}\rangle$
using $\omega$ for estimation. If it reaches the cutoff, it sends
\[
  m=(\pi,I),\qquad I=\{a_1<\cdots<a_t\}\subseteq[M],
  \qquad |m|=c+t\lceil\log_2M\rceil=s,
\]
where $\pi$ is the full $c$-bit transcript, including the messages after
the cutoff and the final decision, and $a_1,\ldots,a_t$ are the addition
indices. Since $\pi$ specifies all messages, the encoder need not share
the randomness used to generate them.

Whenever the simulator makes a query, its address lies outside
$D_{i,j}$. The answer is therefore one exactly at the indices in $I$,
which lets the decoder replay the estimation trials.

\begin{protocol}
  \protocoltitle{$\mathsf{Dec}(\pi,I)$}
  \textbf{Input:} A transcript $\pi\in\bits^c$ and trial indices
  $I\in\binom{[M]}t$.
  \par
  \textbf{Shared randomness:} $\omega=(\omega_1,\ldots,\omega_M)$,
  sampled independently of $F$.
  \par
  \textbf{Output:} A set of $t$ distinct oracle addresses, or $\bot$.
  \protocolbody
  \begin{enumerate}
    \item
      Parse $\pi=(m_1,m'_1,\ldots,m_r,m'_r)$.
    \item
      Replay $\widehat V$ using the messages in $\pi$, without interacting
      or sampling outgoing messages. Number estimation trials consecutively.
      For each trial $a$, with current round and query indices $i,j$:
      \begin{enumerate}[label*=\arabic*.,leftmargin=*]
        \item Sample its address $z_a$ using $\omega_a$.
        \item If $z_a\notin D_{i,j}$, supply $b_a=1$ if
          $a\in I$, and $b_a=0$ otherwise. If
          $z_a\in D_{i,j}$ and $a\in I$, output $\bot$
          and terminate. Otherwise complete the trial's database and
          counter updates as in $\widehat V$.
        \item If $|D_{i,j}|=t$, output $D_{i,j}$ and terminate.
      \end{enumerate}
    \item
      If the simulation ends before the database reaches size $t$,
      output $\bot$ and terminate.
  \end{enumerate}
\end{protocol}

\begin{lemma}[Incompressibility reduction]
  \label{est:cap}
  For every $F\subseteq\mathcal X$ and every prover $P^*$, the
  interaction $\langle (P^*)^{\ora_F},\widehat V^{\ora_F}\rangle$ satisfies
  \[
    \Pr[\text{some }D_{i,j}\text{ reaches size }t]\le g(F).
  \]
  The event ranges over all estimation trials in all rounds.
\end{lemma}
\begin{proof}
  Use the shared randomness $\omega$ for estimation in $\widehat V$ and
  every $\mathsf{Dec}(\pi,I)$, with separate randomness for outgoing messages.
  Suppose the cutoff is reached, and choose $(\pi,I)$ as above.
  By induction on trials, $\mathsf{Dec}$ and $\widehat V$ have the same databases
  and messages, so $\omega_a$ gives the same address $z_a$. At every
  actual query, $\ora_F(z_a)=1$ if and only if $a\in I$, so $\mathsf{Dec}$
  supplies the correct answer. Hence
  \[
    \{\text{some }D_{i,j}\text{ reaches size }t\}
      \subseteq\bigcup_{\pi\in\bits^c}\ \bigcup_{I\in\binom{[M]}t}
        \left\{\mathsf{Dec}(\pi,I;\omega)\in\binom Ft\right\}.
  \]
  Taking probabilities bounds the cutoff event by $g(F)$, as defined
  in \eqref{est:gdefinition}. Since the union includes every transcript,
  this bound holds for every $P^*$, including provers whose replies
  depend on $F$.
\end{proof}

\subsubsection{Accuracy with a Database Cutoff}
We now compare $\widehat V$ with $\widehat V_\infty$. Fix $F,P^*$
and use the same estimation samples and message randomness in both
executions until the cutoff. The executions agree up to that point,
so their acceptance probabilities can differ only if the cutoff is
reached. Therefore,
\begin{equation}
  \label{est:cutoff-error}
  \begin{aligned}
    &\bigl|\Pr[\langle (P^*)^{\ora_F},\widehat V_\infty^{\ora_F}\rangle=\ACCEPT]
      -\Pr[\langle (P^*)^{\ora_F},\widehat V^{\ora_F}\rangle=\ACCEPT]\bigr|\\
    &\hspace{3cm}\le
      \Pr[\text{some }D_{i,j}\text{ reaches size }t
        \text{ in }\langle (P^*)^{\ora_F},\widehat V^{\ora_F}\rangle].
  \end{aligned}
\end{equation}
Combining \eqref{est:cutoff-error} with \cref{est:uncapped,est:cap}
proves the following lemma.

\begin{lemma}[Simulation with exact sampling]
  \label{est:exact}
  For the simulator $\widehat V$ above and $g$ defined in
  \eqref{est:gdefinition}, every $F\subseteq\mathcal X$ and every prover $P^*$ satisfy
  \[
    \begin{aligned}
      &\bigl|\Pr[\langle (P^*)^{\ora_F},V^{\ora_F}\rangle=\ACCEPT]
        -\Pr[\langle (P^*)^{\ora_F},\widehat V^{\ora_F}\rangle=\ACCEPT]\bigr|\\
      &\hspace{5cm}\le\epsilon/3+g(F).
    \end{aligned}
  \]
\end{lemma}
\begin{proof}
  Apply the triangle inequality to \cref{est:uncapped},
  \eqref{est:cutoff-error}, and \cref{est:cap}.
\end{proof}

\paragraph{Classical replacement.}
It remains to obtain a classical verifier $V'$ with bounded
randomness. Since local computation is unrestricted, we can compute
the sampling probabilities of $\widehat V$ classically. We then
round these probabilities, as follows.

\begin{proof}[Proof of \cref{est:simulation}]
  There are at most $M+r$ address and outgoing-message samples,
  each with at most $L$ outcomes, where
  \[
    L=\max\{2^{\ell+1},2^c\},\qquad
    k_{\mathrm{rand}}=\left\lceil\log_2\frac{3(M+r)\cdot L}{\epsilon}\right\rceil.
  \]
  Define $V'$ to follow $\widehat V$, replacing each sampled
  distribution $(p_1,\ldots,p_d)$, $d\le L$, by
  \[
    q_u=2^{-k_{\mathrm{rand}}}\cdot\lfloor2^{k_{\mathrm{rand}}}\cdot p_u\rfloor
      \quad(u<d),
    \qquad q_d=1-\sum_{u<d}q_u.
  \]
  Sampling from $(q_1,\ldots,q_d)$ uses $k_{\mathrm{rand}}$ uniform bits
  and changes the distribution by total variation distance at most
  $L\cdot2^{-k_{\mathrm{rand}}}$. This holds for every history, using the
  specified default state for $\mathsf{Prep}$ when necessary.
  Replacing the samples one at a time gives
  \[
    \begin{aligned}
      &\left|\Pr[\langle (P^*)^{\ora_F},(V')^{\ora_F}\rangle=\ACCEPT]
        -\Pr[\langle (P^*)^{\ora_F},\widehat V^{\ora_F}\rangle=\ACCEPT]\right|\\
      &\hspace{3cm}\le (M+r)\cdot L\cdot2^{-k_{\mathrm{rand}}}
        \le\epsilon/3.
    \end{aligned}
  \]
  The rounded probabilities are computable under our gate convention, so
  $V'$ uses at most $M$ oracle queries and $(M+r)\cdot k_{\mathrm{rand}}$
  random bits. Combining the last bound with \cref{est:exact} gives
  \eqref{est:uniform}. The decoder $\mathsf{Dec}$ uses the exact
  probabilities instead; it samples classically using its random
  tape and terminates with probability one.
  Neither construction depends on $F$ or the prover.
\end{proof}

\subsubsection{Applying Incompressibility}
We now apply \cref{est:simulation} to the distributions from
\cref{def:oracle-distributions}. The only remaining task is to bound
$g(F_\ora)$, where $F_\ora=\{z:\ora(z)=1\}$. We use code-preimage
incompressibility under $\YES_\lambda$ and independent-marking
incompressibility under $\NO_\lambda$.

\begin{proof}[Proof of \cref{est:application}]
  On input $1^\lambda$, let $r(\lambda),N(\lambda),c(\lambda)$
  be the verifier's number of rounds, queries per round, and transcript
  length including the final decision, respectively.
  By \cref{lem:code}, queries $(y,v)\in\mathcal Z_\lambda$ have length
  $O(\lambda^6\log\lambda)$, so we use the upper bound
  $\ell(\lambda)=\lambda^8$ for sufficiently large $\lambda$.
  We apply \cref{est:simulation} with
  \[
    \epsilon=10^{-4},\qquad t(\lambda)=c(\lambda)+\lambda.
  \]
  The choice of $K_0(\lambda)$ in \eqref{est:parameters} gives
  \[
    M(\lambda)=\bigl(r(\lambda)N(\lambda)+t(\lambda)\bigr)K_0(\lambda)
    =\poly(\lambda).
  \]
  For the rest of the proof, suppress $\lambda$ in $c,t,M$ and put
  $s=c+t\lceil\log_2M\rceil$. The construction of $V'$, $\mathsf{Dec}$, and $g$
  is the same for both oracle distributions. Write
  \[
    \mathcal X=\mathcal Z_\lambda=\bits^q\times C,\qquad
    F_\ora=\{(y,v)\in\mathcal Z_\lambda:\ora(y,v)=1\},\qquad
    \beta=2^s\delta^t,
  \]
  where $\delta=2^{-\lceil\lambda/12\rceil}$.
  For sufficiently large $\lambda$, $t\le2^\lambda$. We now apply
  the two incompressibility theorems. In each experiment, the encoder
  uses the sampled oracle and shared tape $\omega$ to choose a message
  decoding to $t$ pairs in $F_\ora$, whenever such a message exists.
  The probability that its message decodes to $t$ pairs in $F_\ora$
  is exactly $g(F_\ora)$.
  \begin{itemize}
    \item \emph{Under $\YES_\lambda$.}
      The oracle $\ora$ is determined by the sampled hashes $H$, and
      \[
        F_\ora
          =\{(y,v)\in\mathcal X:H(v)=y,\ y_{\lambda+1}=\cdots=y_q=0\}
          \subseteq\{(y,v)\in\mathcal X:H(v)=y\}.
      \]
      Therefore \cref{thm:incompressibility} gives
      \[
        \E_{\ora\leftarrow\YES_\lambda}g(F_\ora)
          \le2^s\delta^{4t}\le\beta.
      \]
    \item \emph{Under $\NO_\lambda$.}
      For $\ora\leftarrow\NO_\lambda$, the entries $\ora(y,v)$ for
      $(y,v)\in\mathcal X$ are independent
      $\mathrm{Bernoulli}(\delta)$ bits. By
      \cref{thm:independent-incompressibility},
      \[
        \E_{\ora\leftarrow\NO_\lambda}g(F_\ora)
          \le2^s\delta^t=\beta.
      \]
  \end{itemize}
  Hence, for either $\mathcal P\in\{\YES_\lambda,\NO_\lambda\}$,
  Markov's inequality gives
  \[
    \Pr_{\ora\leftarrow\mathcal P}[g(F_\ora)>\epsilon]
      \le\beta/\epsilon.
  \]
  Since $t=c+\lambda$ and $M=\poly(\lambda)$, we have
  $s\le t(1+\lceil\log_2M\rceil)=O(t\log\lambda)$. Thus, for
  sufficiently large $\lambda$,
  \[
    \begin{aligned}
      \log_2\beta
        &=s-\lceil\lambda/12\rceil t\\
        &\le t\left[O(\log\lambda)-\frac{\lambda}{12}\right]
         \le-\frac{\lambda t}{24}.
    \end{aligned}
  \]

  Finally, fix an oracle $\ora$ with $g(F_\ora)\le\epsilon$.
  By \cref{est:cap}, the probability of reaching the cutoff anywhere
  in the interaction is at most $\epsilon$ for every prover $P^*$.
  Applying \cref{est:simulation}, we obtain
  \[
    \begin{aligned}
      &\bigl|\Pr[\langle(P^*)^\ora,V^\ora\rangle=\ACCEPT]
        -\Pr[\langle(P^*)^\ora,(V')^\ora\rangle=\ACCEPT]\bigr|\\
      &\hspace{3cm}\le\epsilon+g(F_\ora)
       \le2\epsilon<1/100.
    \end{aligned}
  \]
  The exceptional event $g(F_\ora)>\epsilon$ is the same for all
  provers and has probability at most $\beta/\epsilon$ under either
  distribution. This proves \eqref{est:distributional} with
  $\eta_\lambda=\beta/\epsilon\le2^{-\Omega(\lambda^2)}$.
\end{proof}

%% file: sections/qcip/public-coin-lower-bound.tex
\subsection{Lower Bound for Public-Coin Verifiers}
\label{sec:application}
In this section, we prove \cref{app:public-coin-lower-bound} by
constructing a cheating prover whose average acceptance probability
under $\NO_\lambda$ is at least $1-(N+1)\delta$. The construction
uses perfect acceptance under $\YES_\lambda$ and incompressibility
to find replies consistent with the prover's actual oracle.

\paragraph{The two distributions.}
Fix $\lambda$ and the verifier $V$, including its oracle-independent
advice. Let
\[
  \mathcal X=\mathcal Z_\lambda=\bits^q\times C,\qquad
  |\mathcal X|=2^q|C|,\qquad
  \delta=2^{-\lceil\lambda/12\rceil}.
\]
For $S\subseteq\mathcal X$, let $\ora_S:\mathcal X\to\bits$ be the
indicator of $S$. Under $\YES_\lambda$, sample
$H_1,\ldots,H_q:\Sigma\to\bits$ with independent
$\mathrm{Bernoulli}(\lambda^{-4})$ entries, and set
\[
  A=\{(y,v)\in\mathcal X:H(v)=y,\ y_{\lambda+1}=\cdots=y_q=0\}.
\]
The sampled oracle is $\ora_A$. Under $\NO_\lambda$, instead sample
an oracle $\ora_{\NO}:\mathcal X\to\bits$ with independent
$\mathrm{Bernoulli}(\delta)$ entries. Put
\[
  \mathsf{MARKED}=\{z\in\mathcal X:\ora_{\NO}(z)=1\}.
\]
For a fixed nonempty $T\subseteq\mathcal X$, apply
\cref{thm:incompressibility} with $s=0$, $t=|T|$, and a decoder that always
outputs $T$. Since $T\subseteq A$ implies $H(v)=y$ for every
$(y,v)\in T$, this bounds the first probability below; the second
follows from independent marking:
\begin{equation}
  \label{app:initial-bound}
  \Pr_H[T\subseteq A]\le\delta^{4|T|},\qquad
  \Pr_{\ora_{\NO}\leftarrow\NO_\lambda}
    [T\subseteq\mathsf{MARKED}]=\delta^{|T|}
  \qquad(|T|\le2^\lambda).
\end{equation}
Both bounds also hold for $T=\varnothing$.
We will use the gap between these probabilities to choose replies
whose required pairs are marked in $\ora_{\NO}$.

\paragraph{The initial distribution $\mu$.}
By \cref{lem:classical-prover}, we may take the designated prover
$P$ to be deterministic. For each oracle, choose the first reply
table that makes $V$ accept on every public random tape, if one
exists; otherwise, output $\bot$ and terminate. This preserves
perfect acceptance on every oracle where the original prover
achieved it and makes the reply table computable by finite enumeration.

Let
\[
  \mathcal G=\{\widehat\ora\in\operatorname{supp}(\YES_\lambda):
    \Pr[\langle P^{\widehat\ora},V^{\widehat\ora}\rangle=\ACCEPT]=1\},
  \qquad
  \alpha=\Pr_{\widehat\ora\leftarrow\YES_\lambda}
    [\widehat\ora\in\mathcal G]\ge\delta^2,
\]
and let $\mu$ be $\YES_\lambda$ conditioned on $\mathcal G$.
In particular,
\begin{equation}
  \label{app:perfect-source}
  \Pr[\langle P^{\widehat\ora},V^{\widehat\ora}\rangle=\ACCEPT]=1
  \qquad(\widehat\ora\in\operatorname{supp}(\mu)),
\end{equation}
where the probability is over $V$'s public coins. The distribution
$\mu$ remains fixed throughout the construction and is independent
of the actual oracle $\ora_{\NO}$.
Conditioning \eqref{app:initial-bound} on $\mathcal G$ gives
\begin{equation}
  \label{app:source-bound}
  \Pr_{\widehat\ora\leftarrow\mu}
    [\forall z\in T,\ \widehat\ora(z)=1]
    \le\frac{\delta^{4|T|}}{\alpha}\le\delta^{4|T|-2}
  \qquad(T\subseteq\mathcal X,\ 1\le|T|\le2^\lambda).
\end{equation}

Write $N(\lambda)$ for the query bound, $R(\lambda)\ge1$ for the
number of replies, $\ell_j(\lambda)$ for the length of reply
$j$, and $c(\lambda)$ for the total communication.
If necessary, add an empty challenge and reply to ensure $R\ge1$.
Set
\begin{equation}
  \label{app:parameters}
  K=c+R+3,\qquad
  d_j=\sum_{i=1}^j(\ell_i+1),\qquad d_0=0.
\end{equation}
The extra $j$ in $d_j$ allows a factor-two loss per round in the
probability threshold below. Note that $d_j\le K-3$.
Take $\lambda$ sufficiently large that
\[
  \delta\le\tfrac12,\qquad N\delta^2<1,\qquad K\le2^\lambda.
\]
We will impose two further bounds on $\lambda$ when bounding the
abort probability.

\subsubsection{The Cheating Prover}

We now describe the cheating prover for a fixed oracle
$\ora_{\NO}$ sampled from $\NO_\lambda$. It keeps and updates
\[
  \tau=(r_1,m_1,\ldots,r_j,m_j),\qquad
  E\subseteq\operatorname{supp}(\mu),
\]
where $r_i,m_i$ are the challenge and reply in round $i$, and $E$
is the set of candidate oracles. Initially,
$\tau=\varnothing$ and $E=\operatorname{supp}(\mu)$.

Given a transcript $\tau$ and challenge $r$, let
$P^{\widehat\ora}(\tau,r)$ denote the designated prover's reply.
We keep only oracles $\widehat\ora$ for which $P^{\widehat\ora}$
would send the same replies as $\widetilde P$, so after $j$
successful searches,
\begin{equation}
  \label{app:source-conditions}
  P^{\widehat\ora}(\tau_{<i},r_i)=m_i
  \qquad(\widehat\ora\in E,\ 1\le i\le j),
\end{equation}
where $\tau_{<i}=(r_1,m_1,\ldots,r_{i-1},m_{i-1})$.
For $F\subseteq\operatorname{supp}(\mu)$ and $T\subseteq\mathcal X$,
we write
\begin{equation}
  \label{app:hypothetical-validity}
  F[T]=\{\widehat\ora\in F:
    \forall z\in T,\ \widehat\ora(z)=1\}.
\end{equation}
Thus $F[T]$ consists of the candidates in $F$ that mark every pair
in $T$. Throughout the interaction,
$\mu(F)=\Pr_{\widehat\ora\leftarrow\mu}[\widehat\ora\in F]$
is measured without renormalizing when we remove oracles.

\paragraph{Choosing pairs to test.}
On challenge $r_j$, we try replies $m$ one by one, starting each
search with
$F_m=\{\widehat\ora\in E:P^{\widehat\ora}(\tau,r_j)=m\}$.
Restricting to $F_m$ may force additional pairs $(y,v)$ to be marked.
For example, a reply containing a preimage can identify the pair that
the verifier will later query. We therefore choose a set $T$ of pairs
to require simultaneously, while keeping sufficient probability on
the remaining oracles $F_m[T]$. Specifically, we call $T$
\emph{eligible} for $F_m$ if
\begin{equation}
  \label{app:eligible}
  T\subseteq\mathcal X,\qquad |T|<K,\qquad
  \mu(F_m[T])\ge2^{-d_j}\delta^{2|T|}.
\end{equation}
Fix total orders on replies and on $\mathcal X$. For each reply,
while an eligible set exists, choose one of maximum cardinality,
breaking ties lexicographically, and test whether
$T\subseteq\mathsf{MARKED}$. The empty set is allowed.
The threshold lower-bounds $\mu(F_m[T])$. Maximality will ensure
that these candidates have no
common marked pair outside $T$. This use of likely constraints,
maximality, and removal resembles density restoration
in~\cite[Section~3.3]{GPW} and~\cite[Section~5.3]{GGJL25}.
Here, the prover tests each eligible set in $\ora_{\NO}$ before
restricting its candidates. Its reply procedure is as follows.

\begin{protocol}
  \protocoltitle{Reply of $\widetilde P$ to challenge $r_j$}
  \textbf{State:} $\tau,E$. \quad
  \textbf{Oracle:} $\ora_{\NO}$.
  \protocolbody
  \begin{enumerate}
    \item For each possible reply $m$, in order:
      \begin{enumerate}[label*=\arabic*.]
        \item Initialize
          $F_m=\{\widehat\ora\in E:
            P^{\widehat\ora}(\tau,r_j)=m\}$.
        \item While some $T$ satisfies \eqref{app:eligible} for the
          current $F_m$:
          \begin{enumerate}[label*=\arabic*.]
            \item Choose such a $T$ of maximum cardinality,
              breaking ties lexicographically.
            \item Query $\ora_{\NO}$ on each $(y,v)\in T$.
              Retain only the bit
              $b=\bigwedge_{z\in T}\ora_{\NO}(z)$,
              with $b=1$ for $T=\varnothing$.
            \item If $b=1$, set
              \[
                E\gets F_m[T],\qquad \tau\gets\tau\|(r_j,m),
              \]
              then send $m$ and return to the interaction.
            \item If $b=0$, set $F_m\gets F_m\setminus F_m[T]$.
              Keep $(\tau,E)$ unchanged and repeat the loop.
          \end{enumerate}
        \item Discard $F_m$ and try the next reply with
          $(\tau,E)$ unchanged.
      \end{enumerate}
    \item If no reply succeeds, output $\bot$ and terminate.
  \end{enumerate}
\end{protocol}

Only $m$ is sent to $V$. A failed test removes $F_m[T]$, since each
of these candidates disagrees with $\ora_{\NO}$ on some pair in $T$.
The remaining candidates may still give reply $m$ without marking
all of $T$, so the search continues with the same reply. The set $E$
is unchanged by failed tests, and each new reply starts a fresh search
with its full collection $F_m$. If no reply succeeds, we say that
the prover \emph{aborts}.

\paragraph{Maximality and termination.}
Choosing an eligible set of maximum cardinality makes $T$
inclusion-maximal among the eligible sets for the current $F_m$.
After a failed test, $F_m[T]=\varnothing$. Since $F_m$ only shrinks
and the threshold is positive, this $T$ cannot become eligible again.
There are finitely many sets of size less than $K$, so the search
for each reply terminates with a successful test or with no eligible
set remaining.

The search uses finite enumeration and rational arithmetic:
enumerating the hash functions $H$ defining $\YES_\lambda$ and the
designated reply tables computes each probability $\mu(F_m[T])$.

\subsubsection{Invariants of the Prover}

After each successful search, every pair marked by all candidates
in $E$ belongs to $\mathsf{MARKED}$. We also bound the probability
that an oracle sampled from $\mu\mid E$ marks additional pairs.

We first prove these properties for one inclusion-maximal eligible set.

\begin{lemma}[Maximal eligible sets]
  \label{lem:maximal-eligible-set}
  Fix a round $j\in[R]$ and $F\subseteq\operatorname{supp}(\mu)$.
  Let $T$ be inclusion-maximal among sets satisfying
  \[
    T\subseteq\mathcal X,\qquad |T|<K,\qquad
    \mu(F[T])\ge2^{-d_j}\delta^{2|T|}.
  \]
  Then
  \begin{equation}
    \label{app:common-pairs}
    |T|\le K-2,\qquad
    T=\{z\in\mathcal X:\forall\widehat\ora\in F[T],\
         \widehat\ora(z)=1\}.
  \end{equation}
  Moreover, for every nonempty $U\subseteq\mathcal X\setminus T$
  with $|T\cup U|<K$,
  \begin{equation}
    \label{app:maximal-residual}
    \Pr_{\widehat\ora\leftarrow\mu\mid F[T]}
      [\forall z\in U,\ \widehat\ora(z)=1]<\delta^{2|U|}.
  \end{equation}
\end{lemma}
\begin{proof}
  If $|T|=K-1$, eligibility and \eqref{app:source-bound} imply
  \[
    2^{-d_j}\delta^{2(K-1)}\le\mu(F[T])
      \le\frac{\delta^{4(K-1)}}{\alpha}.
  \]
  Since $\alpha\ge\delta^2$, $\delta\le1/2$, and $d_j\le K-3$,
  this gives the contradiction
  \[
    1\le\frac{2^{d_j}\delta^{2(K-1)}}{\alpha}
      \le2^{d_j}\delta^{2K-4}\le2^{1-K}<1.
  \]
  Thus $|T|\le K-2$.

  Now let $\varnothing\ne U\subseteq\mathcal X\setminus T$ with
  $|T\cup U|<K$. By maximality, $T\cup U$ fails the eligibility
  threshold. Therefore,
  \[
    \Pr_{\widehat\ora\leftarrow\mu\mid F[T]}
      [\forall z\in U,\ \widehat\ora(z)=1]
      =\frac{\mu(F[T\cup U])}{\mu(F[T])}
      <\delta^{2|U|}.
  \]
  Every oracle in $F[T]$ marks every pair in $T$. Conversely, taking
  $U=\{z\}$ shows that no $z\notin T$ is marked by every oracle in
  $F[T]$, proving \eqref{app:common-pairs}.
\end{proof}

To apply these bounds throughout the interaction, let
$S=\{z\in\mathcal X:\forall\widehat\ora\in E,\ \widehat\ora(z)=1\}$;
the prover need not store $S$ separately.
We will show by induction on $j$ that after $j$ successful searches,
\eqref{app:source-conditions} holds, together with
\begin{equation}
  \label{app:source-invariant}
  |S|\le K-2,\qquad
  \mu(E)\ge2^{-d_j}\delta^{2|S|},\qquad
  S\subseteq\mathsf{MARKED},
\end{equation}
and
\begin{equation}
  \label{app:residual-bound}
  \Pr_{\widehat\ora\leftarrow\mu\mid E}
    [\forall z\in U,\ \widehat\ora(z)=1]\le\delta^{2|U|}
  \quad
  (\varnothing\ne U\subseteq\mathcal X\setminus S,\ |S\cup U|<K).
\end{equation}
For $j=0$, we have $\mu(E)=1$ and $S=\varnothing$: by the singleton
case of \eqref{app:source-bound}, no pair is marked by every oracle
in $E$. Equation~\eqref{app:source-bound} also implies
\eqref{app:residual-bound}, since $4|U|-2\ge2|U|$ for $U\ne\varnothing$.

For the induction step, suppose the test of $T$ for reply $m$ gives
$b=1$, so the updated set is $E=F_m[T]$.
By eligibility, $\mu(E)\ge2^{-d_j}\delta^{2|T|}$, and
\cref{app:common-pairs,app:maximal-residual} give
$S=T$ and \eqref{app:residual-bound}. Moreover, the successful test
ensures $S\subseteq\mathsf{MARKED}$, and every $\widehat\ora\in E$
is consistent with both the earlier replies and $m$. Receiving the
next public challenge does not change $E$, so these properties hold
until an abort.

We can further bound $|S|$ in terms of the communication. Since
every oracle in $E$ marks every pair in $S$,
\[
  2^{-d_j}\delta^{2|S|}\le\mu(E)
    \le\frac{\delta^{4|S|}}{\alpha}.
\]
Consequently,
\begin{equation}
  \label{app:forced-pair-bound}
  |S|\le\frac{d_j+\log_2(1/\alpha)}{2\log_2(1/\delta)}
    \le1+\frac{d_j}{2\log_2(1/\delta)}.
\end{equation}
\subsubsection{Finding a Reply}

We next prove that the probability of any abort is at most
$J2^{-s}$, for $J,s$ defined below. Fix $(\tau,E)$ and $r_j$, and
simulate the reply search with every test declared to fail. The
tested sets are then determined without querying $\ora_{\NO}$.
We use the probability of the candidates removed at each test to
define a distribution on these sets. Equation~\eqref{app:residual-bound}
will let us apply \cref{app:containment} to this distribution.

\begin{lemma}[Coverage of the reply search]
  \label{app:requirement-distribution}
  Fix a transcript $\tau$ through reply $j-1$, a nonempty collection
  $E\subseteq\operatorname{supp}(\mu)$, and the next challenge $r_j$.
  Let $S$ be the set of pairs $(y,v)$ on which every oracle in $E$
  returns one.
  Suppose
  \[
    |S|\le K-2,\qquad
    \mu(E)\ge2^{-d_{j-1}}\delta^{2|S|},
  \]
  and suppose the residual bound \eqref{app:residual-bound} holds.
  Simulate the reply search with every test declared to fail, without
  querying the oracle. Let $T_i$ be the $i$th tested set and $E_i'$
  the collection of oracles removed at that test.
  These collections are pairwise disjoint and satisfy
  $\sum_i\mu(E_i')>\mu(E)/2$.

  Define the distribution $\pi$ on subsets of $\mathcal X$ by
  \[
    \pi(T)=\frac{\sum_{i:T_i\setminus S=T}\mu(E_i')}
                    {\sum_i\mu(E_i')}.
  \]
  Every set in $\operatorname{supp}(\pi)$ lies outside $S$ and has
  size less than $K$. Moreover,
  \begin{equation}
    \label{app:requirement-bound}
    \Pr_{T\leftarrow\pi}[U\subseteq T]\le2\delta^{2|U|}
    \qquad(U\subseteq\mathcal X,\ 1\le|U|\le K).
  \end{equation}
  If this state and challenge occur in the interaction, the actual
  reply search succeeds whenever some $T_i\setminus S$ is
  contained in $\mathsf{MARKED}$.
\end{lemma}
\begin{proof}
  \emph{Probability of the removed candidates.}
  At test $i$ for reply $m$, the search removes $E_i'=F_m[T_i]$.
  No candidate is removed twice during the same reply search, and distinct
  replies start from disjoint collections $F_m$. Hence the sets
  $E_i'$ are pairwise disjoint. By \eqref{app:common-pairs}, $T_i$
  is the set of pairs marked by every oracle in $E_i'$. Since
  $E_i'\subseteq E$, we have $S\subseteq T_i$.

  Now consider $F_m$ when the search for reply $m$ ends. No eligible
  set remains, and $F_m[S]=F_m$, so
  \[
    \mu(F_m)=\mu(F_m[S])
      <2^{-d_j}\delta^{2|S|}
      \le2^{-\ell_j-1}\mu(E).
  \]
  Summing over the $2^{\ell_j}$ possible replies gives
  \[
    \sum_m\mu(F_m)<\tfrac12\mu(E),
    \qquad
    \sum_i\mu(E_i')>\tfrac12\mu(E).
  \]

  \emph{Containment probabilities.}
  Fix a nonempty $U\subseteq\mathcal X\setminus S$ with
  $|S\cup U|<K$. If $U\subseteq T_i\setminus S$, then
  $E_i'\subseteq E[U]$. By disjointness, the bound
  $\sum_i\mu(E_i')>\mu(E)/2$, and \eqref{app:residual-bound},
  \[
    \begin{aligned}
      \Pr_{T\leftarrow\pi}[U\subseteq T]
        &=\frac{\sum_{i:U\subseteq T_i\setminus S}\mu(E_i')}
                {\sum_i\mu(E_i')}\\
        &\le2\Pr_{\widehat\ora\leftarrow\mu\mid E}
          [\forall z\in U,\ \widehat\ora(z)=1]
         \le2\delta^{2|U|}.
    \end{aligned}
  \]
  If $U$ meets $S$ or $|S\cup U|\ge K$, then
  $\Pr_{T\leftarrow\pi}[U\subseteq T]=0$, since $S\subseteq T_i$
  and $|T_i|<K$. This proves \eqref{app:requirement-bound} in all cases.

  \emph{Success in the actual search.}
  If $(\tau,E)$ is reached before an abort, then
  $S\subseteq\mathsf{MARKED}$. Thus
  $T_i\subseteq\mathsf{MARKED}$ if and only if
  $T_i\setminus S\subseteq\mathsf{MARKED}$. The actual search
  follows the simulated sequence up to its first successful test.
  It therefore finds a reply whenever some
  $T_i\setminus S\subseteq\mathsf{MARKED}$.
\end{proof}

The actual state $(\tau,E)$ depends on earlier tests, so we cannot
assume that the entries of $\ora_{\NO}$ remain independent after
conditioning on this state. Instead, we apply the containment lemma
to each possible state under $\ora_{\NO}\leftarrow\NO_\lambda$, then
take a union bound. We first bound the number of such states.

\begin{lemma}[Counting prover states]
  \label{app:state-count}
  For a fixed full public random tape, as $\ora_{\NO}$ varies, the prover
  has at most $J$ possible states $(\tau,E)$ at the start of rounds
  before an abort, where
  \[
    J=R\cdot2^c\cdot(|\mathcal X|+1)^{K(R-1)}.
  \]
\end{lemma}
\begin{proof}
  We recover the state before round $j$ from the successful choices
  $(m_i,T_i)_{i<j}$ as follows:
  \begin{enumerate}
    \item Initialize $\tau=\varnothing$ and $E=\operatorname{supp}(\mu)$.
    \item For $i=1,\ldots,j-1$, run the reply search with challenge
      $r_i$ from the fixed public tape. Declare every test a failure
      until the test of $T_i$ for reply $m_i$, then apply its successful
      update to $(\tau,E)$.
  \end{enumerate}
  The stopping point is unique because a failed set $T$ cannot be
  tested again for the same reply. Thus at most $c$ reply bits and
  $j-1$ sets of fewer than $K$ pairs determine the state, giving
  \[
    \#\{\text{states at the start of rounds}\}
      \le\sum_{j=1}^R2^c(|\mathcal X|+1)^{K(j-1)}
      \le J.
  \]
\end{proof}

We now apply \cref{app:containment} to the distributions from
\cref{app:requirement-distribution}. Take $\lambda$ large enough that
$K\delta\le1/8$, and set
\[
  s=\left\lfloor\frac{1}{8K\delta}\right\rfloor.
\]
We also take $\lambda$ large enough that $J2^{-s}\le\delta$.
Such a choice is possible because $\log J$ is polynomially bounded
in $\lambda$ by \cref{lem:code}, whereas $s$ grows exponentially.

\begin{lemma}[Successful reply searches]
  \label{app:search-success}
  Over $\ora_{\NO}\leftarrow\NO_\lambda$
  and the verifier's public coins, the cheating prover aborts with
  probability at most $J\cdot2^{-s}$.
\end{lemma}
\begin{proof}
  Fix the full public random tape. By \cref{app:state-count}, at
  most $J$ states at the start of rounds are reachable before an
  abort as $\ora_{\NO}$ varies. For each such state, the tape fixes
  the next challenge, so we obtain a distribution $\pi$ from
  \cref{app:requirement-distribution}. Whenever the state is reached,
  the corresponding set $S$ satisfies $S\subseteq\mathsf{MARKED}$; hence
  \[
    \begin{aligned}
      &\Pr_{\ora_{\NO}\leftarrow\NO_\lambda}
        [\text{the prover reaches }(\tau,E)\text{ and aborts there}]\\
      &\qquad\le\Pr_{\ora_{\NO}\leftarrow\NO_\lambda}
        [\nexists T\in\operatorname{supp}(\pi):
          T\subseteq\mathsf{MARKED}]\\
      &\qquad\le2^{-s},
    \end{aligned}
  \]
  where the last inequality applies \cref{app:containment} under
  $\NO_\lambda$, without conditioning on reaching the state. A union
  bound over the at most $J$ states bounds the abort probability by
  $J2^{-s}$. Averaging over the public random tape proves the lemma.
\end{proof}

\subsubsection{The Final Oracle Queries}

It remains to show that, conditional on any public random tape and
test outcomes with no abort, $V$ rejects with probability at most $N\delta$.
We first prove that its final computation accepts with oracle
$\ora_S$, where $S$ is the final set of common marked pairs. We
then bound the probability that replacing $\ora_S$ by $\ora_{\NO}$
changes any query answer.

Write $V_{\mathrm{fin}}^\ora(\tau)$ for the final computation
(including oracle queries), which is deterministic because $\tau$
includes all public coins. Let $\eta$ consist of the full public
random tape and the sequence of the prover's binary test outcomes.
Each test records only the conjunction
$\bigwedge_{z\in T}\ora_{\NO}(z)$. We bound the rejection probability over
$\ora_{\NO}\leftarrow\NO_\lambda$ conditioned on $\eta$.

\begin{lemma}[Final verification]
  \label{lem:final-database-check}
  Fix $\eta$ of positive probability in which no search
  aborted. Let $\tau,E$ be the resulting state, and let $S$ be the set
  of pairs $(y,v)$ on which every oracle in $E$ returns one.
  Then
  \[
    \begin{gathered}
      V_{\mathrm{fin}}^{\ora_S}(\tau)=\ACCEPT,\\
      \Pr_{\ora_{\NO}\leftarrow\NO_\lambda}
        [V_{\mathrm{fin}}^{\ora_{\NO}}(\tau)=\REJECT\mid\eta]
        \le N\delta.
    \end{gathered}
  \]
\end{lemma}
\begin{proof}
  \emph{Acceptance under $\ora_S$.}
  By transcript consistency \eqref{app:source-conditions} and
  perfect acceptance \eqref{app:perfect-source},
  \[
    V_{\mathrm{fin}}^{\widehat\ora}(\tau)=\ACCEPT
    \qquad(\widehat\ora\in E).
  \]
  Run $V_{\mathrm{fin}}^{\ora_S}(\tau)$ and let $Q\subseteq\mathcal X$
  be the set of queried inputs, so $|Q|\le N$. The query sequence
  is determined by $S$ and $\tau$. Since $|S|\le K-2$, the singleton
  case of \eqref{app:residual-bound} applies to every $z\in Q\setminus S$.
  A union bound gives
  \[
    \begin{aligned}
      \Pr_{\widehat\ora\leftarrow\mu\mid E}
        [\exists z\in Q\setminus S:\widehat\ora(z)=1]
      &\le\sum_{z\in Q\setminus S}
        \Pr_{\widehat\ora\leftarrow\mu\mid E}[\widehat\ora(z)=1]\\
      &\le |Q|\delta^2\le N\delta^2<1.
    \end{aligned}
  \]
  Every candidate agrees with $\ora_S$ on $S$, so some
  $\widehat\ora\in E$ agrees with $\ora_S$ on all of $Q$. The two
  deterministic computations follow the same query sequence, and hence
  \[
    V_{\mathrm{fin}}^{\ora_S}(\tau)
      =V_{\mathrm{fin}}^{\widehat\ora}(\tau)=\ACCEPT.
  \]

  \emph{Conditional marking probabilities.}
  Since $\mu$ is fixed and the prover is deterministic, $\eta$
  determines every tested set. The candidate set $E$ only shrinks
  after successful searches, so its common marked set only grows.
  By \eqref{app:common-pairs}, every successful test set is contained
  in the final $S$, and the last equals $S$. For this fixed public
  random tape, conditioning on $\eta$ therefore imposes on
  $\ora_{\NO}$ exactly the conditions $S\subseteq\mathsf{MARKED}$ and
  \[
    \mathcal F=\bigcap_{\text{failed tests }T}
      \{T\not\subseteq\mathsf{MARKED}\}.
  \]
  We first condition on $S\subseteq\mathsf{MARKED}$, leaving
  the entries $\ora_{\NO}(z)$ for $z\in\mathcal X\setminus S$
  independent. For each $z\in\mathcal X\setminus S$, set
  \[
    p_b=\Pr_{\ora_{\NO}\leftarrow\NO_\lambda}
      [\mathcal F\mid S\subseteq\mathsf{MARKED},\ \ora_{\NO}(z)=b]
    \qquad(b\in\bits).
  \]
  Observe that $\mathcal F$ is decreasing in the remaining entries:
  setting $\ora_{\NO}(z)$ to zero cannot destroy it. Independence
  therefore gives $p_1\le p_0$, and Bayes' rule yields
  \begin{equation}
    \label{app:terminal-marks}
    \Pr_{\ora_{\NO}\leftarrow\NO_\lambda}
      [\ora_{\NO}(z)=1\mid\eta]
      =\frac{\delta p_1}{\delta p_1+(1-\delta)p_0}\le\delta
    \qquad(z\in\mathcal X\setminus S).
  \end{equation}

  \emph{Agreement on the final queries.}
  The computations with $\ora_S$ and $\ora_{\NO}$ agree until
  their first differing query answer. Since $S\subseteq\mathsf{MARKED}$,
  such a difference requires $\ora_{\NO}(z)=1$ at some
  $z\in Q\setminus S$. By \eqref{app:terminal-marks} and a union
  bound along the fixed $\ora_S$ query sequence,
  \[
    \begin{aligned}
      &\Pr_{\ora_{\NO}\leftarrow\NO_\lambda}
        [V_{\mathrm{fin}}^{\ora_{\NO}}(\tau)=\REJECT\mid\eta]\\
      &\qquad\le\sum_{z\in Q\setminus S}
        \Pr_{\ora_{\NO}\leftarrow\NO_\lambda}
          [\ora_{\NO}(z)=1\mid\eta]\\
      &\qquad\le N\delta.
    \end{aligned}
  \]
\end{proof}

\begin{proof}[Proof of \cref{app:public-coin-lower-bound}]
  Average \cref{lem:final-database-check} over the public coins and
  test outcomes with no abort, and add the abort probability from
  \cref{app:search-success}. Since $J2^{-s}\le\delta$,
  \[
    \begin{aligned}
      \E_{\ora\leftarrow\NO_\lambda}
        \Pr[\langle\widetilde P^\ora,V^\ora\rangle=\ACCEPT]
        &\ge1-N\delta-J2^{-s}\\
        &\ge1-(N+1)\delta.
    \end{aligned}
  \]
  The prover uses only its actual oracle and the fixed distribution
  $\mu$. All parameter bounds depend only on the number of rounds $R$,
  total communication $c$, and query bound $N$, uniformly over $V$'s
  advice and the conditioning event $\mathcal G$.
  This proves \eqref{app:public-coin-failure}.
\end{proof}

%% file: appendices/oracle-separations.tex
\section{From Distributions to Oracle Separations}
\label{app:diagonal}

In this section, we prove the oracle-input promise and fixed-oracle
decision separations from \cref{thm:main}. We first restrict to oracles
satisfying the QMA guarantees, then diagonalize against QCIP verifiers.

\subsection{Oracle-Input Promise Separation}
\label{app:promise}
In this section, we define a promise problem from the QMA acceptance
guarantees and show that it admits no QCIP protocol.
Although $\YES_\lambda$ and $\NO_\lambda$ may have overlapping supports,
these guarantees define disjoint promise sets. By \cref{thm:main},
each distribution assigns probability at least $1-2^{-\Omega(\lambda)}$
to its corresponding promise set.

\begin{definition}[Oracle-input promise problem]
  \label{def:oracle-promise}
  Let $V$ be the decoder-based verifier specified in
  \cref{sec:qma-verifier}, with witness length $a(\lambda)$ from
  \cref{lem:code}. Fix a constant $\lambda_0\ge2$ such that
  \cref{lem:code} holds and $\gamma_{\NO}<1/3$ for every
  $\lambda\ge\lambda_0$.
  For $\lambda\ge\lambda_0$, define the promise sets as follows:
  \begin{enumerate}
    \item $\ora\in L_{\YES,\lambda}$ if
      $\ora\in\operatorname{supp}(\YES_\lambda)$ and
      \[
        \exists\ket\psi:\quad
        \Pr[V^\ora(1^\lambda,\ket\psi)=\ACCEPT]\ge1-2^{-\lambda}.
      \]
    \item $\ora\in L_{\NO,\lambda}$ if
      $\ora\in\operatorname{supp}(\NO_\lambda)$ and
      \[
        \forall\ket\psi:\quad
        \Pr[V^\ora(1^\lambda,\ket\psi)=\ACCEPT]\le\gamma_{\NO}.
      \]
  \end{enumerate}
  In both conditions, the states range over $a(\lambda)$ qubits.
  For $\lambda<\lambda_0$, take $L_{\YES,\lambda}=\varnothing$ and let
  $L_{\NO,\lambda}$ contain only the zero oracle.
\end{definition}

As before, the input consists of $1^\lambda$ and query access to
$\ora\in L_{\YES,\lambda}\cup L_{\NO,\lambda}$, without access
to the hash functions $H$. Observe that these promise sets are disjoint:
$1-2^{-\lambda}\ge3/4>\gamma_{\NO}$ for $\lambda\ge\lambda_0$,
and the YES set is empty at smaller lengths.

\begin{theorem}[Oracle-input promise separation]
  \label{thm:oracle-input}
  The promise problem of \cref{def:oracle-promise} admits a QMA
  protocol but no QCIP protocol. More precisely:
  \begin{enumerate}
    \item \textbf{QMA upper bound.}
      A uniform polynomial-time quantum verifier $V$ uses an
      $a(\lambda)$-qubit witness and at most one classical oracle query.
      For every $\lambda$,
      \[
        \begin{aligned}
          \ora\in L_{\YES,\lambda}
          &\ \Longrightarrow\ 
          \exists\ket\psi:\ 
          \Pr[V^\ora(1^\lambda,\ket\psi)=\ACCEPT]\ge2/3,\\
          \ora\in L_{\NO,\lambda}
          &\ \Longrightarrow\ 
          \forall\ket\psi:\ 
          \Pr[V^\ora(1^\lambda,\ket\psi)=\ACCEPT]\le1/3.
        \end{aligned}
      \]
    \item \textbf{QCIP lower bound.}
      For every QCIP protocol $(P',V')$ and every sufficiently large $\lambda$,
      at least one of the following holds:
      \[
        \begin{aligned}
          &\exists\ora\in L_{\YES,\lambda}:\quad
            \Pr[\langle(P')^\ora,(V')^\ora\rangle=\ACCEPT]<2/3;\\
          &\exists\ora\in L_{\NO,\lambda}\quad\exists P^*:\quad
            \Pr[\langle(P^*)^\ora,(V')^\ora\rangle=\ACCEPT]>1/3.
        \end{aligned}
      \]
  \end{enumerate}
  Every prover receives $1^\lambda$ and has unrestricted oracle access.
  The threshold for $\lambda$ depends only on $V'$'s round, query,
  and communication bounds.
  Each acceptance probability is taken over an execution with the
  oracle fixed.
\end{theorem}

\begin{proof}[Proof of \cref{thm:oracle-input}]
  \emph{QMA upper bound.}
  We use $V$, except that it outputs $\REJECT$
  and terminates for $\lambda<\lambda_0$.
  For $\lambda\ge\lambda_0$, the promise gives completeness at least
  $1-2^{-\lambda}\ge3/4>2/3$ and soundness at most $\gamma_{\NO}<1/3$.
  The witness length, running time, and query count are as stated
  in \cref{qma:yesprob}.

  \emph{QCIP lower bound.}
  Fix $(P',V')$ and choose $\lambda$ large enough that \cref{thm:main}
  applies and $\gamma_{\YES}+\gamma_{\NO}<1/3$. If completeness
  fails on some oracle in $L_{\YES,\lambda}$, we obtain the first
  alternative. Otherwise,
  \[
    \Pr_{\ora\leftarrow\YES_\lambda}
      \left[\Pr[\langle(P')^\ora,(V')^\ora\rangle=\ACCEPT]\ge2/3\right]
      \ge1-\gamma_{\YES}.
  \]
  Applying \cref{thm:main} gives
  \[
    \Pr_{\ora\leftarrow\NO_\lambda}
      \left[\forall P^*:\
        \Pr[\langle(P^*)^\ora,(V')^\ora\rangle=\ACCEPT]\le1/3\right]
      \le2/3+\gamma_{\YES},
  \]
  whereas
  \[
    \Pr_{\ora\leftarrow\NO_\lambda}[\ora\in L_{\NO,\lambda}]
    \ge1-\gamma_{\NO}>2/3+\gamma_{\YES}.
  \]
  Hence some $\ora\in L_{\NO,\lambda}$ and prover $P^*$ satisfy
  $\Pr[\langle(P^*)^\ora,(V')^\ora\rangle=\ACCEPT]>1/3$.
\end{proof}

\subsection{Fixed-Oracle Decision Separation}
\label{app:fixed-oracle}

In this section, we obtain a unary language outside QCIP relative
to a fixed classical oracle by diagonalization~\cite{BGS75,FS88};
see also~\cite[Section~4.6]{Zha25}. We choose one oracle block for
each verifier and place later blocks beyond that verifier's query-length
bound. Note that an unbounded prover can still read the later blocks.
We therefore first show that a completeness failure excludes every
prover, even when a fixed auxiliary oracle is available.

\begin{lemma}[Fixed auxiliary oracle]
  \label{app:auxiliary}
  Let $V^{\ora,\ora'}$ be a QCIP verifier with access to an oracle input
  $\ora:\mathcal Z_\lambda\to\bits$ and an auxiliary Boolean oracle
  $\ora'$. For every sufficiently large $\lambda$ and every fixed
  $\ora'$, at least one of the following holds:
  \[
    \begin{aligned}
      &\exists\ora\in L_{\YES,\lambda}\quad\forall P^*:\quad
        \Pr[\langle(P^*)^{\ora,\ora'},V^{\ora,\ora'}\rangle=\ACCEPT]<2/3;\\
      &\exists\ora\in L_{\NO,\lambda}\quad\exists P^*:\quad
        \Pr[\langle(P^*)^{\ora,\ora'},V^{\ora,\ora'}\rangle=\ACCEPT]>1/3.
    \end{aligned}
  \]
  The promise sets are from \cref{def:oracle-promise}. The threshold for
  $\lambda$ depends only on $V$'s round, query, and communication bounds,
  not on $\ora'$.
  Each acceptance probability is over a protocol execution with both
  $\ora$ and $\ora'$ fixed.
\end{lemma}
\begin{proof}
  \emph{Accounting for the auxiliary oracle.}
  Give both parties the advice
  $\ora'|_{\bits^{\le T(\lambda)}}$, where $T$ bounds $V$'s running time,
  and implement $V$'s queries to $\ora'$ by reversible lookup.
  Since this table is independent of the sampled oracle input $\ora$,
  the resulting verifier is in the query model of
  \cref{def:query-model}. Hence \cref{rem:query-model-hardness}
  applies uniformly over the advice for fixed round, query, and
  communication bounds.

  \emph{Choosing a prover.}
  Now fix a sufficiently large $\lambda$. For each oracle input $\ora$,
  choose a deterministic reply table maximizing acceptance by
  $V^{\ora,\ora'}$, which exists by \cref{lem:classical-prover}.
  Let $P_\lambda$ read $\ora$ and use its maximizing reply table
  on $1^\lambda$. On other inputs, it outputs $\bot$ and terminates.
  These finitely many tables can be hardwired into $P_\lambda$.

  If the first alternative fails, $(P_\lambda,V)$ is complete on
  all of $L_{\YES,\lambda}$. The proof of \cref{thm:oracle-input}
  then implies that some $\ora\in L_{\NO,\lambda}$ and prover $P^*$ satisfy
  \[
    \Pr[\langle(P^*)^{\ora,\ora'},V^{\ora,\ora'}\rangle=\ACCEPT]>1/3,
  \]
  since its threshold for $\lambda$ is uniform over provers.
\end{proof}

\begin{theorem}[Fixed-oracle decision separation]
  \label{cor:oracle}
  There are a classical oracle $\ora:\bits^*\to\bits$ and a unary language
  $L\in\QMA^\ora\setminus\QCIP^\ora$.
\end{theorem}
\begin{proof}
  \emph{Preparing the oracle blocks.}
  We first realize each oracle input
  $\ora_\lambda:\mathcal Z_\lambda\to\bits$ as a block of a Boolean
  oracle on binary strings. Recall that
  $\mathcal Z_\lambda=\bits^q\times C$,
  $\Sigma=\mathbb F_q^\lambda$, and
  $\lambda^5<q\le2\lambda^5$. Writing each field element in
  $\lceil\log_2 q\rceil$ bits encodes a pair $(y,v)\in\mathcal Z_\lambda$
  in
  \[
    q+q\lambda\lceil\log_2q\rceil
      =O(\lambda^6\log\lambda)\le\lambda^8
  \]
  bits for all sufficiently large $\lambda$.
  Enlarge $\lambda_0$ accordingly and pad each encoding with zeros
  to length $\ell(\lambda)=\lambda^8$. Write
  \[
    \operatorname{enc}_\lambda:\mathcal Z_\lambda\hookrightarrow
      \bits^{\ell(\lambda)}
  \]
  for this encoding and call its image \emph{block $\lambda$}.
  Encoding and decoding, including testing $v\in C$, take polynomial
  time. A block realizes $\ora_\lambda$ by assigning
  $\ora(\operatorname{enc}_\lambda(y,v))=\ora_\lambda(y,v)$;
  strings outside all block encodings have value zero.

  Let $V$ be the QMA verifier of \cref{sec:qma}, and enumerate the
  uniform QCIP verifiers as $(V_i,T_i)_{i\ge1}$, where $T_i$ is a
  polynomial bound on time and query length.

  When applying \cref{app:auxiliary}, we route queries in block
  $\lambda$ to $\ora$ and all other queries to $\ora'$.
  The encoding can be computed reversibly, so routing preserves
  polynomial running time and query bounds.

  \emph{Defeating the verifiers one at a time.}
  We set $d_0=0$. At stage $i$, we choose an input length $\lambda_i$
  for $V_i$ and fix the oracle at every address it can query on that input:
  \begin{enumerate}
    \item Choose $\lambda_i$ so that $\ell(\lambda_i)>d_{i-1}$ and
      $\lambda_i$ is sufficiently large for the promise guarantees of
      \cref{thm:oracle-input} and for \cref{app:auxiliary} applied to $V_i$.
    \item Put $d_i=\max\{d_{i-1},\ell(\lambda_i),T_i(\lambda_i)\}$. Fix every
      unassigned oracle value at length at most $d_i$ to zero, except the
      canonical addresses of block $\lambda_i$.
    \item Let $\ora'_i$ agree with the values fixed so far and equal zero
      elsewhere. Choose a block $\ora_i$ and a label $b_i\in\bits$
      satisfying one of the following:
      \begin{itemize}
        \item $b_i=1$, $\ora_i\in L_{\YES,\lambda_i}$, and
          \[
            \forall P^*:\quad
            \Pr[\langle(P^*)^{\ora_i,\ora'_i},V_i^{\ora_i,\ora'_i}\rangle=\ACCEPT]<2/3.
          \]
        \item $b_i=0$, $\ora_i\in L_{\NO,\lambda_i}$, and
          \[
            \exists P^*:\quad
            \Pr[\langle(P^*)^{\ora_i,\ora'_i},V_i^{\ora_i,\ora'_i}\rangle=\ACCEPT]>1/3.
          \]
      \end{itemize}
      Set block $\lambda_i$ to $\ora_i$.
  \end{enumerate}
  Such a pair $(\ora_i,b_i)$ exists by \cref{app:auxiliary}.
  After all stages, we set every unassigned value to zero and define
  $L=\{1^{\lambda_i}:b_i=1\}$.

  \emph{QCIP lower bound.}
  Fix a stage $i$. Later blocks have address length greater than
  $d_i\ge T_i(\lambda_i)$, so they cannot change $V_i$'s oracle
  answers on $1^{\lambda_i}$. For the prover, we consider two cases:
  \begin{itemize}
    \item If $b_i=1$, Step 3 excludes every strategy, including any
      strategy selected using later blocks.
    \item If $b_i=0$, the deterministic reply table witnessing
      soundness failure still makes $V_i$ accept with probability
      greater than $1/3$.
      Use that table on $1^{\lambda_i}$ and arbitrary replies elsewhere.
  \end{itemize}
  Thus $V_i$ fails completeness or soundness on $1^{\lambda_i}$ even
  after the entire oracle $\ora$ has been fixed.

  \emph{QMA upper bound.}
  Finally, our QMA verifier outputs $\REJECT$ and terminates on
  nonunary inputs or $1^\lambda$ with $\lambda<\lambda_0$; otherwise,
  it runs $V$ on block $\lambda$.
  On block $\lambda_i$, \cref{thm:oracle-input} gives completeness
  if $b_i=1$ and soundness if $b_i=0$. On unselected blocks, $V$
  always outputs $\REJECT$, since those blocks are zero. The
  algorithm is uniform: it need not identify the selected blocks.
  Hence $L\in\QMA^\ora\setminus\QCIP^\ora$.
\end{proof}

%% file: appendices/qcip-amplification.tex
\section{Amplification for QCIP}
\label{app:amplification}

In this section, we record the standard sequential-repetition argument
for QCIP~\cite[Section~2.5]{GHMW15}. It applies relative to every fixed
oracle, even when a cheating prover retains its state across repetitions.

\begin{lemma}[Relativizing amplification]
  \label{lem:qcip-amplification}
  Fix a classical oracle $\ora$ and a QCIP protocol $(P,V)$ with completeness
  $c(n)$ and soundness $s(n)$, where $c,s$ are polynomial-time computable
  rational-valued functions and $0\le s(n)<c(n)\le1$.
  Put $\Delta(n)=c(n)-s(n)$.
  For every polynomial $k:\mathbb N\to\mathbb N_{\ge1}$, sequential
  repetition gives a QCIP protocol $(P',V')$ with
  completeness at least $1-2^{-k(n)\Delta(n)^2/(2\ln2)}$ and soundness at most
  $2^{-k(n)\Delta(n)^2/(2\ln2)}$ relative to $\ora$.
  The verifier's number of messages and oracle queries is at most $k(n)$
  times that of $V$. The construction is the same for every $\ora$.
\end{lemma}
\begin{proof}
  Fix an input $x$ of length $n$ and suppress the argument $n$ in
  $c,s,\Delta,k$. We define $(V')^\ora$ as follows:
  \begin{enumerate}
    \item For $i=1,\ldots,k$, execute $V^\ora$ on $x$ with fresh workspace
      and randomness. Record its measured acceptance bit $X_i\in\bits$ and
      discard the remaining workspace. Use the recorded bits only in
      the final decision.
    \item Output $\ACCEPT$ if $\sum_{i=1}^k X_i\ge k(c+s)/2$
      and $\REJECT$ otherwise, then terminate.
  \end{enumerate}
  We let the designated prover $P'$ run independent copies of $P$.

  \emph{Soundness.}
  Suppose $x$ is a NO input, and fix a cheating prover $P^*$.
  Condition on any values of $X_1,\ldots,X_{i-1}$ with nonzero
  probability. The prover's conditional private state and its
  strategy in repetition $i$ define a prover $P_i^*$ for $V^\ora$.
  Since the verifier starts afresh in repetition $i$, soundness
  of $V^\ora$ implies
  \[
    \Pr[X_i=1\mid X_1,\ldots,X_{i-1}]\le s.
  \]
  By Hoeffding's lemma~\cite{Hoeffding63}, for $t\ge0$ we have
  \[
    \E[e^{t(X_i-s)}\mid X_1,\ldots,X_{i-1}]\le e^{t^2/8}.
  \]
  Iterating this inequality and applying exponential Markov's
  inequality with $t=2\Delta$, we obtain
  \begin{gather*}
    \E\exp\!\left(t\sum_{i=1}^k(X_i-s)\right)\le e^{kt^2/8},\\
    \Pr\!\left[\sum_{i=1}^kX_i\ge k\cdot\frac{c+s}{2}\right]
      \le \exp\!\left(-\frac{tk\Delta}{2}+\frac{kt^2}{8}\right)
      =2^{-k\Delta^2/(2\ln2)}.
  \end{gather*}

  \emph{Completeness.}
  If $x$ is a YES input, each independent execution with $P'$ accepts
  with probability at least $c$. Hoeffding's inequality therefore gives
  \[
    \Pr\!\left[\sum_{i=1}^kX_i<k\cdot\frac{c+s}{2}\right]
      \le 2^{-k\Delta^2/(2\ln2)}.
  \]
  Finally, each repetition exchanges only classical messages and
  uses the same oracle $\ora$, with at most $k$ times as many messages
  and queries as $V$.
\end{proof}

In particular, if $\Delta(n)\ge1/\poly(n)$, polynomially many repetitions give
completeness at least $1-2^{-p(n)}$ and soundness at most $2^{-p(n)}$
for any polynomial $p$. Thus the choice $2/3,1/3$ in \cref{def:qcip}
is without loss of generality, also relative to an oracle.

%% file: appendices/interactive-proof-tools.tex
\section{Standard Tools for Interactive Proofs}
\label{app:ip-transformations}

In this section, we show that classical provers suffice when all
messages are classical. We then verify that the standard public-coin
and perfect-completeness transformations preserve polynomial query
bounds, as stated in \cref{lem:public-coin-perfect}.

\input{appendices/classical-provers}

\input{appendices/protocol-transformations}

%% file: appendices/classical-provers.tex
\subsection{Classical Prover Strategies}
\label{app:classical-provers}
In this section, we prove \cref{lem:classical-prover}: classical
messages allow a classical prover to reproduce any quantum prover's
strategy, and a deterministic prover suffices to maximize acceptance.
The key observation is that the parties' private states remain
independent conditioned on the transcript, for a fixed input and oracle.

\begin{proof}[Proof of \cref{lem:classical-prover}]
  Fix the input and oracle $\ora$, and let $P^*$ be the given prover.
  Let $\mathsf V,\mathsf P$ be the parties' private registers.
  We use the purification convention of \cref{sec:prelims}.

  \emph{Conditional states.}
  We proceed by induction on the number of messages. The initial
  state is a product state. Suppose that, after a transcript $\pi$
  of nonzero probability, the conditional state is
  $\ket{\xi_\pi}_{\mathsf V}\ket{\zeta_\pi}_{\mathsf P}$.
  If the verifier sends the next message, its local computation acts as
  \[
    \ket0_{\mathsf M}\ket{\xi_\pi}_{\mathsf V}
      \ket{\zeta_\pi}_{\mathsf P}
      \longmapsto
    \sum_m\ket m_{\mathsf M}\ket{\xi_{\pi,m}}_{\mathsf V}
      \ket{\zeta_\pi}_{\mathsf P},
  \]
  where $\mathsf M$ is the message register. Upon measuring
  $\mathsf M$, outcome $m$ has probability $\|\xi_{\pi,m}\|_2^2$;
  conditioned on an outcome of nonzero probability, the private state is
  \[
    \frac{\ket{\xi_{\pi,m}}_{\mathsf V}}{\|\xi_{\pi,m}\|_2}
      \otimes\ket{\zeta_\pi}_{\mathsf P}.
  \]
  The argument is the same when the prover sends a message.
  Thus the conditional state remains a product, with each factor
  determined by that party's strategy, the input, $\ora$, and the
  transcript so far.
  In particular, the original private states satisfy
  \[
    \rho_\pi^{\mathsf V\mathsf P}
      =\rho_\pi^{\mathsf V}\otimes\rho_\pi^{\mathsf P}.
  \]

  \emph{Classical provers.}
  We can therefore define the classical prover from $P^*$ alone.
  For each transcript $\pi$ preceding a reply, run $P^*$ from its
  initial state, supplying the verifier messages in $\pi$ and
  conditioning on its replies agreeing with $\pi$. Write its
  next-reply state and the resulting reply distribution as
  \[
    \sum_m\ket m_{\mathsf M}\ket{\zeta_{\pi,m}}_{\mathsf P},
    \qquad q_\pi(m)=\|\zeta_{\pi,m}\|_2^2,
  \]
  where we choose an arbitrary reply distribution if the conditioning
  event has probability zero. Given the input, oracle access, and full
  transcript $\pi$, the classical strategy samples and returns
  $m\leftarrow q_\pi$ using fresh randomness. It retains no private
  state between replies.
  Observe that $q_\pi$ depends only on $P^*$, the input, $\ora$,
  and $\pi$. The conditional-state calculation above shows that
  the classical strategy preserves both the transcript distribution
  and the conditional verifier states against every verifier.

  \emph{Deterministic provers.}
  Finally, fix $V$. We may sample the classical strategy's replies
  in advance: independently choose
  $f(\pi)\leftarrow q_\pi$ for every transcript $\pi$ preceding a reply.
  Let $\nu$ be the resulting distribution on the finite set
  $\mathcal T$ of deterministic reply tables. Since these advance
  choices do not change the interaction's distribution,
  \[
    \Pr[\langle(P^*)^\ora,V^\ora\rangle=\ACCEPT]
      =\E_{f\leftarrow\nu}a(f)\le\max_{f\in\mathcal T}a(f),
  \]
  where $a(f)$ is $V$'s acceptance probability when the prover
  replies with $f(\pi)$ after transcript $\pi$. Since
  $\mathcal T$ is finite, a deterministic prover attains this
  maximum and hence does at least as well as $P^*$.
\end{proof}

%% file: appendices/protocol-transformations.tex
\subsection{Public Coins and Perfect Completeness}
\label{app:combined-transformations}

In this section, we observe that the transformations of Goldwasser
and Sipser~\cite{GS} and F\"urer et al.~\cite{FGMSZ} apply in the query
model of \cref{def:query-model}. Their combined guarantee is
\cref{lem:public-coin-perfect}, whose proof invokes the two lemmas below.
For public-coin protocols, we count a challenge and reply as one round
(with an empty challenge if the prover speaks first).

\begin{proof}[Proof of \cref{lem:public-coin-perfect}]
  Apply \cref{lem:gs} to $(P,V)$ and then apply
  \cref{lem:perfect-completeness} to the resulting protocol.
  For each fixed oracle, the first lemma changes the two acceptance
  bounds from $2/3-1/100$ and $1/3+1/100$ to $2/3$ and $1/3$,
  respectively. The second changes $2/3$ to $1$ and preserves $1/3$.
  Both lemmas preserve the query model and keep all verifier queries
  after the last reply. Moreover, the second produces a deterministic
  prover, and both verifiers depend only on the original verifier and
  the polynomial bounds in \cref{def:query-model}.
\end{proof}

\subsubsection{Public Coins}
\label{app:public-coin-conversion}

In this section, we observe that the Goldwasser--Sipser
transformation~\cite{GS} preserves polynomial query bounds.
We first record why public-coin verifiers can postpone their queries
until after the interaction.

\begin{fact}[Delaying public-coin queries]
  \label{fact:public-coin-delay}
  Let $V$ be a classical public-coin verifier with fixed message lengths
  and a fixed number of rounds. Its oracle queries can be postponed until
  after the last reply, preserving their number and lengths and the
  acceptance probability for every fixed oracle and prover $P^*$.
\end{fact}
\begin{proof}
  We define the modified verifier in two phases:
  \begin{enumerate}
    \item In each round, sample and send the public challenge, then
      save the prover's reply. Make no oracle queries.
    \item Simulate $V^\ora$ using the saved challenges as its random tape
      and supplying the saved replies when requested. Return its decision;
      ignore replies after the simulated verifier halts.
  \end{enumerate}
  For any fixed oracle $\ora$ and prover $P^*$, the challenges and
  replies have the original distribution up to the point where
  $V$ halts. Thus the second phase preserves its queries and decision.
\end{proof}

\begin{lemma}[Public coins]
  \label{lem:gs}
  Let $(P,V)$ be a protocol in the classical query model of
  \cref{def:query-model}. There is a public-coin protocol $(P',V')$
  in the same model, where $V'$ depends only on $V$ and the polynomial
  bounds in \cref{def:query-model}. The verifier $V'$
  makes all oracle queries after the last prover reply.
  For every oracle input $\ora$:
  \begin{enumerate}
    \item \textbf{Designated prover.}
      If $P^\ora$ makes $V^\ora$ accept with probability at least
      $2/3-1/100$, then $(P')^\ora$ makes $(V')^\ora$ accept with
      probability at least $2/3$.
    \item \textbf{Arbitrary provers.}
      If every prover $P^*$ makes $V^\ora$ accept with probability at most
      $1/3+1/100$, then every prover $P^*$ makes $(V')^\ora$ accept with probability
      at most $1/3$.
  \end{enumerate}
\end{lemma}

\begin{proof}[Proof of \cref{lem:gs}]
  We first strengthen the two acceptance bounds to $2/3$ and $1/3$
  under the respective hypotheses, using a constant number of sequential
  repetitions with majority decision. We then apply Goldwasser--Sipser and its
  round-preserving amplification~\cite[Sections~4.1--4.2]{GS}, obtaining
  a public-coin verifier $V'$.

  \emph{Oracle queries.}
  To bound the queries, let $\widehat V$ be the original verifier
  after both amplification steps. The only oracle-dependent test checks
  whether a transcript $\tau$ and private random tape $\xi$ give a
  consistent accepting execution:
  \begin{enumerate}
    \item Run $\widehat V^\ora$ with random tape $\xi$.
    \item Whenever it sends a message, compare it with the corresponding
      verifier message in $\tau$. On a mismatch, output $\REJECT$ and
      terminate the test. Whenever it requests a prover reply, supply
      the corresponding reply from $\tau$.
    \item Return the decision of $\widehat V^\ora$.
  \end{enumerate}
  Observe that this test makes no more queries than one execution
  of $\widehat V$. All other tests evaluate public hashes or compare
  integers, so they need no oracle queries. Their parameters depend
  only on the round, communication, and random-bit bounds of $V$,
  independently of its local running time.

  \emph{Fixed public challenges.}
  We let $V'$ send uniform hash descriptions and strings of fixed
  maximum lengths, using earlier replies to select the required
  submatrices and prefixes. Out-of-range parameters cause $V'$ to
  output $\REJECT$ and terminate. Unused bits do not affect either
  acceptance guarantee: the designated prover ignores them, and an
  arbitrary prover can sample them privately in the original protocol.
  We then pad early halts as in \cref{def:query-model} and apply
  \cref{fact:public-coin-delay}.

  \emph{Acceptance bounds.}
  For each fixed $\ora$, the Goldwasser--Sipser acceptance guarantees
  establish both implications in \cref{lem:gs}: acceptance at least
  $2/3$ for a designated prover under the first hypothesis, and
  acceptance at most $1/3$ for every prover under the second.
  The prover $P'$ reads $\ora$ and enumerates deterministic reply
  tables, using the first making $(V')^\ora$ accept with probability
  at least $2/3$. If none exists, it outputs $\bot$ and terminates.
\end{proof}

\subsubsection{Perfect Completeness}
\label{app:perfect-completeness}

In this section, we observe that the perfect-completeness
transformation of F\"urer et al.~\cite[Section~3]{FGMSZ} preserves
polynomial query bounds. Recall that their verifier runs several
copies with shifted public random tapes and accepts if any copy accepts.

\begin{lemma}[Perfect completeness]
  \label{lem:perfect-completeness}
  Let $(P,V)$ be a public-coin protocol in the classical query model of
  \cref{def:query-model}, whose verifier $V$ makes all its oracle queries
  after the last prover reply. There is another such protocol $(P',V')$,
  where $P'$ is deterministic and $V'$ depends only on $V$ and the
  polynomial bounds in \cref{def:query-model}.
  For every oracle input $\ora$:
  \begin{enumerate}
    \item \textbf{Designated prover.}
      If $\langle P^\ora,V^\ora\rangle$ accepts with probability
      at least $2/3$, then
      \[
        \Pr[\langle(P')^\ora,(V')^\ora\rangle=\ACCEPT]=1.
      \]
    \item \textbf{Arbitrary provers.}
      If every prover $P^*$ makes $V^\ora$ accept
      with probability at most $1/3$, then every prover $P^*$ makes $(V')^\ora$
      accept with probability at most $1/3$.
  \end{enumerate}
\end{lemma}

\begin{proof}[Proof of \cref{lem:perfect-completeness}]
  We first amplify $V$ by sequential repetition with majority decision,
  reducing the stated error bounds from $1/3$ to $\epsilon$. Write
  $\widehat V$ for the resulting verifier and $\ell\ge1$ for its
  number of public random bits (adding an unused bit if necessary).
  Polynomially many repetitions suffice to make
  $\epsilon<1/(3\ell)$, since $\epsilon$ decreases exponentially
  while $\ell$ grows linearly. We postpone $\widehat V$'s queries
  using \cref{fact:public-coin-delay}.

  To define the new verifier, let $R$ be the number of rounds of
  $\widehat V$, with challenge length $\ell_j$ and reply length $a_j$
  in round $j$, so that $\ell=\sum_{j=1}^R\ell_j$. For each
  $i\in[\ell]$ and shift $s^{(i)}\in\bits^\ell$, write
  $s^{(i)}=s_1^{(i)}\|\cdots\|s_R^{(i)}$ with
  $s_j^{(i)}\in\bits^{\ell_j}$.
  The verifier is as follows~\cite[Section~3.1]{FGMSZ}.
  \begin{protocol}
    \protocoltitle{Perfect-completeness verifier $(V')^\ora$}
    \textbf{Input:} $1^\lambda$.\par
    \textbf{Oracle:} $\ora$.\par
    \textbf{Output:} $\ACCEPT$ or $\REJECT$.
    \protocolbody
    \begin{enumerate}
      \item Receive $s^{(1)},\ldots,s^{(\ell)}\in\bits^\ell$.
      \item For $j=1,\ldots,R$:
        \begin{enumerate}[label*=\arabic*.]
          \item Sample $r_j\leftarrow\bits^{\ell_j}$ uniformly and send it
            to the prover.
          \item Receive $m_j^{(1)},\ldots,m_j^{(\ell)}\in\bits^{a_j}$.
        \end{enumerate}
      \item For each $i\in[\ell]$, run the final computation of
        $\widehat V^\ora$ on the transcript
        \[
          \tau_i=\bigl(r_j\oplus s_j^{(i)},m_j^{(i)}\bigr)_{j=1}^R.
        \]
        Output $\ACCEPT$ if at least one computation accepts and
        $\REJECT$ otherwise, then terminate.
    \end{enumerate}
  \end{protocol}
  The initial prover message has an empty public challenge.
  Fix the oracle $\ora$ for the rest of the proof.

  \emph{Designated prover.}
  Suppose $P$ makes $V^\ora$ accept with probability at least $2/3$.
  Amplification and \cref{lem:classical-prover} give a deterministic
  prover $\widehat P$ making $\widehat V^\ora$ accept with probability at least
  $1-\epsilon$. Let $G\subseteq\bits^\ell$ be the tapes on which
  $\widehat V^\ora$ accepts against $\widehat P$.
  We choose shifts whose translates of $G$ cover every public tape.
  For independent uniform shifts $s^{(1)},\ldots,s^{(\ell)}$, the covering
  argument of~\cite{Lautemann83,FGMSZ} gives
  \[
    \E\left[\left|\bits^\ell\setminus
      \bigcup_{i=1}^\ell(G\oplus s^{(i)})\right|\right]
    =2^\ell\left(1-\frac{|G|}{2^\ell}\right)^\ell
    \le 2^\ell\epsilon^\ell<1.
  \]
  Since the number of uncovered tapes is a nonnegative integer,
  some choice of shifts satisfies
  \[
    \bigcup_{i=1}^\ell(G\oplus s^{(i)})=\bits^\ell.
  \]
  The prover sends these shifts and follows $\widehat P$ in each copy.
  Thus at least one copy accepts for every public tape, so
  $(V')^\ora$ accepts with probability $1$ against this prover.

  \emph{Arbitrary provers.}
  Suppose every prover $P^*$ makes $V^\ora$ accept with probability at most
  $1/3$, so amplification reduces this bound to $\epsilon$ for
  $\widehat V^\ora$.
  By \cref{lem:classical-prover}, it suffices to consider a deterministic
  cheating prover $P^*$ for $V'$. Fix its initial shifts.
  To bound copy $i$'s acceptance probability, we use the reduction
  of~\cite[Section~3.3]{FGMSZ}. The prover $\widehat P_i^*$ starts an
  internal execution of $P^*$ and, on challenge
  $z_j\in\bits^{\ell_j}$ from $\widehat V$, does the following:
  \begin{enumerate}
    \item Supply $r_j=z_j\oplus s_j^{(i)}$ to the internal $P^*$.
    \item Receive its reply tuple $(m_j^{(1)},\ldots,m_j^{(\ell)})$ and
      return $m_j^{(i)}$ to $\widehat V$.
  \end{enumerate}
  This strategy uses only challenges received so far. Since the
  $r_j$ are independent and uniform, copy $i$ has the same
  distribution in the two executions. Thus each copy accepts with
  probability at most $\epsilon$, and a union bound yields
  \[
    \begin{aligned}
      \Pr[\text{copy $i$ accepts}]
        &=\Pr[\langle(\widehat P_i^*)^\ora,\widehat V^\ora\rangle=\ACCEPT]
          \le\epsilon,\\
      \Pr[\langle(P^*)^\ora,(V')^\ora\rangle=\ACCEPT]
        &\le\sum_{i=1}^\ell\Pr[\text{copy $i$ accepts}]
          \le\ell\epsilon<1/3.
    \end{aligned}
  \]

  \emph{Query and communication bounds.}
  If $\widehat V$ makes at most $Q$ queries, then $V'$ makes at most
  $\ell\cdot Q$, all after the last reply and with the same query
  lengths. The total communication and number of random bits are also
  polynomially bounded, independently of local running time.

  Finally, $P'$ reads $\ora$ and enumerates deterministic reply
  tables, using the first making $(V')^\ora$ accept on every public
  tape. If none exists, it outputs $\bot$ and terminates.
\end{proof}

%% file: appendices/random-set-containment.tex
\section{Containment in an Independent Random Set}
\label{app:random-set-containment}

In this section, we bound the probability that an independent random
set contains none of the sets in the support of a distribution $\pi$.
The hypothesis limits how often any fixed group of elements occurs
jointly in a sample from $\pi$.

\begin{lemma}[Containment in an independent marked set]
  \label{app:containment}
  Let $K\ge1$ and $0<\delta\le1/(8K)$.
  Let $\mathcal U$ be a finite universe and $\pi$ a distribution on
  subsets of $\mathcal U$, each of size at most $K$, such that
  \[
    \Pr_{T\leftarrow\pi}[U\subseteq T]\le2\delta^{2|U|}
  \]
  for every nonempty $U\subseteq\mathcal U$ with $|U|\le K$.
  Let $\mathsf{MARKED}\subseteq\mathcal U$ include each element
  independently with probability $\delta$. Then
  \[
    \Pr_{\mathsf{MARKED}}
      [\nexists T\in\operatorname{supp}(\pi):T\subseteq\mathsf{MARKED}]
      \le2^{-s},
    \qquad s=\left\lfloor\frac{1}{8K\delta}\right\rfloor\ge1.
  \]
\end{lemma}
\begin{proof}
  \emph{Second-moment bound.}
  Let $\mathsf M\subseteq\mathcal U$ include each element
  independently with probability $\delta_0\in(0,1)$, to be chosen
  below. Define
  \[
    Z(\mathsf M)=\sum_{T\subseteq\mathsf M}\pi(T)\delta_0^{-|T|}.
  \]
  Then $Z>0$ if and only if $\mathsf M$ contains some
  $T\in\operatorname{supp}(\pi)$, and $\E_{\mathsf M}Z=1$.
  For the variance, let
  $T,T'\leftarrow\pi$ be independent. Then
  \begin{align*}
    \operatorname{Var}(Z)
      &=\E_{T,T'}\delta_0^{-|T\cap T'|}-1\\
      &=\E_T\sum_{\varnothing\ne U\subseteq T}
          (\delta_0^{-1}-1)^{|U|}\Pr_{T'\leftarrow\pi}[U\subseteq T']\\
      &\le2\bigl((1+\delta^2(\delta_0^{-1}-1))^K-1\bigr)
       \le2\bigl(e^{K\delta^2/\delta_0}-1\bigr),
  \end{align*}
  where the inequalities use the hypothesis on
  $\Pr_{T'\leftarrow\pi}[U\subseteq T']$, $|T|\le K$, and
  $1+x\le e^x$. If $K\delta^2/\delta_0\le1/2$, Chebyshev's
  inequality and $e^x-1\le2x$ for $0\le x\le1/2$ imply
  \[
    \Pr_{\mathsf M}[Z=0]\le\operatorname{Var}(Z)
      \le4K\delta^2/\delta_0.
  \]

  \emph{Independent repetitions.}
  Set $\delta_0=1-(1-\delta)^{1/s}$ and sample independent sets
  $\mathsf M_1,\ldots,\mathsf M_s$ as above. Their union has the
  distribution of $\mathsf{MARKED}$, since each
  element is included independently with probability
  $1-(1-\delta_0)^s=\delta$. Moreover, $\delta\le s\delta_0$ by a
  union bound, so
  \[
    K\delta^2/\delta_0\le K\delta s\le1/8.
  \]
  Hence $\Pr_{\mathsf M_i}[Z(\mathsf M_i)=0]\le4K\delta s\le1/2$.
  If $\bigcup_{i=1}^s\mathsf M_i$ contains no member of
  $\operatorname{supp}(\pi)$, then $Z(\mathsf M_i)=0$ for every $i$.
  Independence therefore implies
  \[
    \begin{aligned}
      &\Pr_{\mathsf{MARKED}}
        [\nexists T\in\operatorname{supp}(\pi):T\subseteq\mathsf{MARKED}]\\
      &\qquad\le\prod_{i=1}^s\Pr_{\mathsf M_i}[Z(\mathsf M_i)=0]
        \le2^{-s}.
    \end{aligned}
  \]
\end{proof}